\documentclass{amsart}

\usepackage[usenames]{color}

\usepackage{amssymb, amsmath, amscd, amsthm, mathrsfs}
\usepackage{graphicx}
\usepackage{mathptmx}
\usepackage{psfrag}
\usepackage[all]{xy}
\usepackage{enumitem}

\newtheorem{theorem}{Theorem}[section]
\newtheorem{lem}[theorem]{Lemma}
\newtheorem{cor}[theorem]{Corollary}
\newtheorem{prop}[theorem]{Proposition}

\theoremstyle{definition}
\newtheorem{defi}[theorem]{Definition}

\newtheorem{exa}{\bf Example}
\theoremstyle{remark}

\newcommand{\F}{\mathcal F}
\newcommand{\R}{\mathbb R}
\newcommand{\N}{\mathbb N}

\newcommand{\p}{\varphi}

\newcommand{\h}{h}

\newcommand{\eps}{\varepsilon}
\newcommand{\M}{{\mathcal M}}

\newcommand{\lf}{\ell_f}
\newcommand{\Lg}{\ell_g}

\numberwithin{equation}{section}

\begin{document}

\title{The edit distance for Reeb graphs of surfaces}

\author{B. Di Fabio}
\address{Dipartimento
di Matematica, Universit\`a di Bologna, P.zza di Porta S. Donato
5, I-$40126$ Bologna, Italia\newline ARCES, Universit\`a di
Bologna, via Toffano $2/2$, I-$40125$ Bologna, Italia}
\email{barbara.difabio@unibo.it}

\author{C. Landi}
\address{Dipartimento
di Scienze e Metodi dell'Ingegneria, Universit\`a di Modena e
Reggio Emilia, Via Amendola 2, Pad. Morselli, I-42100 Reggio
Emilia, Italia\newline ARCES, Universit\`a di Bologna, via Toffano
$2/2$, I-$40125$ Bologna, Italia} \email{claudia.landi@unimore.it}


%
%


\begin{abstract}
Reeb graphs are structural descriptors that capture shape
properties of a topological space from the perspective of a chosen
function. In this work we define a combinatorial metric for Reeb
graphs of orientable surfaces in terms of the cost necessary to
transform one graph into another by edit operations. The main
contributions of this paper are the stability property and the
optimality of this edit distance. More precisely, the stability
result states that changes in the functions, measured by the
maximum norm, imply not greater changes in the corresponding Reeb
graphs, measured by the edit distance. The optimality result
states that our edit distance discriminates Reeb graphs better
than any other  metric for Reeb graphs of surfaces satisfying the stability
property.
\end{abstract}

\maketitle

\section*{Introduction}
In shape comparison, a widely used scheme is to measure the
dissimilarity between descriptors associated with each shape rather
than to match shapes directly. 
Reeb graphs describe shapes from topological and geometrical
perspectives. In this framework, a shape is modeled as a
topological space $X$ endowed with a scalar function $f:X\to\R$.
The role of $f$ is to explore geometrical properties of the space
$X$. The Reeb graph of  $f$ is obtained by shrinking each
connected component of a level set of $f$ to a single point
\cite{reeb46}. Reeb graphs have been used as an effective tool for
shape analysis and description tasks since \cite{ShKuKe91,Sh91}.

One of the most important questions is whether Reeb graphs are
robust against perturbations that may occur because of noise and
approximation errors in the data acquisition process. Whereas in
the past researchers dealt with this problem developing
heuristics so that Reeb graphs would be resistant to connectivity
changes caused by simplification, subdivision and remesh, and
robust against noise and certain changes due to deformation
\cite{HiSh*01,BiMa*06}, in the last years the question of Reeb
graph stability has been investigated from the theoretical point
of view. In \cite{DiLa12} an edit distance between Reeb graphs
of curves endowed with Morse functions is introduced and shown to
yield stability. Importantly, despite the combinatorial nature of
this distance, it coincides with the natural pseudo-distance
between shapes \cite{DoFr04}, thus showing the maximal
discriminative power for this sort of distances. Very recently a
functional distortion distance between Reeb graphs has been
proposed in  \cite{BaGeWa}, with proven  stable and discriminative
properties. The functional distortion distance  is based on
continuous maps between the topological spaces realizing the Reeb
graphs, so that it is not combinatorial in its definition.
Noticeably, it allows for comparison of  non-homeomorphic spaces
meaning that it can be used  to deal also with artifacts that
change the homotopy type of the space, although as a consequence it cannot  fully
discriminate shapes and stability is not
proven in that case.

In this paper we deal with the comparison problem for Reeb graphs
of surfaces. Indeed the case of surfaces seems to us  the most
interesting area of application  of the Reeb graph as a shape
descriptor. As a tradeoff between generality and simplicity, we
confine ourselves to the case of smooth compact orientable
surfaces without boundary endowed with simple Morse functions.

The basic properties we consider important for a metric between
Reeb graphs are: the robustness to perturbations of the input
functions; the ability to discriminate functions on the same
manifold; the deployment of the combinatorial nature of graphs.
For this reason, we apply to the case of surfaces the same
underlying ideas as used  in \cite{DiLa12} for curves. Starting from
Reeb graphs labeled on the vertices by the function values, the following steps are
carried out: first, a set of admissible edit operations is detected to
transform a labeled Reeb graph into another; then a
suitable cost is associated to each edit operation; finally, a
combinatorial dissimilarity measure between labeled Reeb graphs,
called an \emph{edit distance}, is defined in terms of the
least cost necessary to transform one graph into another by edit
operations. However, the passage from curves to surfaces is not
automatic since Reeb graphs of surfaces are structurally different
from those of curves. For example, the degree of vertices is
different for Reeb graphs of curves and surfaces. Therefore, the
set of edit operations as well as their costs cannot be directly
imported from the case of curves but need to be suitably defined.
In conclusion, our edit distance between Reeb graphs belongs to
the family of Graph Edit Distances \cite{Gao*10}, widely used in
pattern analysis.

Our first main result is that changes in the functions, measured
by the maximum norm, imply not greater changes in this edit
distance, yielding the stability property under function
perturbations. To prove this result, we track the changes in the
Reeb graphs as the function varies along a linear path avoiding
degeneracies. From the stability property, we deduce that the
edit distance between the Reeb graphs of two functions $f$ and
$g$ defined on a surface is a lower bound for the natural
pseudo-distance between $f$ and $g$ obtained by minimizing the
change in the functions due to the application of a
self-diffeomorphism of the manifold, with respect to the maximum
norm. The natural pseudo-distance can be thought as a way to
compare $f$ and $g$ directly, while the edit distance provides
an indirect comparison between $f$ and $g$ through their Reeb
graphs. Thus, by virtue of the stability result, the edit
distance provides a combinatorial tool to estimate the natural
pseudo-distance.

Our second contribution is the proof that the edit distance
between Reeb graphs of surfaces actually coincides with the
natural pseudo-distance. This is proved by showing that for every
edit operation on a Reeb graph there is a self-homeomorphism of
the surface whose  cost is not greater than that of
the considered edit operation. This result implies that the
edit distance is actually a metric and not only a
pseudo-metric. Morever it shows that the edit distance is an
optimal distance for Reeb graphs of surfaces in that it has the
maximum discriminative power among all the distances between Reeb
graphs of surfaces with the stability property. 

In conclusion, this paper shows that the results of \cite{DiLa12}
for curves also hold in the more interesting case of surfaces.
\\

\noindent The paper is organized as follows. In Section
\ref{labreebgraphs} we recall the basic properties of labeled Reeb
graphs of orientable surfaces. In Section
\ref{SectionDeformations} we define the edit deformations
between labeled Reeb graphs, and show that through a finite
sequence of these deformations we can always transform a Reeb
graph into another. In Section \ref{edit} we define the cost
associated with each type of edit deformation and the edit
distance in terms of this cost. Section \ref{stab} illustrates the
robustness of Reeb graphs with respect to the edit distance.
Eventually, Section \ref{lowbound} provides relationships between
our edit distance and other stable metrics: the natural
pseudo-distance, the bottleneck distance and the functional distortion distance.\\

A number of questions remain open and are not treated in this
paper. The most important one is how to  compute  the edit
distance. Indeed, whereas in some particular cases we can deduce
the value of the edit distance from  the lower bounds provided
by the bottleneck distance of persistence diagrams or the
functional distortion distance of Reeb graphs, in general we do
not know how to compute it. A second open problem is to which
extent the theory developed in this paper for the smooth category
can be transported into the piecewise linear category. A third
question that would deserve investigation is how to generalize the
edit distance to compare functions on non-homeomorphic surfaces
as well, and the relationship with the functional distortion
distance in that case.

\section{Labeled Reeb graphs of orientable surfaces}\label{labreebgraphs}

Hereafter, $\M$ denotes a connected, closed (i.e.
compact and without boundary), orientable, smooth surface of genus
$\mathfrak{g}$, and $\F(\M)$ the set of $C^{\infty}$ real functions on
$\M$.

For $f\in \F(\M)$, we denote by $K_f$ the set of its critical
points. If $p\in K_f$, then the real number $f(p)$ is called a
\emph{critical value} of $f$, and the set $\{q\in \M: q\in
f^{-1}(f(p))\}$ is called a \emph{critical level} of $f$.
Moreover, a critical point $p$ is called
\emph{non-degenerate} if the Hessian matrix of $f$ at $p$ is
non-singular. The \emph{index} of a non-degenerate critical point
$p$ of $f$ is the dimension of the largest subspace of the tangent
space to $\M$ at $p$ on which the Hessian matrix is negative definite. In
particular, the index of a point $p\in K_f$ is equal to 0,1, or 2
depending on whether $p$ is a minimum, a saddle, or a maximum
point of $f$.

A function $f\in\F(\M)$ is called a \emph{Morse function} if all its
critical points are non-degenerate. Besides, a Morse function is
said to be \emph{simple} if each critical level contains exactly
one critical point.  The set of simple Morse functions will be
denoted by $\F^0(\M)$, as a reminder that it is a sub-manifold of $\F(\M)$
of co-dimension 0 (see also Section \ref{stab}).

\begin{defi}
Let $f\in \F^0(\M)$, and define on $\M$ the following equivalence
relation: for every $p,q\in\M$, $p\sim_f q$ whenever $p,q$ belong to
the same connected component of $f^{-1}(f(p))$. The quotient space
$\M/\sim_f$ is the \emph{Reeb graph} associated with $f$.
\end{defi}

Our assumption that $f$ is a simple Morse function allows us to consider the space $\M/\sim_f$ as a graph whose  vertices  correspond bijectively to the critical points of $f$. For this reason, in the following, we will often
identify vertices with the corresponding critical points.

\begin{prop}(\cite{reeb46})
The Reeb graph $\Gamma_f$ associated with $f\in\F^0(\M)$ is a
finite and connected simplicial complex of dimension 1. A vertex of $\Gamma_f$ has degree equal to 1 if it corresponds to a critical point
of $f$ of index 0 or 2, while it has degree equal to 2,3, or 4 if it corresponds to a
critical point of $f$ of index 1.
\end{prop}

Throughout this paper, Reeb graphs are regarded as combinatorial
graphs and not as topological spaces. The vertex set of $\Gamma_f$ will be denoted by
$V(\Gamma_f)$, and its edge set by $E(\Gamma_f)$. Moreover, if
$v_1,v_2\in V(\Gamma_f)$ are adjacent vertices, i.e., connected by
an edge, we will write $e(v_1,v_2)\in E(\Gamma_{f})$.

Our assumptions that $\M$ is orientable, compact and without
boundary ensure that there are no vertices of degree 2 or 4.
Moreover, if $p,q,r$ denote the number of minima, maxima, and saddle points of $f$,
from the relationships between the Euler characteristic of $\M$,
$\chi(\M)$, and $p,q,r$, i.e. $\chi(\M)=p+q-r$, and between
$\chi(\M)$ and the genus $\mathfrak{g}$ of $\M$, i.e.
$\chi(\M)=2-2\mathfrak{g}$, it follows that the cardinality of
$V(\Gamma_f)$, which is $p+q+r$, is also equal to
$2(p+q+\mathfrak{g}-1)$, i.e. is even in number. The minimum
number of vertices of a Reeb graph is achieved whenever $p=q=1$,
and consequently $r=2\mathfrak{g}$. In this case the cardinality
of $V(\Gamma_f)$ is equal to $2\mathfrak{g}+2$.  In general, if $\M$ has genus $\mathfrak{g}$ then
$\Gamma_{f}$ has exactly $\mathfrak{g}$ linearly independent
cycles. We will call a cycle of length $m$ in the graph  an
$m$-\emph{cycle}. These observations motivate the following definition.

\begin{defi}\label{minimal-canonical}
We shall call \emph{minimal} the Reeb graph $\Gamma_{f}$ of a
function $f$ having $p=q=1$. Moreover, we say that $\Gamma_{f}$ is
\emph{canonical} if it is minimal and all its cycles, if any, are
2-cycles.
\end{defi}

We  underline that our definition of canonical Reeb graph
is slightly different from the one in \cite{Ku99}. This choice has
been done to simplify the proof of Proposition \ref{connected}.

Examples of minimal and canonical Reeb graphs are displayed in
Figure \ref{mini-cano}. 
\begin{center}
\begin{figure}[htbp]
\psfrag{S}{}\psfrag{G}{$\Gamma_f$}\psfrag{L}{$(\Gamma_f,
\lf)$}\psfrag{v1}{$v_1$}\psfrag{v2}{$v_2$}\psfrag{v3}{$v_3$}\psfrag{v4}{$v_4$}\psfrag{v5}{$v_5$}\psfrag{v6}{$v_6$}\psfrag{v7}{$v_7$}\psfrag{v8}{$v_8$}\psfrag{f}{$f$}
\psfrag{a}{$(a)$}\psfrag{b}{$(b)$}\psfrag{c}{$(c)$}
\includegraphics[height=4cm]{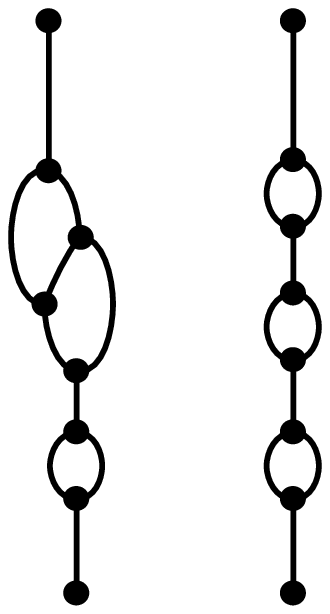}
\caption{\footnotesize{Examples of minimal Reeb graphs. The graph
on the right is also canonical.}}\label{mini-cano}
\end{figure}
\end{center}

In what follows, we label the vertices of $\Gamma_f$ by equipping
each of them with the value of $f$ at the corresponding critical
point. We denote such a labeled graph by $(\Gamma_f, \lf)$, where
$\lf:V(\Gamma_f)\to\R$ is the restriction of $f:\M\to\R$ to $K_f$.
In a labeled Reeb graph, each vertex $v$ of degree 3 has at least
two of its adjacent vertices, say $v_1,v_2$, such that
$\lf(v_1)<\lf(v)<\lf(v_2)$. An example is displayed in Figure
\ref{labeledReeb}.

\begin{figure}[htbp]
\psfrag{M}{$\M$}\psfrag{L}{$(\Gamma_f,
\lf)$}\psfrag{v1}{$a_1$}\psfrag{v2}{$a_2$}\psfrag{v3}{$a_3$}\psfrag{v4}{$a_4$}\psfrag{v5}{$a_5$}\psfrag{v6}{$a_6$}\psfrag{v7}{$a_7$}\psfrag{v8}{$a_8$}\psfrag{v9}{$a_9$}\psfrag{v10}{$a_{10}$}
\psfrag{f}{$f$}\psfrag{a1}{$a_1$}\psfrag{a2}{$a_2$}\psfrag{a3}{$a_3$}\psfrag{a4}{$a_4$}
\psfrag{a5}{$a_5$}\psfrag{a6}{$a_6$}\psfrag{a7}{$a_7$}\psfrag{a8}{$a_8$}\psfrag{a9}{$a_9$}\psfrag{a10}{$a_{10}$}
\psfrag{a}{$(a)$}\psfrag{b}{$(b)$}\psfrag{c}{$(c)$}
\includegraphics[height=4cm]{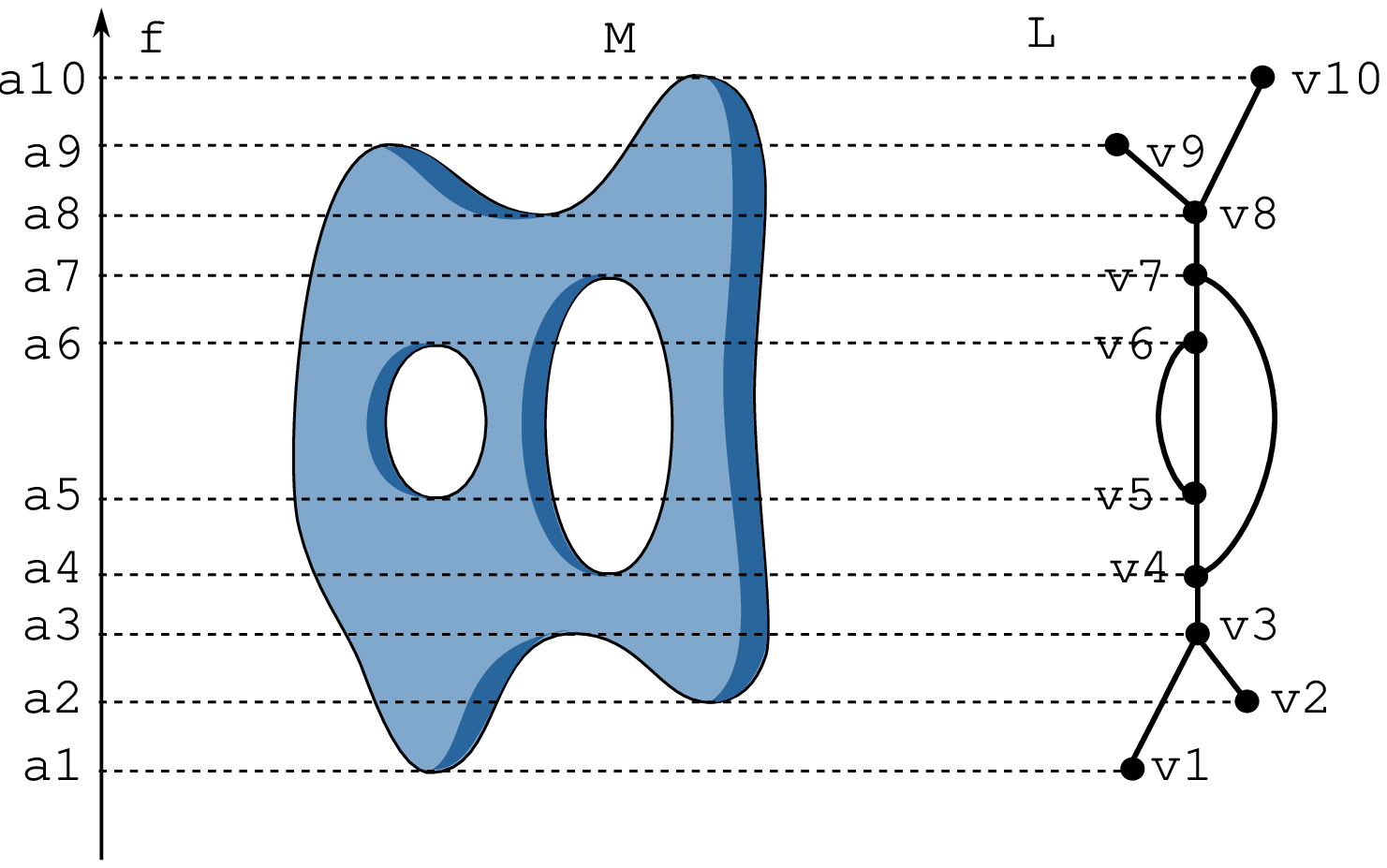}
\caption{\footnotesize{Left: the height function $f:\M\to\R$;
center: the surface $\M$ of genus $\mathfrak{g}=2$; right: the
associated labeled Reeb graph $(\Gamma_f, \lf)$.}}\label{labeledReeb}
\end{figure}

Let us consider the realization problem, i.e. the problem of
constructing a smooth surface and a simple Morse function on it from a graph on an even number
of vertices, all of which are of degree 1 or 3, appropriately
labeled. We need the following definition.

\begin{defi}\label{isolabel}
We shall say that two labeled Reeb graphs $(\Gamma_{f}, \lf),
(\Gamma_{g}, \Lg)$ are \emph{isomorphic}, and we write
$(\Gamma_{f}, \lf)\cong (\Gamma_{g}, \Lg)$, if there exists a
graph isomorphism $\Phi: V(\Gamma_{f}) \rightarrow V(\Gamma_{g})$
such that, for every $v\in V(\Gamma_f)$, $\lf(v) = \Lg(\Phi(v))$ (i.e.
$\Phi$ preserves edges and vertex labels).
\end{defi}

\begin{prop}[Realization Theorem]\label{realizth}
Let $(\Gamma,\ell)$ be a labeled graph, where $\Gamma$ is a graph
with $m$ linearly independent cycles, on an even number of
vertices, all of which are of degree 1 or 3, and
$\ell:V(\Gamma)\rightarrow \R$ is an injective function such that,
for any vertex $v$ of degree 3, at least two among its adjacent
vertices, say $w,w'$, are such that $\ell(w)<\ell(v)<\ell(w')$.
Then an orientable closed surface $\M$ of genus $\mathfrak{g}=m$,
and a simple Morse function $f:\M\to \R$ exist such that
$(\Gamma_{f}, \lf)\cong (\Gamma,\ell)$.
\end{prop}

\begin{proof}
Under our assumption on the degree of vertices of $\Gamma$, $\M$
and $f$ can be constructed as in the proof of Thm.~2.1 in
\cite{MaSa11}.
\end{proof}

We now deal with the uniqueness problem, up to isomorphism of
labeled Reeb graphs. First of all we consider the following two
equivalence  relations on $\F^0(\M)$.

\begin{defi}\label{equivalence}
Let $\mathcal D(\M)$ be the set of self-diffeomorphisms of $\M$.
Two functions $f,g\in\F^0(\M)$ are called \emph{right-equivalent}
(briefly, \emph{$R$-equivalent}) if there exists $\xi\in\mathcal
D(\M)$ such that $f=g\circ\xi$. Moreover, $f,g$ are called
\emph{right-left equivalent} (briefly, \emph{RL-equivalent}) if
there exist $\xi\in\mathcal D(\M)$ and an orientation preserving
self-diffeomorphism $\eta$ of $\R$ such that $ f=\eta\circ
g\circ\xi$.
\end{defi}

These equivalence relations on functions are mirrored by Reeb graphs
isomorphisms. 

\begin{prop}[Uniqueness Theorem]\label{unique}
If $f,g\in\F^0(\M)$, then
\begin{enumerate}
\item $f$ and $g$ are RL-equivalent if and only if their Reeb
graphs $\Gamma_{f}$ and $\Gamma_{g}$ are isomorphic by an
isomorphism $\Phi:V(\Gamma_{f})\to V(\Gamma_{g})$ that preserves
the vertex order, i.e., for every $v,w\in V(\Gamma_f)$,
$\lf(v)<\lf(w)$ if and only if $\Lg(\Phi(v))< \Lg(\Phi(w))$; \item $f$ and
$g$ are $R$-equivalent
 if and only if their labeled
Reeb graphs $(\Gamma_{f},\lf)$ and  $(\Gamma_{g},\Lg)$ are
isomorphic.
 \end{enumerate}
\end{prop}

\begin{proof} For the proof of statement (1) see \cite{Ku98,Sh03}. As for statement (2), we note that two $R$-equivalent functions
are, in particular, $RL$-equivalent. Hence, by statement (1),
their Reeb graphs are isomorphic by an isomorphism that preserves
the vertex order. Since $f$ and $g$ necessarily have the same
critical values, this isomorphism also preserves labels.
Vice-versa, if $(\Gamma_{f},\lf)$ and  $(\Gamma_{g},\Lg)$ are
isomorphic, then $f$ and $g$ have the same critical values.
Moreover, by statement (1), there exist $\xi\in\mathcal D(\M)$ and
an orientation preserving self-diffeomorphism $\eta$ of $\R$ such
that $ f=\eta\circ g\circ\xi$. Let us set $h=g\circ\xi$. The
function $h$ belongs to $\F^0(\M)$ and has the same critical
points with the same indexes as $f$, and the same critical
values as $g$ and hence as $f$. Thus, we can apply \cite[Lemma 1]{Ku09} to $f$ and
$h$ and deduce the existence of a self-diffeomorphism $\xi'$ of
$\M$ such that $f=h\circ \xi'$. Thus $f=g\circ \xi\circ \xi'$,
yielding that $f$ and $g$ are $R$-equivalent. A direct proof of the $R$-equivalence of functions with isomorphic labeled Reeb graphs is also obtainable from Lemma~\ref{c(T)=R}.
\end{proof}

\section{edit deformations between labeled Reeb
graphs}\label{SectionDeformations}

In this section we list the edit deformations admissible to
transform labeled Reeb graphs into one another when different simple Morse
functions on the same surface are considered. We introduce at
first elementary deformations, then, by virtue of the Realization
Theorem (Proposition \ref{realizth}), the deformations obtained by
their composition.

Elementary deformations allow us to  insert or delete a vertex of degree 1 together with an adjacent vertex of degree 3   (deformations of \emph{birth} type (B)   and \emph{death} type (D)),  maintain  the same vertices and edges while changing the vertex
labels (deformations of \emph{relabeling} type (R)), or  change some vertices adjacencies and labels (deformations of type (K$_1$), (K$_2$), (K$_3$) introduced by  Kudryavtseva in   \cite{Ku99}). A sketch of these elementary deformations can be
found in Table \ref{deformations}. The formal definition is as follows.

\begin{defi}\label{elementaryDef}
With the convention of denoting the open interval with endpoints
$a,b\in\R$ by $]a,b[$, the elementary deformations of type  \rm{(B), (D), (R),
(K$_i$)}, $i=1,2,3$, are defined as follows.

\begin{itemize}[leftmargin=0.7cm]
\item [\mbox{\rm{(B)}}] $T$ is an
\emph{elementary deformation} of $(\Gamma_{f},\lf)$ \emph{of type}
\rm{(B)} if, for a fixed edge $e(v_1, v_2)\in
E(\Gamma_{f})$, with $\lf(v_{1})<\lf(v_{2})$,  $T(\Gamma_{f},\lf)$ is a labeled graph $(\Gamma,\ell)$
such that
\begin{itemize}
\item $V(\Gamma)=V(\Gamma_{f})\cup \{u_{1},u_{2}\}$; \item
$E(\Gamma)= \left(E(\Gamma_{f})-\{e(v_1, v_2)\}\right)\cup\{e(v_1,
u_1), e(u_1, u_2), e(u_1, v_2)\}$; \item ${\ell}_{|V(\Gamma_{f})}= \lf$ and 
$\lf(v_1)<\ell(u_i)<\ell(u_j)<\lf(v_2)$, with
$\ell^{-1}(]\ell(u_i),\ell(u_j)[)=\emptyset$, $i,j\in\{1,2\},i\neq
j$.
\end{itemize}

\item [\mbox{\rm{(D)}}] $T$ is an \emph{elementary deformation} of $(\Gamma_{f},\lf)$
\emph{of type} \rm{(D)} if, for fixed edges $e(v_1, u_1)$, $e(u_1, u_2)$,
$e(u_1, v_2)\in E(\Gamma_{f})$, $u_2$ being of degree 1, such that
$\lf(v_1)<\lf(u_i)<\lf(u_j)<\lf(v_2)$, with
$\lf^{-1}(]\lf(u_i),\lf(u_j)[)=\emptyset$, $i,j\in\{1,2\},i\neq
j$,  $T(\Gamma_{f},\lf)$ is a labeled graph
$(\Gamma,\ell)$ such that
\begin{itemize}
\item $V(\Gamma)=V(\Gamma_{f})-\{u_{1},u_{2}\}$; \item $E(\Gamma)=
\left(E(\Gamma_{f})-\{e(v_1, u_1), e(u_1, u_2), e(u_1,
v_2)\}\right)\cup\{e(v_1, v_2)\}$; \item
$\ell={\lf}_{|V(\Gamma_{f})-\{u_1,u_2\}}$.
\end{itemize}

\item[\mbox{\rm{(R)}}] $T$ is an \emph{elementary deformation} of
$(\Gamma_{f},\lf)$ \emph{of type} \rm{(R)} if $T(\Gamma_{f},\lf)$
is a labeled graph $(\Gamma,\ell)$ such that
\begin{itemize} \item $\Gamma=\Gamma_f$; \item $\ell:V(\Gamma)\to\R$ induces the same vertex-order as $\lf$ except for
at most two non-adjacent vertices, say $u_1,u_2$, with $\lf(u_1)<\lf(u_2)$ and
$\lf^{-1}(]\lf(u_1),\lf(u_2)[)=\emptyset$, for which
$\ell(u_1)>\ell(u_2)$ and
$\ell^{-1}(]\ell(u_2),\ell(u_1)[)=\emptyset$.
\end{itemize}

\item[\mbox{\rm{(K$_1$)}}] $T$ is an
\emph{elementary deformation} of $(\Gamma_{f},\lf)$ \emph{of type}
\rm{(K$_1$)} if, for fixed edges $e(v_1, u_1)$, $e(u_1,
u_2)$, $e(u_1, v_4)$, $e(u_2, v_2)$, $e(u_2, v_3) \in E(\Gamma_{f})$,
with two among $v_2, v_3, v_4$ possibly coincident,
$\lf(v_1)<\lf(u_1)<\lf(u_2)<\lf(v_2),\lf(v_3)$, $\lf(v_4)$, and
$\lf^{-1}(]\lf(u_1),\lf(u_2)[)=\emptyset$ (resp. $\lf(v_2)$,
$\lf(v_3)$, $\lf(v_4)<\lf(u_2)<\lf(u_1)$ $<\lf(v_1)$, and
$\lf^{-1}(]\lf(u_2),\lf(u_1)[)=\emptyset$),  $T(\Gamma_{f},\lf)$ is a labeled graph
$(\Gamma,\ell)$ such that:
\begin{itemize} \item
$V(\Gamma)=V(\Gamma_{f})$; \item $E(\Gamma)=
\left(E(\Gamma_{f})-\{e(v_1, u_1), e(u_2,
v_2)\}\right)\cup\{e(v_1,u_2),e(u_1,v_2)\}$; \item
${\ell}_{|V(\Gamma)-\{u_1,u_2\}}=\lf$ and
$\lf(v_1)<\ell(u_2)<\ell(u_1)<\lf(v_2)$, $\lf(v_3)$, $\lf(v_4)$,
with $\ell^{-1}(]\ell(u_2),\ell(u_1)[)$ $=\emptyset$ (resp.
$\lf(v_2)$, $\lf(v_3)$, $\lf(v_4)<\ell(u_1)<\ell(u_2)<\lf(v_1)$,
with $\ell^{-1}(]\ell(u_1),\ell(u_2)[)=\emptyset$).
\end{itemize}

\item[\mbox{\rm{(K$_2$)}}] $T$ is an
\emph{elementary deformation} of $(\Gamma_{f},\lf)$ \emph{of type}
\rm{(K$_2$)} if, for fixed edges $e(v_1, u_1)$, $e(v_2,
u_1)$, $e(u_1, u_2)$,  $e(u_2, v_3)$, $e(u_2, v_4) \in E(\Gamma_{f})$,
with $u_1,u_2$ of degree 3, $v_2,v_3$ possibly coincident with
$v_1,v_4$, respectively, and $\lf(v_1),\lf(v_2)$
$<\lf(u_1)<\lf(u_2)<\lf(v_3),\lf(v_4)$, with
$\lf^{-1}(]\lf(u_1),\lf(u_2)[)=\emptyset$,  $T(\Gamma_{f},\lf)$ is a labeled graph
$(\Gamma,\ell)$ such that:
\begin{itemize} \item
$V(\Gamma)=V(\Gamma_{f})$; \item $E(\Gamma)=
\left(E(\Gamma_{f})-\{e(v_1, u_1), e(u_2,
v_3)\}\right)\cup\{e(v_1,u_2), e(u_1,v_3)\}$; \item
${\ell}_{|V(\Gamma)-\{u_1,u_2\}}=\lf$ and $\lf(v_1)$,
$\lf(v_2)<\ell(u_2)<\ell(u_1)<\lf(v_3)$, $\lf(v_4)$, with
$\ell^{-1}(]\ell(u_2),\ell(u_1)[)=\emptyset$.
\end{itemize}

\item[\mbox{\rm{(K$_3$)}}] $T$ is an
\emph{elementary deformation} of $(\Gamma_{f},\lf)$ \emph{of type}
\rm{(K$_3$)} if, for fixed edges $e(v_1, u_2)$, $e(u_1,
u_2)$, $e(v_2, u_1)$, $e(u_1, v_3)$, $e(u_2, v_4) \in E(\Gamma_{f})$,
with $u_1,u_2$ of degree 3, $v_2,v_3$ possibly coincident with
$v_1,v_4$, respectively, and
$\lf(v_1),\lf(v_2)<\lf(u_2)<\lf(u_1)<\lf(v_3),\lf(v_4)$, with
$\lf^{-1}(]\lf(u_2),\lf(u_1)[)=\emptyset$,  $T(\Gamma_{f},\lf)$ is a labeled graph
$(\Gamma,\ell)$ such that:
\begin{itemize} \item
$V(\Gamma)=V(\Gamma_{f})$; \item $E(\Gamma)=
\left(E(\Gamma_{f})-\{e(v_1, u_2), e(u_1,
v_3)\}\right)\cup\{e(v_1,u_1), e(u_2,v_3)\}$; \item
${\ell}_{|V(\Gamma)-\{u_1,u_2\}}=\lf$ and $\lf(v_1)$,
$\lf(v_2)<\ell(u_1)<\ell(u_2)<\lf(v_3)$, $\lf(v_4)$, with
$\ell^{-1}(]\ell(u_1),\ell(u_2)[)=\emptyset$.
\end{itemize}
\end{itemize}
\end{defi}

\begin{table}[htbp]
\begin{center}
\psfrag{a1}{$\lf(v_1)$}\psfrag{b1}{$\ell(u_1)$}\psfrag{b2}{$\ell(u_2)$}\psfrag{a2}{$\lf(v_2)$}
\psfrag{c1}{$\lf(u_1)$}\psfrag{c2}{$\lf(u_2)$}
\psfrag{a3}{$\lf(v_3)$}\psfrag{b3}{$\ell(v_3)$}\psfrag{b4}{$\ell(v_4)$}\psfrag{a4}{$\lf(v_4)$}
\psfrag{d1}{$\ell(v_1)$}\psfrag{d2}{$\ell(v_2)$}
\psfrag{a5}{$\lf(v_5)$}\psfrag{b5}{$\ell(v_5)$}\psfrag{b6}{$\ell(v_6)$}\psfrag{a6}{$\lf(v_6)$}
\psfrag{a7}{$\lf(v_7)$}\psfrag{b7}{$\ell(v_7)$}\psfrag{b8}{$\ell(v_8)$}\psfrag{a8}{$\lf(v_8)$}
\psfrag{B}{(B)}\psfrag{D}{(D)}\psfrag{R}{(R)}\psfrag{H}{(R$_2$)}\psfrag{K1}{(K$_1$)}\psfrag{K4}{(K$_3$)}
\psfrag{K3}{(K$_2$)}\psfrag{K2}{(K$_2$)}\psfrag{K'2}{(K$_3$)}
\begin{tabular}{c}
\includegraphics[width=\linewidth]{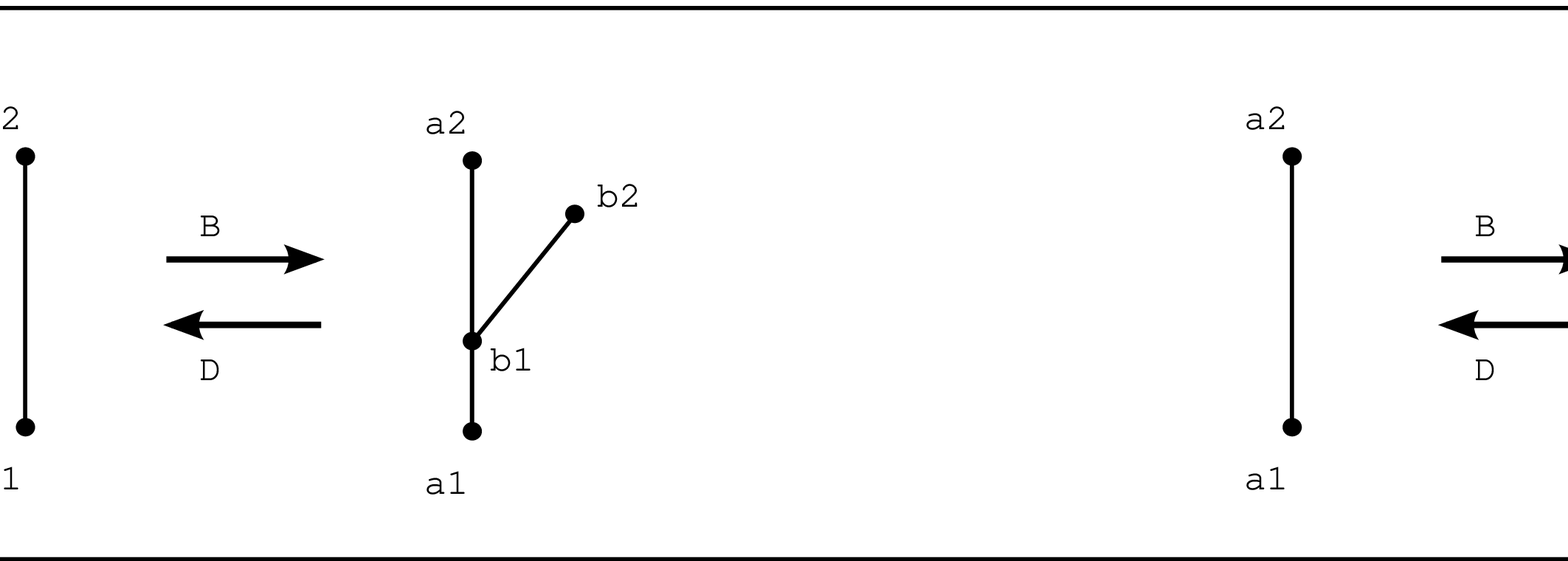}\\
\includegraphics[width=\linewidth]{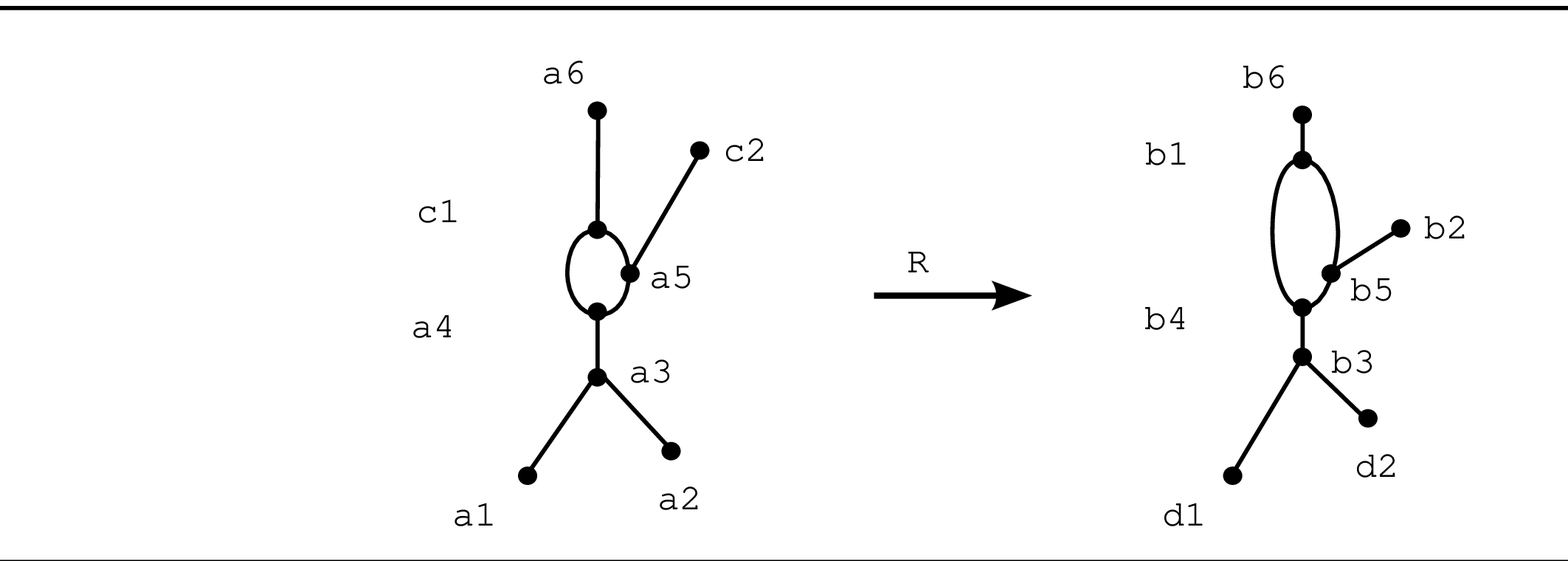}\\
\includegraphics[width=\linewidth]{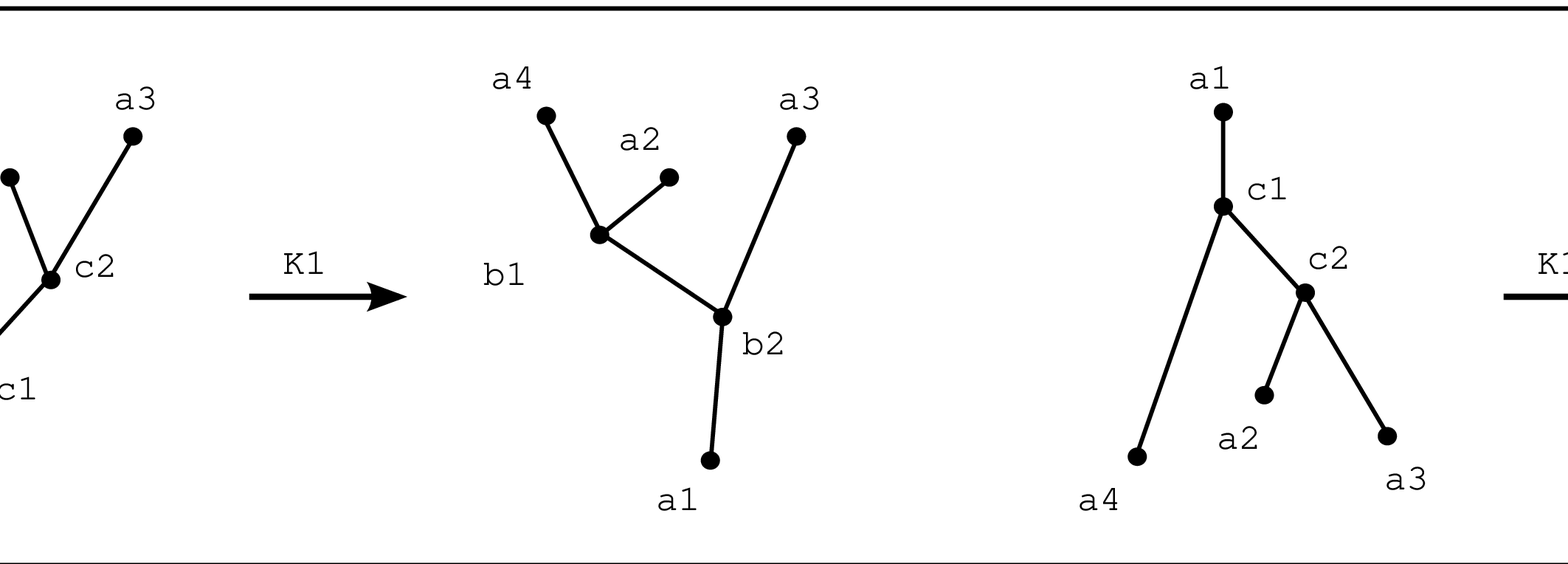}\\
\includegraphics[width=\linewidth]{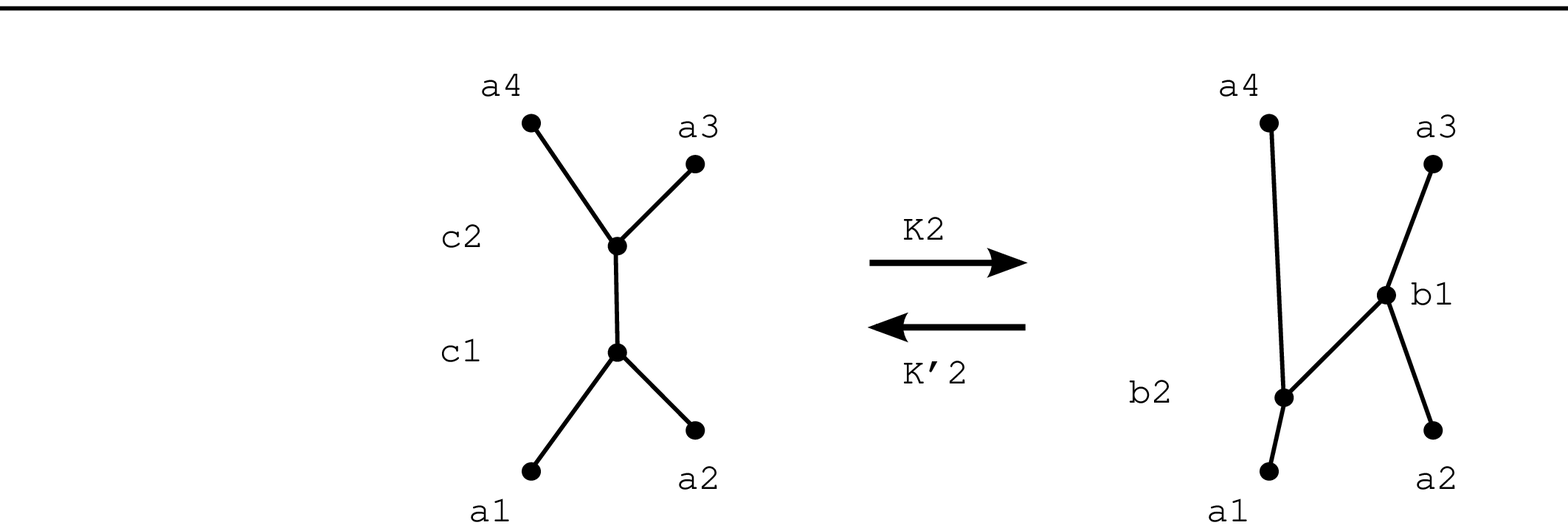}\\
\end{tabular}
\caption{\footnotesize{A schematization of the elementary deformations of a labeled Reeb
graph provided by Definition \ref{elementaryDef}.}}\label{deformations}
\end{center}
\end{table}

We underline that the definition of the deformations of type (B),
(D) and (R) is essentially different from the definition of
analogous deformations in the case of Reeb graphs of curves as
given in \cite{DiLa12}, even if the associated cost will be the
same (see Section \ref{edit}). This is because the degree of the
involved vertices is 2 for Reeb graphs of closed curves, whereas
it is 1 or 3 for Reeb graphs of surfaces.

\begin{prop}\label{defGl}
Let $f\in\F^0(\M)$ and $(\Gamma,\ell)=T(\Gamma_f,\lf)$ for some elementary deformation $T$. Then
there exists $g\in\F^0(\M)$ such that
$(\Gamma_g,\Lg)\cong(\Gamma,\ell)$.
\end{prop}

\begin{proof}
The claim follows from Proposition \ref{realizth}.
\end{proof}

As a consequence of Proposition \ref{defGl}, we can apply
elementary deformations iteratively. This fact is used in the next
Definition \ref{def}. Given an elementary deformation $T$ of $(\Gamma_{f},\lf)$ and an
elementary deformation $S$ of $T(\Gamma_{f},\lf)$, the
juxtaposition $ST$ means applying first $T$ and then  $S$.

\begin{defi}\label{def}
We shall call {\em deformation} of $(\Gamma_{f},\lf)$ any finite
ordered sequence $T=(T_1,T_2,\ldots ,T_r)$ of elementary
deformations such that $T_1$ is an elementary deformation of
$(\Gamma_{f},\lf)$, $T_2$ is an elementary deformation of
$T_1(\Gamma_{f},\lf)$, ..., $T_r$ is an elementary deformation of
$T_{r-1}T_{r-2}\cdots T_1(\Gamma_{f},\lf)$. We shall denote by
$T(\Gamma_{f},\lf)$ the result of the deformation
$T_rT_{r-1}\cdots T_1$ applied to $(\Gamma_{f},\lf)$.
\end{defi}

In the rest of the paper we write
${\mathcal T}((\Gamma_{f},\lf), (\Gamma_{g},\Lg))$ to denote the set of deformations turning $(\Gamma_{f},\lf)$ into $(\Gamma_{g},\Lg)$ up to isomorphism:
$${\mathcal T}((\Gamma_{f},\lf), (\Gamma_{g},\Lg))=\{T=(T_1,\ldots,T_n), n\ge 1: T(\Gamma_{f},\lf)\cong(\Gamma_{g},\Lg)\}.$$
We now introduce the concept of inverse deformation.

\begin{defi}\label{definverse}
Let $T\in{\mathcal T}((\Gamma_{f},\lf), (\Gamma_{g},\Lg))$ and  let $\Phi$ be the labeled graph isomorphism between $T(\Gamma_{f},\lf)$ and $(\Gamma_{g},\Lg)$.  We denote  by $T^{-1}$, and call it the \emph{inverse} of $T$ in ${\mathcal T}((\Gamma_{g},\Lg), (\Gamma_{f},\lf))$, the deformation  that acts on the vertices, edges, and labels of  $(\Gamma_{g},\Lg)$ as follows:  identifying    $T(\Gamma_{f},\lf)$ with $(\Gamma_{g},\Lg)$ via $\Phi$, 

\begin{itemize}
\item if $T$ is an elementary deformation of type (D)
deleting two vertices, then $T^{-1}$ is of type (B) inserting the
same vertices, with the same labels, and viceversa;

\item if $T$ is  an elementary deformation of type (R) relabeling vertices of
$V(\Gamma_{f})$, then $T^{-1}$ is again of type (R) relabeling
these vertices in the inverse way;

\item if $T$ is  an elementary deformation of type (K$_1$) relabeling two
vertices, then $T^{-1}$ is again of type (K$_1$) relabeling the
same vertices in the inverse way;

\item if $T$ is  an elementary deformation of type (K$_2$) relabeling two
vertices, then $T^{-1}$ is of type (K$_3$) relabeling the same
vertices in the inverse way, and viceversa;

\item if $T=(T_1, \ldots, T_r)$, then $T^{-1}=(T^{-1}_r, \ldots,
T^{-1}_1)$.
\end{itemize}
\end{defi}

%

From the fact that $T^{-1}T(\Gamma_{f},\lf)\cong (\Gamma_{f},\lf)$ it follows that the set ${\mathcal
T}((\Gamma_{f},\lf), (\Gamma_{g},\Lg))$, when non-empty,  always contains infinitely many deformations.
We end the section proving that for $f,g\in \F^0(\M)$ it  is always non-empty. 
We first need two lemmas which are widely inspired by \cite[Lemma 1 and Thm.
1]{Ku99}, respectively.

\begin{lem}\label{pathLemma}
Let $(\Gamma_{f},\lf)$ be a labeled Reeb graph. The following
statements hold:
\begin{itemize}
\item[$(i)$] For any $u,v\in V(\Gamma_f)$ corresponding to two
minima or two maxima of $f$, there exists a deformation $T$ such
that $u$ and $v$ are adjacent to the same vertex $w$ in
$T(\Gamma_{f},\lf)$. \item[$(ii)$] For any $m$-cycle $C$ in
$\Gamma_{f}$, $m\ge 2$, there exists a deformation $T$ such that
$C$ is a 2-cycle in $T(\Gamma_{f},\lf)$.
\end{itemize}
\end{lem}

\begin{proof}
Let us prove statement $(i)$ assuming that in $(\Gamma_{f},\lf)$
there exist two vertices $u,v$ corresponding to two minima of $f$.
The  case of maxima is analogous.

Let us consider a path $\gamma$ on $\Gamma_{f}$ having $u,v$ as
endpoints, whose length is $m\ge 2$, and the finite sequence of
vertices through which it passes is $(w_0,w_1,\ldots,w_m)$, with
$w_0= u, w_m= v$, and $w_i\neq w_j$ for $i\neq j$. We aim at showing
that there exists a deformation $T$ such that in
$T(\Gamma_{f},\lf)$ the vertices $u,v$ are adjacent to the same
vertex $w$, with $w\in\{w_1,\dots,w_{m-1}\}$, and thus the path
$\gamma$ is transformed by $T$ into a path $\gamma'$ which is of length 2
and passes through the vertices $u,w,v$.

If $m=2$, then it is sufficient to take $T$ as the deformation of
type (R) such that $T(\Gamma_{f},\lf)=(\Gamma_{f},\lf)$ since
$\gamma$ already coincides with $\gamma'$. If $m>2$, let $w_i=\mathrm{argmax} \{\lf(w_j): w_j \mbox{ with } 0\le j\le  m\}$.  It holds that $w_i\neq
u,v$  because $u,v$ are minima of $f$ and is unique because $f$ is simple. It is easy to observe
that, in a neighborhood of $w_i$, possibly after a finite sequence
of deformations of type (R), the graph  gets one of the
configurations shown in Figure~\ref{pathReduction}~$(a)-(e)$
(left).

\begin{figure}[htbp]
\begin{center}
\psfrag{wi}{$w_i$}\psfrag{w+1}{$w_{i+1}$}\psfrag{w-1}{$w_{i-1}$}\psfrag{w-2}{$w_{i-2}$}
\psfrag{w-3}{$w_{i-3}$}\psfrag{w}{$w$}\psfrag{K1}{(K$_1$)}\psfrag{K5}{(K$_3$)}
\begin{tabular}{cc}
\includegraphics[height=2.8cm]{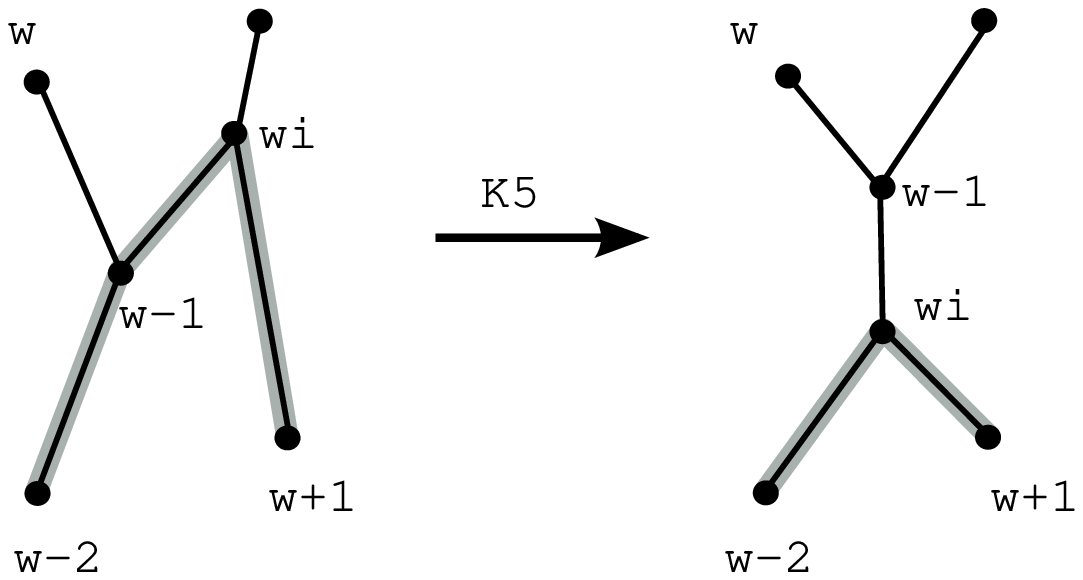}&\includegraphics[height=2.8cm]{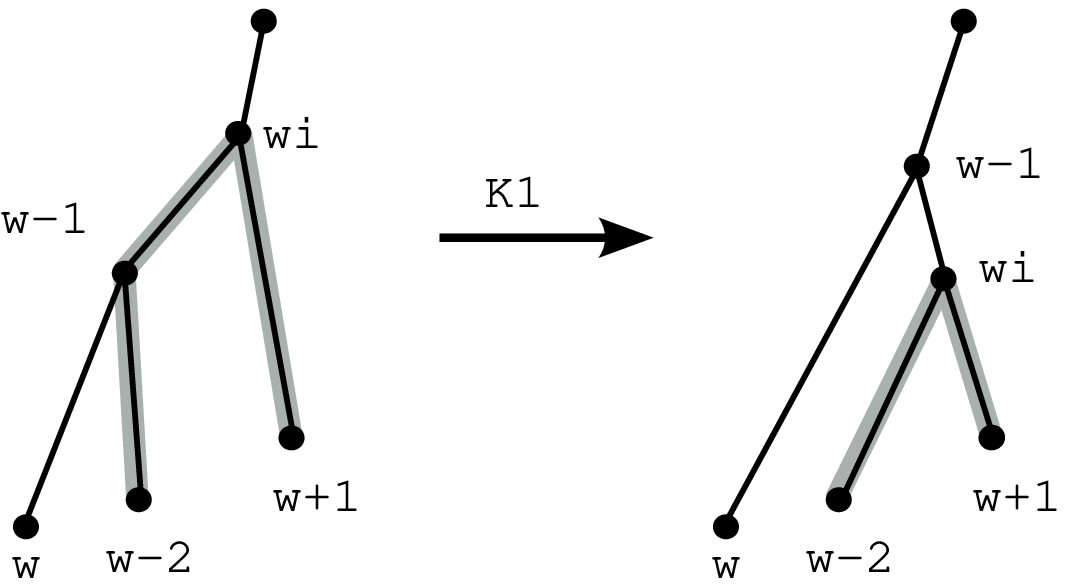}\\
$(a)$&$(b)$\\
\end{tabular}
\begin{tabular}{c}
\includegraphics[height=2.8cm]{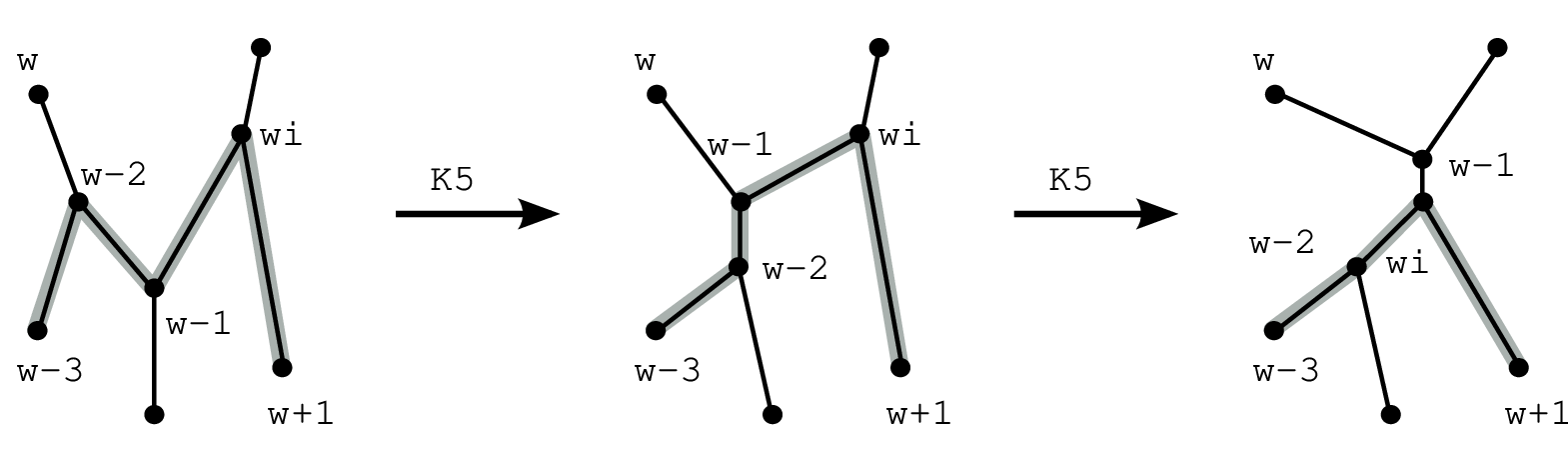}\\
$(c)$\\
\includegraphics[height=2.8cm]{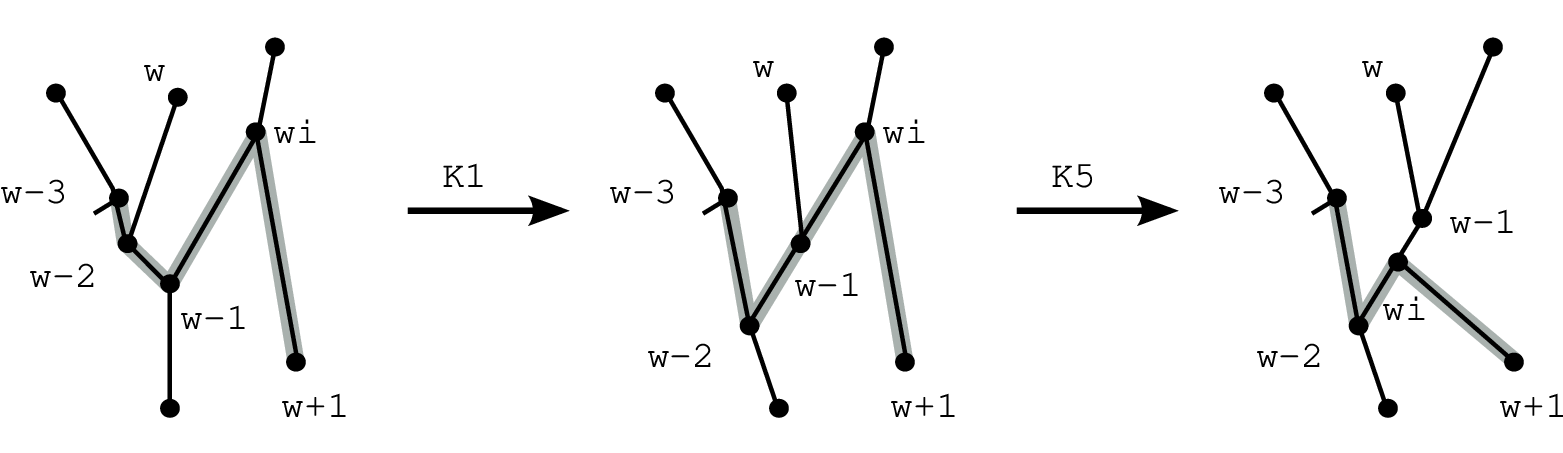}\\
$(d)$\\
\includegraphics[height=2.8cm]{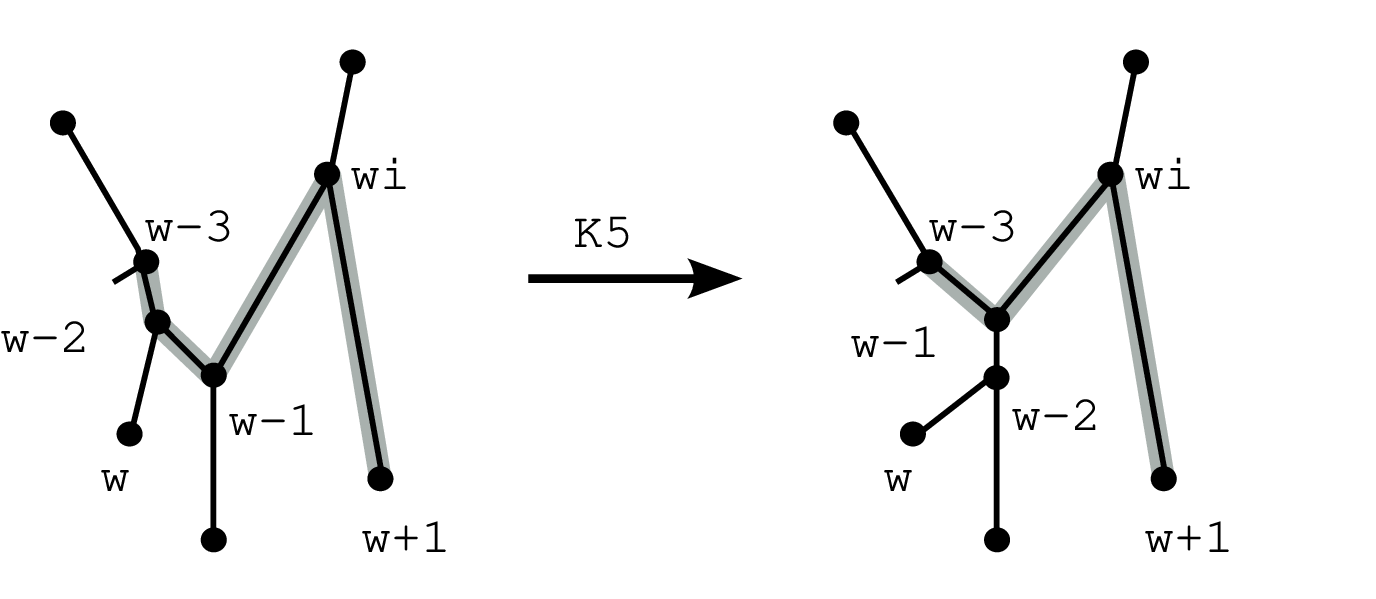}\\
$(e)$\\
\end{tabular}
\caption{\footnotesize{Possible configurations of a simple path on
a labeled Reeb graph in a neighborhood of its maximum point, and
elementary deformations which reduce its length. To facilitate the
reader, $f$ has been represented as the height function, so that
$\lf(w_a)<\lf(w_b)$ if and only if $w_a$ is lower than $w_b$ in
the pictures.}}\label{pathReduction}
\end{center}
\end{figure}

The same figure shows that a finite sequence of deformations of
type (K$_1$), (K$_3$), and, possibly, (R) transforms the simple path $\gamma$, which has length $m$, into a simple path of length $m-1$. Iterating this
procedure, we deduce the desired claim.

The proof of statement $(ii)$ is analogous to that of statement
$(i)$, provided that $\gamma$ is taken to be an $m$-cycle with
$u\equiv v$ of degree 3, and
$u=\mathrm{argmin}\{\ell_f(w_j):w_j \mbox{ with } 0\le j\le m-1\}$.
\end{proof}
%

\begin{lem}\label{canonicalLemma}
Every labeled Reeb graph $(\Gamma_{f},\lf)$ can be transformed
into a canonical one through a finite sequence of elementary
deformations.
\end{lem}

\begin{proof}

Our proof is in two steps: first we show how to transform an
arbitrary Reeb graph into a minimal one; then how to reduce a
minimal Reeb graph to the canonical form.

The first step is by induction on $s=p+q$, with $p$ and $q$
denoting the number of minima and maxima of $f$. If $s=2$, then
$\Gamma_f$ is already minimal  (see Definition
\ref{minimal-canonical}). Let us assume that any Reeb graph  with
$s\ge 2$ vertices of degree 1 can be transformed into a minimal
one through a certain deformation. Let $\Gamma_{f}$ have $s+1$
vertices of degree 1. Thus, at least one between $p$ and $q$ is
greater than one. Let $p>1$ (the case $q>1$ is analogous). By
Lemma~\ref{pathLemma}~$(i)$, if $u,v$ correspond to two minima of
$f$, we can construct a deformation $T$ such that in
$T(\Gamma_{f},\lf)$ these vertices are both adjacent to a certain
vertex $w$ of degree 3. Let $T(\Gamma_{f},\lf)= (\Gamma,\ell)$,
with $\ell(u)<\ell(v)<\ell(w)$. If there exists a vertex $w'\in
\ell^{-1}(]\ell(v),\ell(w)[)$, since $v,w'$ cannot be adjacent, we
can apply a deformation of type (R)
 relabeling only $v$, and get a new labeling $\ell'$ such that
 $\ell'(w')$ is not contained in $]\ell'(v),\ell'(w)[$. Possibly repeating
this procedure finitely many times, we get a new labeling, that
for simplicity we still denote by $\ell$, such that
$\ell^{-1}(]\ell(v),\ell(w)[)=\emptyset$.  Hence, through a
deformation of type (D) deleting $v,w$, the resulting labeled Reeb
graph  has  $s$ vertices of
degree 1. By the inductive hypothesis, it can be transformed into
a minimal Reeb graph.

Now we prove the second step. Let $\Gamma_{f}$ be  a minimal Reeb
Graph, i.e. $p=q=1$. The total number of splitting saddles (i.e.
vertices of degree 3 for which there are two higher adjacent
vertices) of $\Gamma_{f}$ is $\mathfrak{g}$. If $\mathfrak{g}=0$,
then $\Gamma_{f}$ is already canonical. Let us consider the case
$\mathfrak{g}\ge 1$. Let $v\in V(\Gamma_{f})$ be a splitting
saddle such that, for every cycle $C$ containing $v$,
$\lf(v)=\underset{w\in C}\min\{\lf(w)\}$, and let $C$ be one of
these cycles. By Lemma \ref{pathLemma}~$(ii)$, there exists a
deformation $T$ that transforms $C$ into a 2-cycle,  still having
$v$ as the lowest vertex. Let $v'$ be the highest vertex in this
$2$-cycle. We observe that no other cycles of $T(\Gamma_f,\lf)$
contain $v$ and $v'$, otherwise the initial assumption on $\lf(v)$
would be contradicted. Hence $v$, $v'$ and the edges adjacent to
them are not touched when applying again Lemma
\ref{pathLemma}~$(ii)$ to reduce the length of another cycle.
Therefore, iterating the same argument on a different splitting
saddle, after at most $\mathfrak{g}$ iterations  (actually  at
most $\mathfrak{g}-1$ would suffice)  $\Gamma_{f}$ is transformed
into a canonical Reeb graph.
\end{proof}

\begin{prop}\label{connected}
Let $f,g\in\F^0(\M)$. The set ${\mathcal T}((\Gamma_{f},\lf),
(\Gamma_{g},\Lg))$ is non-empty.
\end{prop}

\begin{proof}
By Lemma \ref{canonicalLemma} we can find two deformations $T_f$
and $T_g$ transforming $(\Gamma_{f},\lf)$ and $(\Gamma_{g},\Lg)$,
respectively, into canonical Reeb graphs. Apart from the labels, $\Gamma_{f}$ and $\Gamma_{g}$ are
isomorphic because associated with the same surface $\M$. Hence,
$T_f(\Gamma_{f},\lf)$ can be transformed into a graph isomorphic
to $T_g(\Gamma_{g},\Lg)$ through an elementary deformation of type
(R), say $T_{\mbox{\tiny{R}}}$. Thus $(\Gamma_{g},\Lg)\cong
T_g^{-1}T_{\mbox{\tiny{R}}}T_f(\Gamma_{f},\lf)$, i.e.
$T_g^{-1}T_{\mbox{\tiny{R}}}T_f\in{\mathcal T}((\Gamma_{f},\lf),
(\Gamma_{g},\Lg))$.
\end{proof}

A simple example illustrating the above proof is given in Figure
\ref{Tnoempty}.
\begin{figure}[htbp]
\psfrag{F}{$(\Gamma_{f},\lf)$}\psfrag{G}{$(\Gamma_{g},\Lg)$}\psfrag{R}{(R)}
\psfrag{K1}{(K$_1$)}\psfrag{K2}{(K$_2$)}\psfrag{K3}{(K$_3$)}\psfrag{D}{(D)}\psfrag{B}{(B)}
\includegraphics[width=12cm]{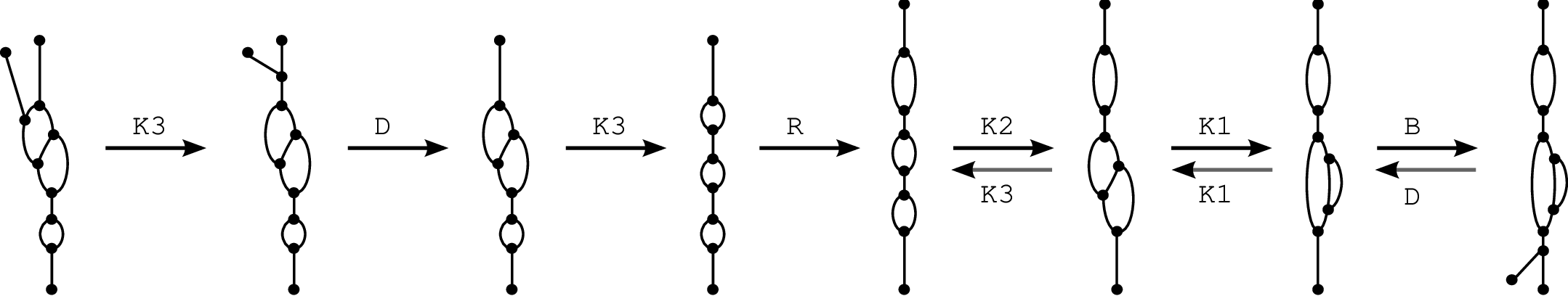}
\caption{\footnotesize{Using the procedure followed in the proof
of Proposition \ref{connected}, the leftmost labeled Reeb graph is
transformed into the rightmost one applying first the deformation
which reduces the former into its canonical form, then an
elementary deformation of type (R), and eventually the inverse of
the deformation which reduces the latter into its canonical
form.}}\label{Tnoempty}
\end{figure}

\section{edit distance between labeled Reeb graphs}\label{edit}

In this section we introduce an edit distance between labeled
Reeb graphs, in terms of the cost necessary to transform one graph
into another.

We begin  by defining the cost of a deformation. For the sake of
simplicity, in view of Proposition \ref{defGl}, whenever
$(\Gamma_{g},\Lg)\cong(\Gamma,\ell)$, we identify $V(\Gamma_g)$
with $V(\Gamma)$, and $\Lg$ with $\ell$. For all the notation
referring to the elementary deformations, see Definition
\ref{elementaryDef}.

\begin{defi}\label{cost 1}
Let $T\in{\mathcal T}((\Gamma_{f},\lf), (\Gamma_{g},\Lg))$ be a
 deformation.
\begin{itemize}
\item For $T$ elementary of type (B), inserting the vertices $u_1,u_2\in
V(\Gamma_{g})$,  the associated cost is
$$c(T)=\frac{|\Lg(u_1)-\Lg(u_2)|}{2}.$$
\item For $T$ elementary of type (D), deleting the vertices $u_1,u_2\in
V(\Gamma_{f})$, the associated cost is
$$c(T)=\frac{|\lf(u_1)-\lf(u_2)|}{2}.$$
\item For $T$ elementary of type (R), relabeling the vertices $v\in
V(\Gamma_{f})=V(\Gamma_{g})$, the associated cost is
$$c(T)=\underset{v \in V(\Gamma_{f})}\max|\lf(v)-\Lg(v)|.$$
\item For $T$ elementary of type (K$_i$), with $i=1,2,3$, relabeling the
vertices $u_1,u_2\in V(\Gamma_{f})$, the associated cost is
$$c(T)=\max\{|\lf(u_1)-\Lg(u_1)|,|\lf(u_2)-\Lg(u_2)|\}.$$
\item For $T\in{\mathcal T}((\Gamma_{f},\lf),
(\Gamma_{g},\Lg))$,  with $T=(T_1,\ldots ,T_r)$, the associated cost is
$$c(T)=\underset{i=1}{\overset{r}\sum} c(T_i).$$
\end{itemize}
\end{defi}

\begin{prop}\label{inverse}
For every deformation $T\in{\mathcal T}((\Gamma_{f},\lf),
(\Gamma_{g},\Lg))$, $c(T^{-1})=c(T)$.
\end{prop}
\begin{proof}
It is sufficient to observe that, for every deformation $T=(T_1,
\ldots, T_r)$ such that $T(\Gamma_{f}, \lf)\cong(\Gamma_{g},
\Lg)$, Definitions \ref{cost 1} and \ref{definverse} imply the
following equalities:
\begin{equation*}
c(T)=\underset{i=1}{\overset{r}\sum}c(T_i)=\underset{i=1}{\overset{r}\sum}c(T_i^{-1})=c(T^{-1}).
\end{equation*}
\end{proof}


\begin{theorem}\label{editdist}
For every $f,g\in\F^0(\M)$, we set
$$d_E((\Gamma_{f}, \lf),(\Gamma_{g},\Lg))=\inf_{T\in
\mathcal{T}((\Gamma_{f}, \lf),(\Gamma_{g},\Lg))}c(T).$$ It holds that
$d_E$ is a pseudo-metric on isomorphism classes of labeled Reeb
graphs.
\end{theorem}

\begin{proof}
By Proposition \ref{connected}, $d_E$ is a real number.
The coincidence property can be verified by observing that the
deformation of type (R) such that $T(\Gamma_{f}, \lf)=(\Gamma_{f},
\lf)$ has a cost equal to 0; the symmetry property is
a consequence of Proposition \ref{inverse}; the triangle
inequality can be proved in the standard way.
\end{proof}

In order to say that $d_E$ is actually a metric, we need to prove that if $d_E((\Gamma_{f},
\lf),(\Gamma_{g},\Lg))$ vanishes then $(\Gamma_{f}, \lf)\cong
(\Gamma_{g},\Lg)$. This will be done in Section
\ref{lowbound}. Nevertheless, for simplicity, we already refer to
$d_E$ as to the {\em edit distance}.

The following proposition shows that when a labeled Reeb graph can
be transformed into another one through a finite sequence of
deformations of type (D), the same transformation can be realized
also through a cheaper deformation which involves a relabeling of
vertices. Analogous propositions can be given for other types of
deformations. These results yield, in some cases, sharper
estimates of the edit distance between labeled Reeb graphs.

\begin{prop}\label{DversusDR}
For $T\in{\mathcal T}((\Gamma_{f},\lf), (\Gamma_{g},\Lg))$, if
$T=(T_1,\ldots ,T_n)$ and $T_i$ is an elementary deformation of
type (D) for each  $i=1,\ldots, n$, then there exists a
deformation $S\in{\mathcal T}((\Gamma_{f},\lf),
(\Gamma_{g},\Lg))$, with $S=(S_0,S_1\ldots, S_n)$ such that
$S_{0}$ is an elementary deformation of type (R), $S_1,\ldots,
S_{n}$ are elementary deformations of type (D), and
$c(S)=\underset{i=1,\dots,n}\max c(T_i)$. Hence $c(S)<c(T)$ when
$n>1$.
\end{prop}

\begin{proof}
Let $T=(T_1,\ldots ,T_n)$, with each $T_i$ of type (D), and let
$v_i,w_i$ be the vertices of $\Gamma_f$ deleted by $T_i$. It is not
restrictive to assume that $\lf(v_i)<\lf(w_i)$. For $n=1$, it is
sufficient to take $S_0$ as the elementary deformation of type (R)
such that $S_0(\Gamma_{f},\lf)=(\Gamma_{f},\lf)$ and $S_1=T_1$.
For $n>1$, for every $i,j$ with $1\le i,j\le n$, let us set
$T_i\preceq T_j$ if and only if $[\lf(v_i),\lf(w_i)]\subseteq
[\lf(v_j),\lf(w_j)]$. Let us denote by $T_{r_1}, \ldots ,T_{r_m}$
the maximal elements of the poset $(\{T_1,\ldots,T_n\},\preceq)$.

We observe that, for $1\le i\le n$,    there exists exactly one
value $k$, with $1\le k\le m$, for which
$[\lf(v_{i}),\lf(w_{i})]\subseteq [\lf(v_{r_k}),\lf(w_{r_k})]$.
Moreover, $[\lf(v_{i}),\lf(w_{i})]\cap
[\lf(v_{r_h}),\lf(w_{r_h})]=\emptyset$ for every $h\ne k$ because
$T_i$ is an elementary deformation of type (D).

To define $S_0$, we take $\ell:V(\Gamma_f)\to \R$  as
follows. Let $0<\eps<\underset{k=1,\ldots,m}\min\dfrac{\lf(w_{r_k})-\lf(v_{r_k})}{2}$. For $1\le k\le m$, we set
$\ell(v_{r_k})=\dfrac{\lf(w_{r_k})+\lf(v_{r_k})}{2}-\eps$ and
$\ell(w_{r_k})=\dfrac{\lf(w_{r_k})+\lf(v_{r_k})}{2}+\eps$. Next,
for $1\le i\le n$, assuming $[\lf(v_{i}),\lf(w_{i})]\subseteq
[\lf(v_{r_k}),\lf(w_{r_k})]$, we let $\lambda_{i},\mu_{i}\in
[0,1]$ be the unique values for which
$\lf(v_{i})=(1-\lambda_{i})\lf(v_{r_k})+\lambda_{i}\lf(w_{r_k})$
and $\lf(w_{i})=(1-\mu_{i})\lf(v_{r_k})+\mu_{i}\lf(w_{r_k})$, and
we set $\ell(v_{i})=(1-\lambda_{i})\ell(v_{r_k})+\lambda_{i}\ell(w_{r_k})$
and $\ell(w_{i})=(1-\mu_{i})\ell(v_{r_k})+\mu_{i}\ell(w_{r_k})$.
We observe that $\ell$ preserves the vertex order induced by $\lf$
and, therefore, $S_0$  defined by setting 
$S_0(\Gamma_f,\lf)= (\Gamma_f,\ell)$ is an elementary deformation of type (R).  
By Definition \ref{cost 1},
the cost of $S_0$ is {\setlength\arraycolsep{2pt}\begin{eqnarray*}
c(S_0)&=&\underset{\tiny{ i=1,\ldots,
n}}\max\left\{\max\left\{\left|\lf(v_{i})-\ell(v_{i})\right|,\left|\lf(w_{i})-\ell(w_{i})\right|\right\}\right\}.
\end{eqnarray*}}
A direct computation shows that $\ell(v_{i})-\lf(v_{i})\le
\ell(v_{r_k})-\lf(v_{r_k})$ and $\lf(v_{i})-\ell(v_{i})\le
\lf(w_{r_k})-\ell(w_{r_k})$. Analogously,
$\ell(w_{i})-\lf(w_{i})\le \ell(v_{r_k})-\lf(v_{r_k})$ and
$\lf(w_{i})-\ell(w_{i})\le \lf(w_{r_k})-\ell(w_{r_k})$. Hence
{\setlength\arraycolsep{2pt}\begin{eqnarray}\label{c(S0)}
c(S_0)&=&\underset{\tiny{ k=1,\ldots, m}}\max\left\{\max\left\{\ell(v_{r_k})-\lf(v_{r_k}),\lf(w_{r_k})-\ell(w_{r_k})\right\}\right\}\nonumber\\
&=&\underset{k=1,\dots,m}\max\dfrac{\lf(w_{r_k})-\lf(v_{r_k})}{2}-\eps
= \underset{k=1,\dots,m}\max c(T_{r_k})-\eps.
\end{eqnarray}}

Now we set $S_i$, for $i=1,\ldots ,n$, to be  the elementary
deformation of type (D) that deletes the vertices $v_i,w_i$  from $S_0(\Gamma_f,\lf)$. If $[\lf(v_{i}),\lf(w_{i})]\subseteq
[\lf(v_{r_k}),\lf(w_{r_k})]$, then
{\setlength\arraycolsep{2pt}\begin{eqnarray}\label{c(S1)}
c(S_i)&=&\dfrac{\ell(w_{i})-\ell(v_{i})}{2}\le\dfrac{\ell(w_{r_k})-\ell(v_{r_k})}{2}=\eps.
\end{eqnarray}}

Setting $S=(S_0,S_1,\ldots ,S_n)$, we have
$S\in\mathcal{T}((\Gamma_f,\ell_f),(\Gamma_g,\ell_g))$, and by
formulas (\ref{c(S0)}) and (\ref{c(S1)}):
{\setlength\arraycolsep{2pt}\begin{eqnarray*}
c(S)&=&c(S_0)+\underset{i=1}{\overset{n}\sum}c(S_{i})\le\underset{k=1,\dots,m}\max
c(T_{r_k})-\eps+n\cdot \eps.
\end{eqnarray*}}
Therefore, $\underset{k=1,\dots,m}\max c(T_{r_k})-\eps\le c(S)\le
\underset{k=1,\dots,m}\max c(T_{r_k})+(n-1)\eps$. By the
arbitrariness of $\eps$,  we get $c(S)=\underset{k=1,\dots,m}\max
c(T_{r_k})$, yielding the claim.

\end{proof}


\section{Stability}\label{stab}

This section is devoted to proving that Reeb graphs of orientable
surfaces are stable under function perturbations. More precisely,
it will be shown that arbitrary changes in simple Morse functions
with respect to the $C^0$-norm imply not greater changes in the
edit distance between the associated labeled Reeb graphs.
Formally:
\begin{theorem}\label{global}
For every $f, g \in \F^0(\M)$, letting
${\|f-g\|}_{C^0}=\underset{p\in\M}\max|f(p)-g(p)|$, we have
$$d_E((\Gamma_{f},\lf),(\Gamma_{g},\Lg))\le {\|f-g\|}_{C^0}.$$
\end{theorem}

We observe that such a result is strictly related the way the cost
of an elementary deformation of type (R) was defined as the  following Example \ref{example} shows.

\begin{exa}\label{example}
Let $f,g:\M\to\R$ with $f,g\in\F^0(\M)$ as illustrated in
Figure~\ref{exastab}. 

\begin{figure}[htbp]
\begin{center}
\psfrag{q1}{$q_1$} \psfrag{p1}{$p_1$}\psfrag{q2}{$q_2$}
\psfrag{p2}{$p_2$}\psfrag{q3}{$q_3$} \psfrag{p3}{$p_3$}
\psfrag{q}{$q$} \psfrag{q'}{$q'$} \psfrag{p}{$p$}
\psfrag{p'}{$p'$}\psfrag{c1}{$c_1$}
\psfrag{c2}{$c_2$}\psfrag{c3}{$c_3$}\psfrag{c1+a}{$c_1+a$}
\psfrag{c2+a}{$c_2+a$}\psfrag{c3+a}{$c_3+a$}\psfrag{b}{$b$}\psfrag{d}{$d$}
\psfrag{L}{$(\Gamma_{f},\lf)$} \psfrag{S}{$(\Gamma_{g},\Lg)$}
\psfrag{f}{$f$} \psfrag{g}{$g$} \psfrag{M}{$\max z$}
\psfrag{m}{$\min z$}
\includegraphics[height=4cm]{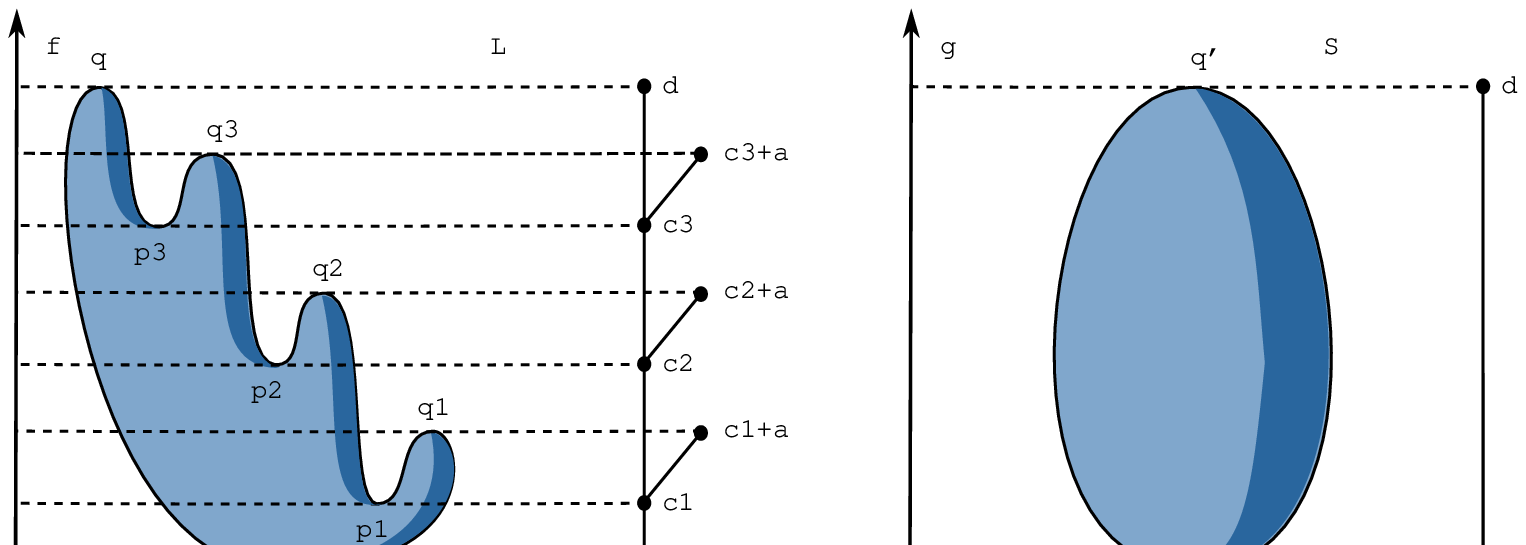}
\caption{\footnotesize{The functions $f,g\in\F^0(\M)$ considered
in Example \ref{example}.}}\label{exastab}
\end{center}
\end{figure}

Let $f(q_i)-f(p_i)=a$, $i=1,2,3$. Up to
re-parameterization of $\M$, we have
${\|f-g\|}_{C^0}=\frac{a}{2}$. The
deformation $T$ that deletes the three edges $e(p_i,q_i)\in
E(\Gamma_f)$ has cost $c(T)=3\cdot\frac{a}{2}$, implying $d_E((\Gamma_{f},\lf),(\Gamma_{g},\Lg))\le 3\cdot{\|f-g\|}_{C^0}$. On the other
hand, applying Proposition~\ref{DversusDR} we see that actually
$d_E((\Gamma_{f},\lf),(\Gamma_{g},\Lg))\le{\|f-g\|}_{C^0}$.
Indeed, for every $0<\epsilon<\frac{a}{2}$, we can apply to
$(\Gamma_{f},\lf)$ a deformation of type (R), that relabels the
vertices $p_i,q_i$, $i=1,2,3$, in such a way that $\lf(p_i)$ is
increased by $\frac{a}{2}-\epsilon$, and $\lf(q_i)$ is decreased
by $\frac{a}{2}-\epsilon$, composed with three deformations of
type (D) that delete $p_i$, $q_i$ and the edge $e(p_i,q_i)$,
for $i=1,2,3$ respectively. 
\end{exa}

In order to prove Theorem \ref{global},  we consider the set
$\F(\M)$ of smooth real-valued functions on $\M$ endowed with the
$C^2$ topology, which may be defined as follows. Let
$\{U_\alpha\}$ be a coordinate covering of $\M$ with coordinate
maps $\varphi_\alpha: U_\alpha\to \R^2$, and let $\{C_\alpha\}$ be
a compact refinement of $\{U_\alpha\}$. For every positive
constant $\delta >0$ and every $f\in\F(\M)$, define $N(f,\delta)$
as the subset of $\F(\M)$ consisting of all maps $g$ such that,
denoting $f_\alpha=f\circ \varphi_\alpha^{-1}$ and
$g_\alpha=g\circ \varphi_\alpha^{-1}$, it holds that
$\underset{i+j\le 2}\max\left|\frac{\partial^{i+j}}{\partial
x^i\partial y^j}(f_\alpha-g_\alpha)\right|<\delta$ at all points
of $\varphi_\alpha(C_\alpha)$. The $C^2$ topology is the topology
obtained by taking the sets $N(f,\delta)$ as a base of
neighborhoods.

Next we consider the strata $\F^0(\M)$ and $\F^1(\M)$ of the
\emph{natural stratification} of $\F(\M)$, as presented by Cerf in
\cite{Cerf70}.

\begin{itemize}
\item The stratum $\F^0(\M)$ is the set of simple Morse functions.
\item The stratum $\F^1(\M)$ is the disjoint union of two sets
$\F^1_{\alpha}(\M)$ and $\F^1_{\beta}(\M)$, where
\begin{itemize}
\item $\F^1_{\alpha}(\M)$ is the set of functions whose critical
levels contain exactly one critical point, and the critical points
are all non-degenerate, except exactly one. \item
$\F^1_{\beta}(\M)$ is the set of Morse functions whose critical
levels contain at most one critical point, except for one level
containing exactly two critical points.
\end{itemize}
\end{itemize}
$\F^1(\M)$ is a sub-manifold of co-dimension 1 of
$\F^0(\M)\cup\F^1(\M)$, and the complement of
$\F^0(\M)\cup\F^1(\M)$ in $\F(\M)$ is of co-dimension greater than
1. Hence, given two functions $f,g\in\F^0(\M)$, we can always find
$\widehat{f},\widehat{g}\in\F^0(\M)$ arbitrarily near to $f, g$,
respectively, for which
\begin{itemize}
\item $\widehat{f}$, $\widehat{g}$ are RL-equivalent to
$f$, $g$, respectively,
\end{itemize}
and the path $h(\lambda)=(1-\lambda)\widehat{f}+\lambda
\widehat{g}$, with $\lambda\in [0,1]$, is such that
\begin{itemize}
\item $h(\lambda)$ belongs to $\F^0(\M)\cup\F^1(\M)$ for every
$\lambda\in [0,1]$; \item $h(\lambda)$ is transversal to
$\F^1(\M)$.
\end{itemize}
As a consequence, $h(\lambda)$ belongs to $\F^1(\M)$ for at most a
finite collection of values $\lambda$, and does not traverse
strata of co-dimension greater than 1 (see, e.g., \cite{EdHa02}).\\

With these preliminaries set, the stability theorem will be proven
by considering a path that connects $f$ to $g$ via $\widehat{f}$,
$h(\lambda)$, and $\widehat{g}$ as aforementioned. This path can
be split into a finite number of linear sub-paths whose endpoints
are such that the stability theorem holds on them, as will be
shown in some preliminary lemmas. In conclusion,
Theorem~\ref{global} will be proven by applying the triangle
inequality of the edit distance. \\

In the following preliminary lemmas, $f$ and $g$ belong to $\F^0(\M)$
and $h:[0,1]\to\F(\M)$ denotes their convex linear combination:
$h(\lambda)=(1-\lambda)f+\lambda g$.

\begin{lem}\label{pathH}
${\|h(\lambda')-h(\lambda'')\|}_{C^0}=|\lambda'-\lambda''|\cdot{\|f-g\|}_{C^0}$
for every $\lambda',\lambda''\in[0,1]$.
\end{lem}

\begin{proof}
{\setlength\arraycolsep{2pt}\begin{eqnarray*}
{\|h(\lambda')-h(\lambda'')\|}_{C^0}&=&{\|(1-\lambda')f+\lambda'
g-(1-\lambda'')f-\lambda''
g\|}_{C^0}\\&=&{\|(\lambda''-\lambda')f-(\lambda''-\lambda')g\|}_{C^0}=|\lambda'-\lambda''|\cdot{\|f-g\|}_{C^0}.\end{eqnarray*}}
\end{proof}

\begin{lem}\label{pathstabilityF0}
 If $\h(\lambda)
\in \F^0(\M)$ for every $\lambda \in [0,1]$, then
$d_E((\Gamma_{f},\lf),(\Gamma_{g},\Lg))\le {\|f-g\|}_{C^0}.$
\end{lem}

\begin{proof}
The statement can be proved in the same way as \cite[Prop.
5.4]{DiLa12}.
\end{proof}

\begin{lem}\label{pointonF1}
Let $h(\lambda)$ intersect $\F^1(\M)$ transversely at $h(\overline
\lambda)$, $0<\overline\lambda<1$, and nowhere else. Then, for
every constant value $\delta
>0$, there exist two real numbers $\lambda', \lambda''$ with
$0<\lambda'<\overline\lambda<\lambda''<1$,  such that
$$d_E((\Gamma_{h(\lambda')},\ell_{h(\lambda')}),(\Gamma_{h(\lambda'')},\ell_{h(\lambda'')}))\le
\delta.$$
\end{lem}

\begin{proof}
In this proof we use the notion of universal deformation of a function. More
details on universal deformations may be found in \cite{
Ma82, Se72}. In particular, we will consider two different
universal deformations $F$ and $G$ of $\overline{h}=h(\overline
\lambda)$. Firstly we show how $F$ and $G$ yield the claim, and
then construct them.

We use the fact that, being two universal deformations
of $\overline h\in\F^1(\M)$, $F$ and $G$ are equivalent. This
means that there exist a diffeomorphism $\eta(s)$ of $\R$ with
$\eta(0)=0$, and a local diffeomorphism $\phi(s,(x,y))$, with
$\phi(s,(x,y))=(\eta(s),\psi(\eta(s),(x,y)))$ and
$\phi(0,(x,y))=(0,(x,y))$, such that $F=(\eta^*G)\circ \phi$.
Hence, apart from the labels,  the Reeb graphs of $F(s,\cdot)$ and $G(\eta(s),\cdot)$ are
isomorphic. Moreover, the difference  the labels
at  corresponding vertices in the Reeb graphs of $F(s,\cdot)$ and
$G(\eta(s),\cdot)$ continuously depends on $s$, and is $0$ for
$s=0$. Therefore, for every $\delta>0$, taking $|s|$ sufficiently
small, it is possible to transform the labeled Reeb graph of
$F(s,\cdot)$ into that of $G(\eta(s),\cdot)$, or viceversa, by a
deformation of type (R) whose cost is not greater than $\delta/3$.
Moreover, as equality (\ref{costalphabeta}) will show, for every
$\delta>0$, $|s|$ can be taken sufficiently small that the
distance between the labeled Reeb graphs of $G(\eta(s),\cdot)$ and
$G(\eta(-s),\cdot)$ is not greater that $\delta/3$. Thus, applying
the triangle inequality, we deduce that, for every $\delta>0$,
there exists a sufficiently small $s>0$ such that the distance
between the labeled Reeb graphs of $F(s,\cdot)$ and $F(-s,\cdot)$
is not greater than $\delta$. The claim follows taking
$\lambda'=\overline\lambda-s$ and $\lambda''=\overline\lambda+s$.

We now construct the universal deformations $F(s,p)$ and $G(s,p)$,
with $s\in\R$ and $p\in\M$. We define $F$ by setting
$F(s,p)=\overline{h}(p)+s\cdot (g-f)(p)$. This deformation is
universal because $h(\lambda)$ intersects $\F^1(\M)$ transversely
at $h(\overline \lambda)$. In order to construct $G$, let us
consider separately the two cases $\F^1_{\alpha}(\M)$ and
$\F^1_{\beta}(\M)$.\\

{\em Case $\overline{h}\in\F^1_{\alpha}(\M)$:} Let $\overline p$
be the sole degenerate critical point of $\overline h$. Let
$(x,y)$ be a suitable local coordinate system around $\overline p$
in which the canonical expression of $\overline h$ is
$\overline{h}(x,y)=\overline h(\overline p)\pm x^2+y^3$.
\begin{figure}[htbp]
\psfrag{f}{$\overline{h}=G(0,\cdot)$}\psfrag{f1}{$G(\eta,\cdot)$,
$\eta>0$}\psfrag{f2}{$G(\eta,\cdot)$,
$\eta<0$}\psfrag{B}{(B)}\psfrag{D}{(D)}\psfrag{p}{$p$}\psfrag{q}{$q$}
\begin{center}
\includegraphics[width=12cm]{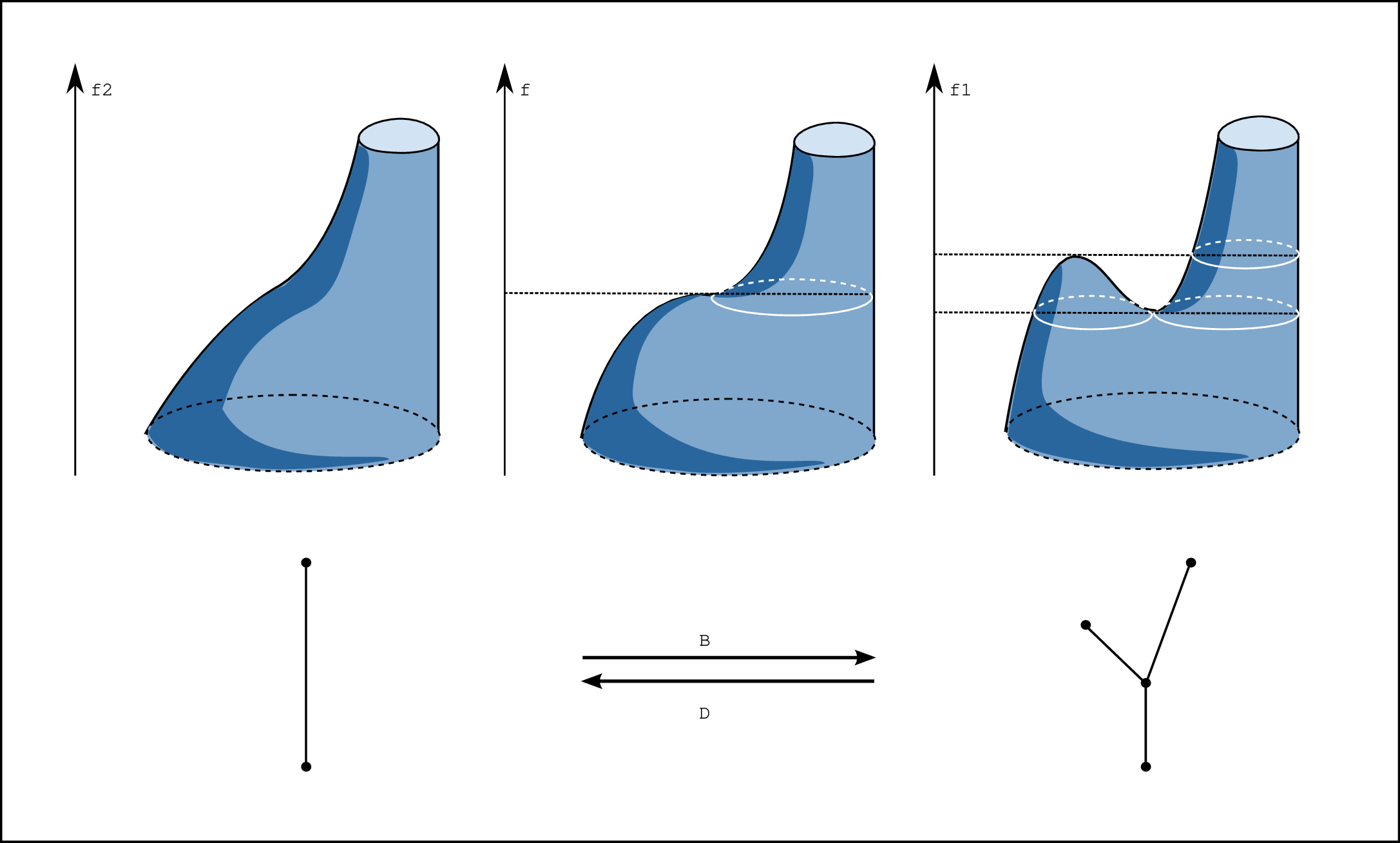}
\end{center}
\caption{\footnotesize{Center: A function $\overline
h\in\F^1_{\alpha}(\M)$; left-right: The universal deformation
$G(\eta,\cdot)$ with the associated labeled Reeb graphs for
$\eta<0$ and $\eta>0$.}\label{BD}}
\end{figure}
Let $\omega: \M\to \R$ be a smooth function equal to $1$ in a
neighborhood of $\overline p$, which decreases moving from
$\overline p$, and whose support is contained in the chosen
coordinate chart around $\overline p$. Finally, let
$G(\eta,(x,y))=\overline{h}(x,y)-\eta\cdot \omega(x,y)\cdot y$,
where $\eta\in\R$. For $\eta<0$, $G(\eta,\cdot)$ has
no critical points in the support of $\omega$ and is equal to
$\overline h$ everywhere else, while, for $\eta>0$,
$G(\eta,\cdot)$ has exactly two critical points in the support of
$\omega$, precisely $p_1=\left(0,-\sqrt{\frac{\eta}{3}}\right)$
and $p_2=\left(0,\sqrt{\frac{\eta}{3}}\right)$, and is equal to
$\overline h$ everywhere else  (see Figure~\ref{BD}). Therefore,
for every $\eta>0$ sufficiently small, the labeled Reeb graph of
$G(-\eta,\cdot)$ can be transformed into that of  $G(\eta,\cdot)$
by an elementary deformation $T$ of type~(B). Obviously, in the case $\eta<0$, the deformation we consider is of
type (D).

By
Definition~\ref{cost 1} and  Proposition
\ref{inverse}, a direct computation shows that the cost
of $T$ is
{\setlength\arraycolsep{2pt}\begin{eqnarray}\label{costalpha}
c(T)&=&2\cdot\left(\frac{|\eta|}{3}\right)^{3/2}.
\end{eqnarray}}

{\em Case $\overline{h}\in\F^1_{\beta}(\M)$:} Let $\overline{p}$
and $\overline{q}$ be the critical points of $\overline{h}$ such
that $\overline{h}(\overline{p})=\overline{h}(\overline{q})$. Since $\overline{p}$ is
non-degenerate, there exists a suitable local coordinate system
$(x,y)$ around $\overline{p}$ in which the canonical expression of
$\overline{h}$ is $\overline{h}(x,y)=\overline{h}(\overline{p})+
x^2+ y^2$ if $\overline{p}$ is a minimum, or
$\overline{h}(x,y)=\overline{h}(\overline{p})- x^2- y^2$ if
$\overline{p}$ is a maximum, or
$\overline{h}(x,y)=\overline{h}(\overline{p})\pm x^2\mp y^2$ if
$\overline{p}$ is a saddle point. Let $\omega: \M\to \R$ be a smooth function equal to $1$ in a
neighborhood of $\overline p$, which decreases moving from
$\overline p$, and whose support is contained in the coordinate
chart around $\overline{p}$ in which $\overline{h}$ has one of the
above expressions. Finally, let
$G(\eta,(x,y))=\overline{h}(x,y)+\eta\cdot \omega(x,y)$, where
$\eta=\eta(s)$, $s\in\R$. For every $\eta\in\R$, with $|\eta|$ sufficiently small,
$G(\eta,\cdot)$ has the same critical points, with the same
indices,
 as $\overline{h}$. As for critical values,
they are the same as well, apart from the value taken at
$\overline{p}$:
$G(\eta,\overline{p})=\overline{h}(\overline{p})+\eta$. 

We 
distinguish the following two situations illustrated in Figures \ref{MinSadMax} and  \ref{Saddles}:
\begin{enumerate}
\item the points $\overline{p}$ and $\overline{q}$ belong to two
different connected components of
$\overline{h}^{-1}(\overline{h}(\overline{p}))$;
\begin{figure}[htbp]
\begin{center}
\psfrag{p}{$\overline{p}$}\psfrag{q}{$\overline{q}$}\psfrag{f2}{$\widetilde{f_2}$}\psfrag{K1}{(K$_1$)}
\begin{tabular}{ccc}
\includegraphics[width=4cm]{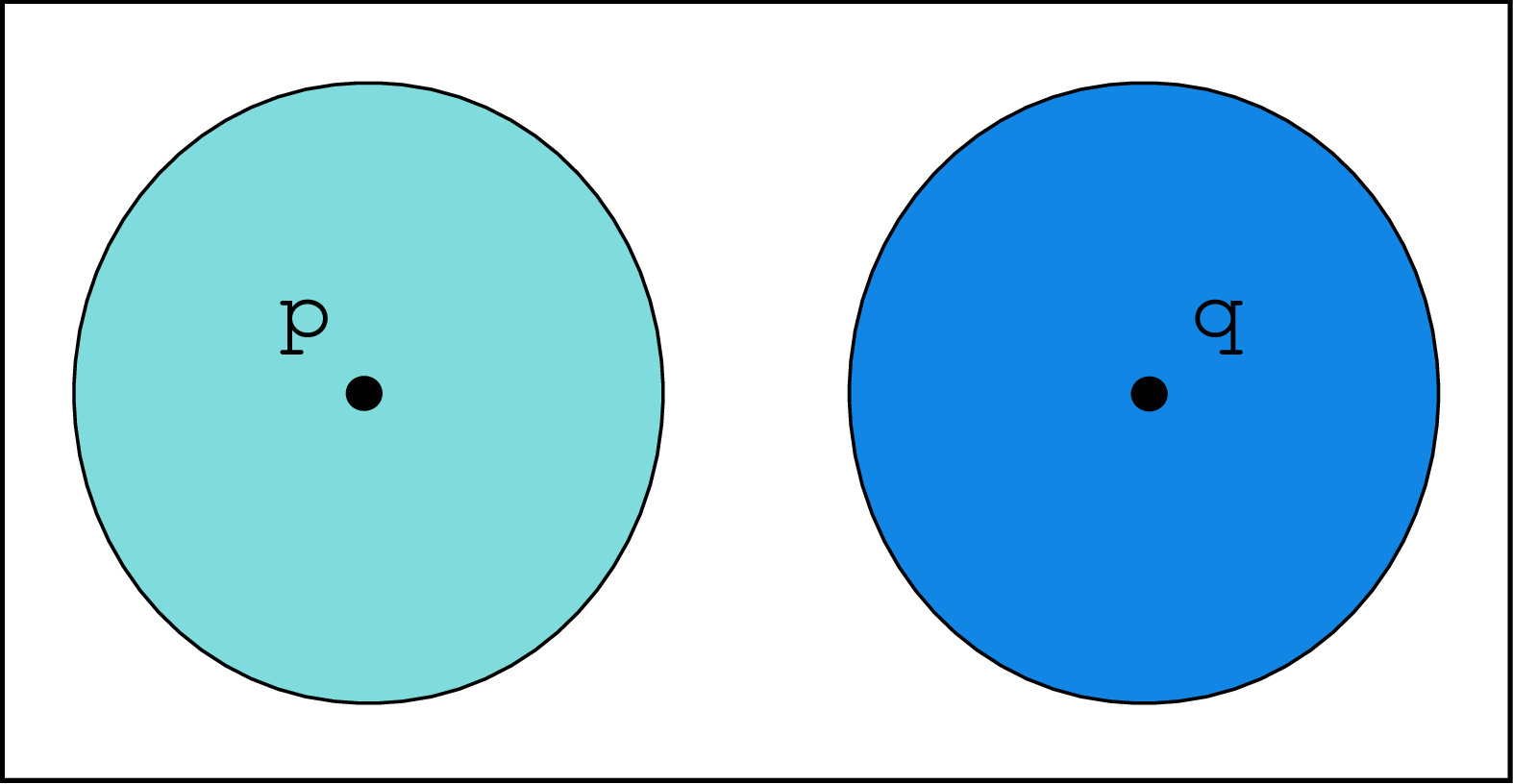}&\includegraphics[width=4cm]{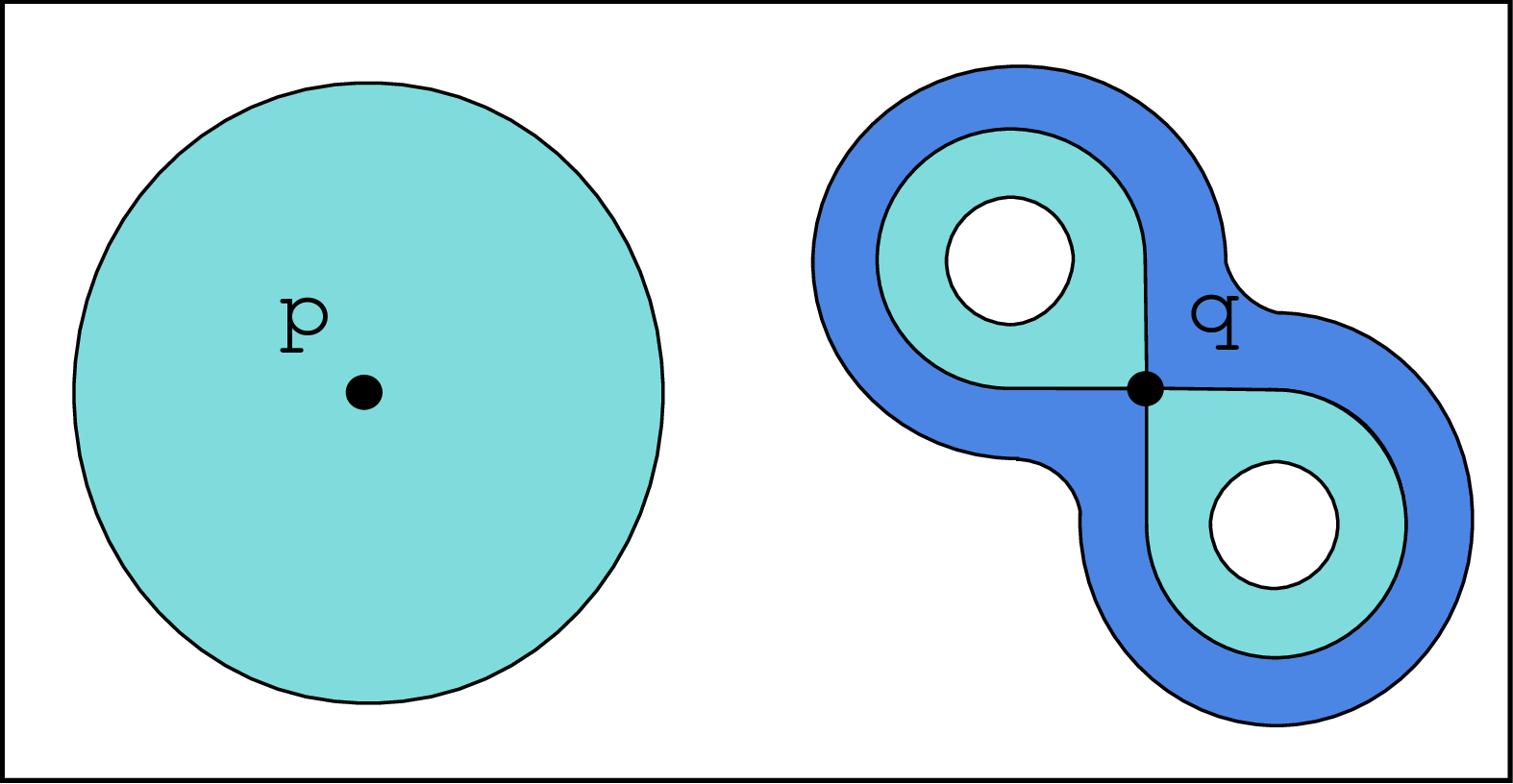}&
\includegraphics[width=4cm]{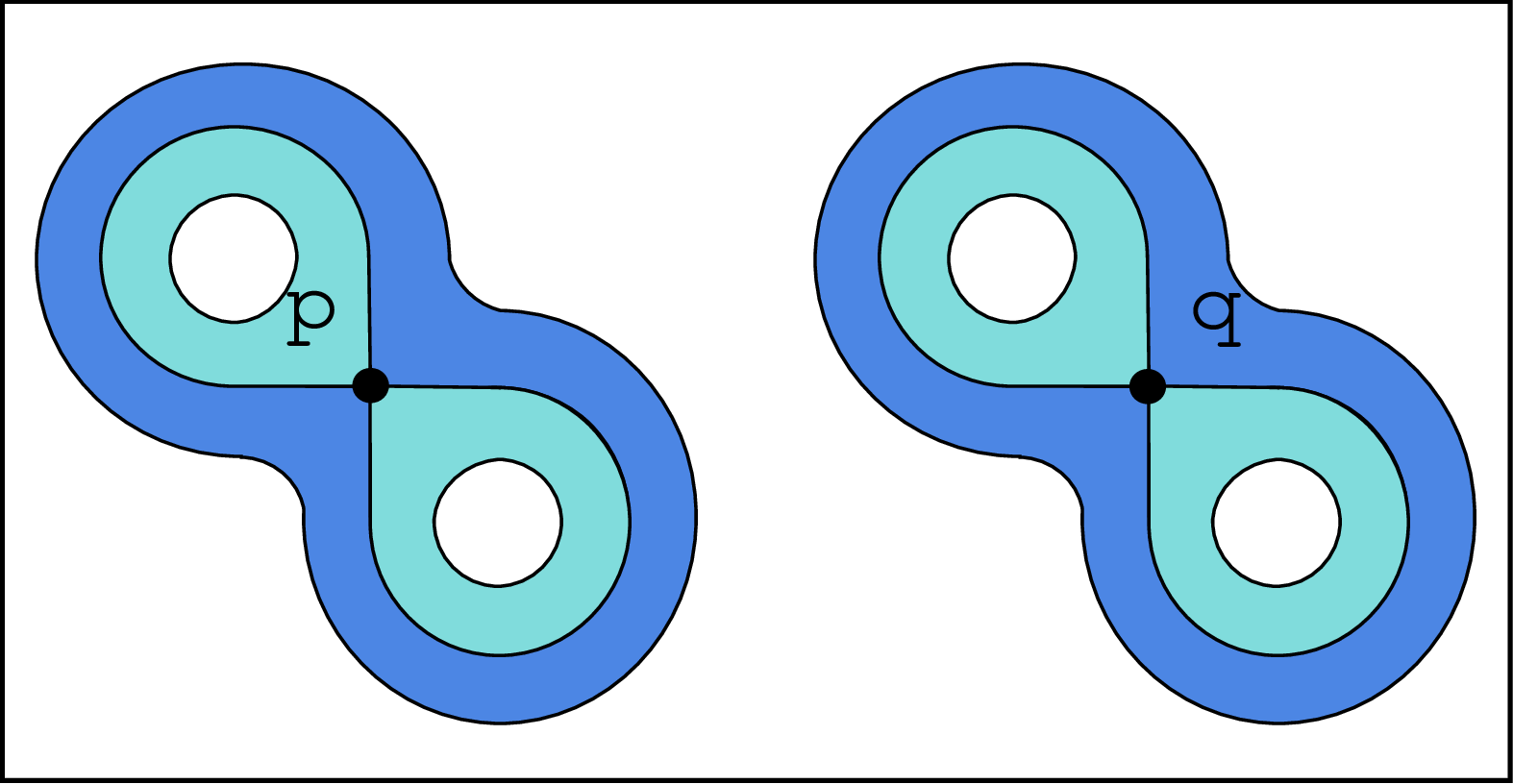}
\end{tabular}
\end{center}
\caption{\footnotesize{Two critical points in different connected
components of the same critical level. The dark (resp. light)
regions correspond to points below (resp. above) this critical
level. Possibly inverting the colors of one or both the
components, we have all the possible cases.}\label{MinSadMax}}
\end{figure}
\item the points $\overline{p}$ and $\overline{q}$ belong to the
same connected component of
$\overline{h}^{-1}(\overline{h}(\overline{p}))$.
\begin{figure}[htbp]
\psfrag{a}{$(a)$}\psfrag{b}{$(b)$}\psfrag{c}{$(c)$}\psfrag{d}{$(d)$}\psfrag{p}{$\overline{p}$}\psfrag{q}{$\overline{q}$}
\begin{center}
\begin{tabular}{cc}
\includegraphics[width=4cm]{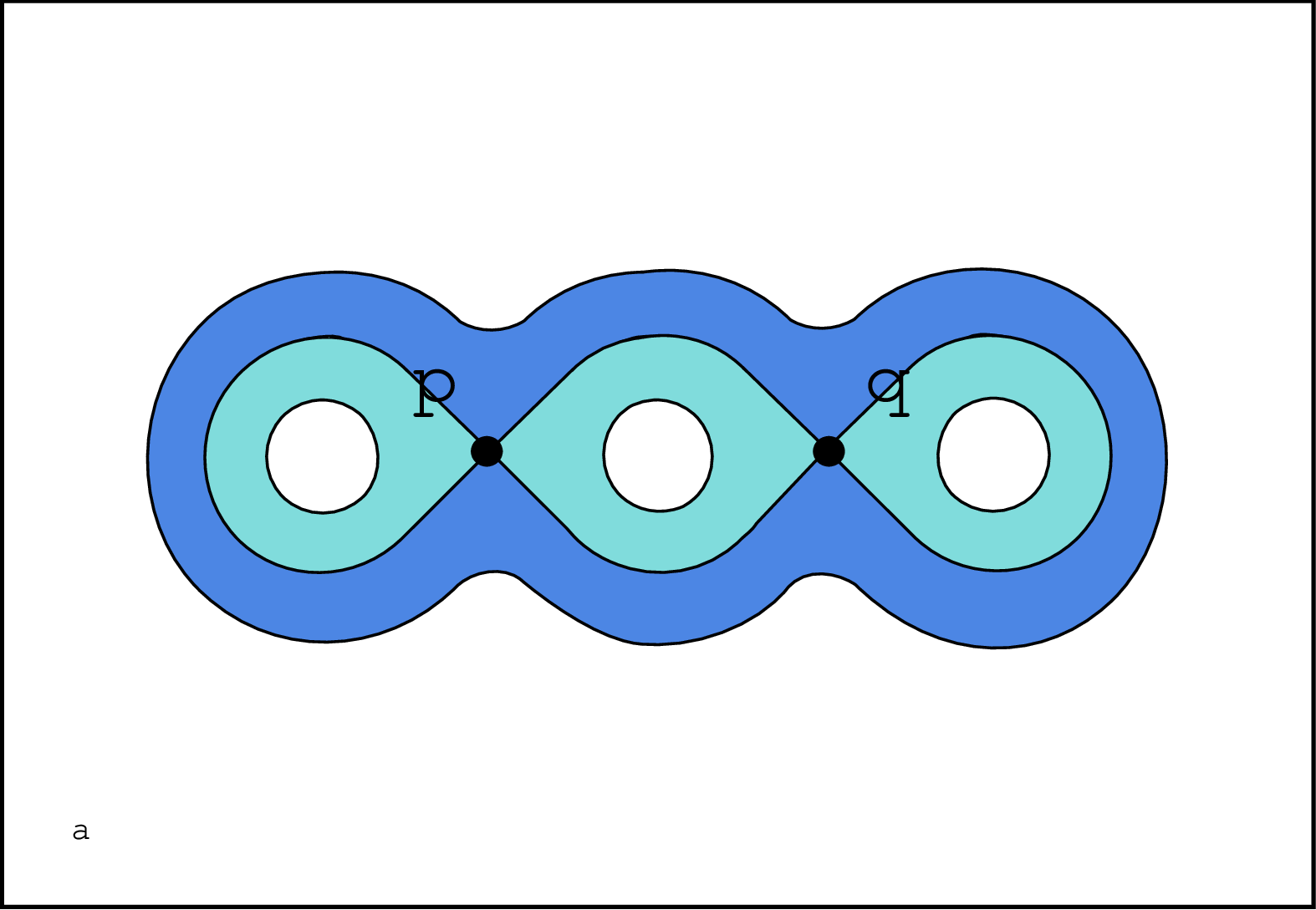}&\includegraphics[width=4cm]{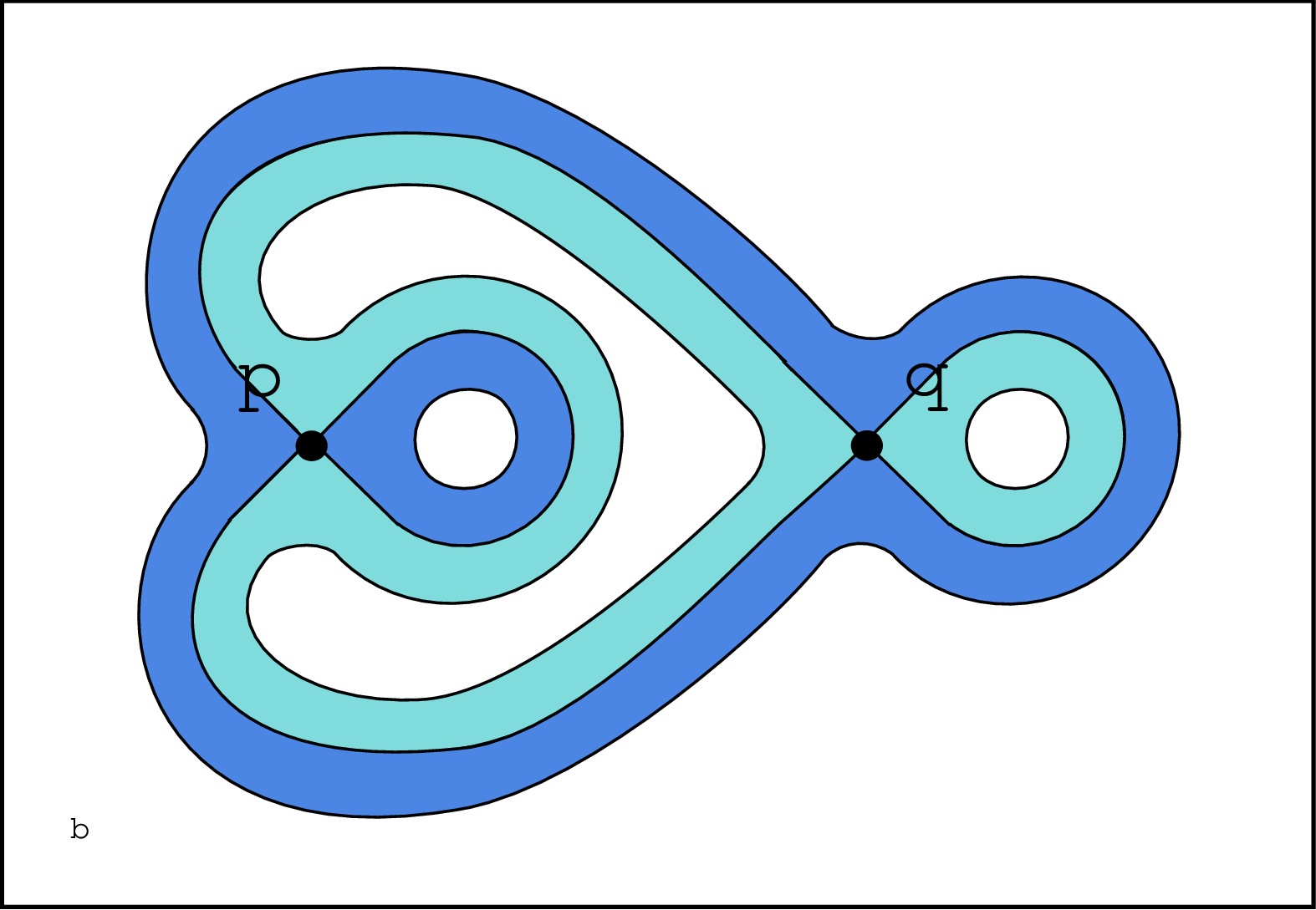}\\
\includegraphics[width=4cm]{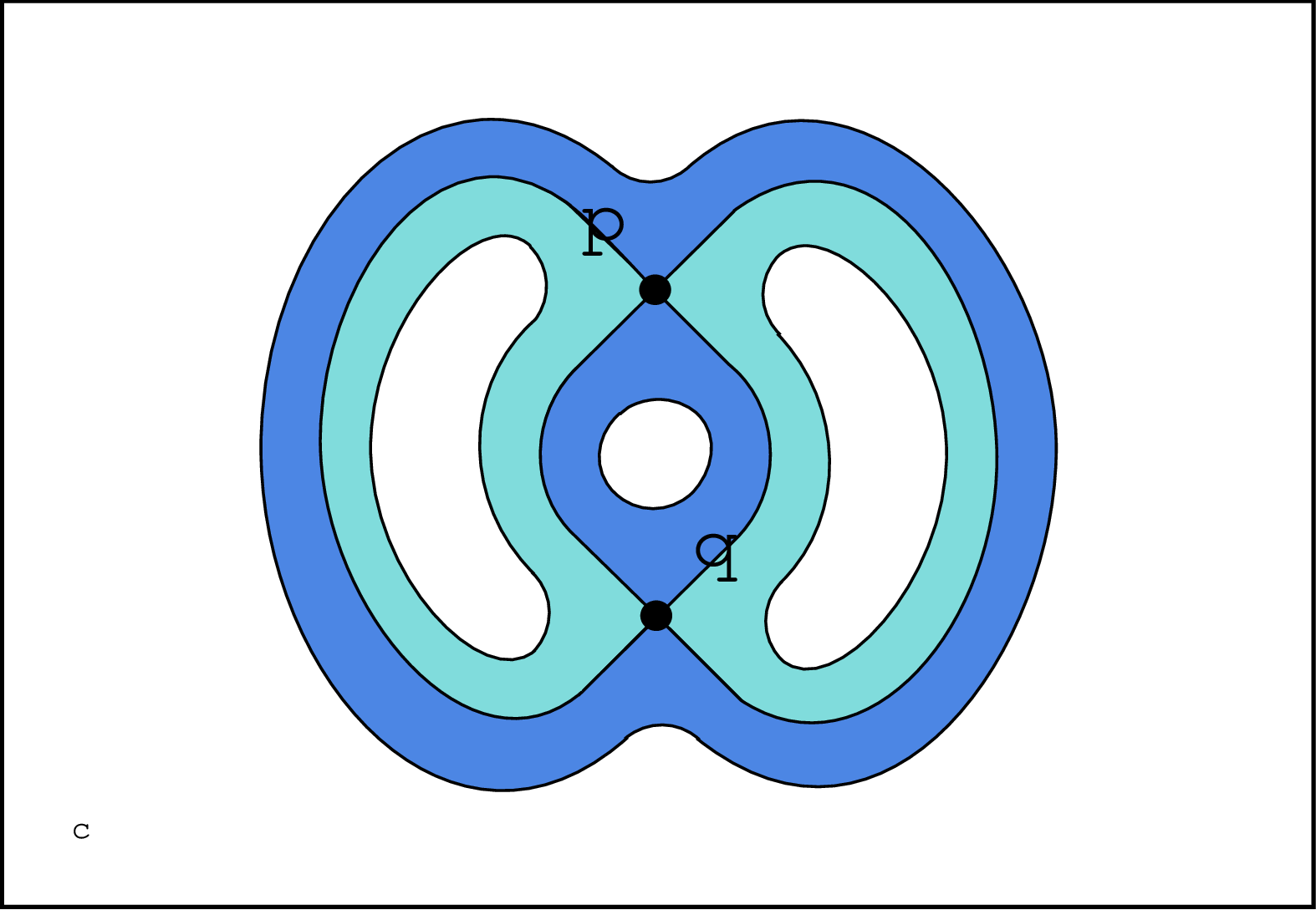}&\includegraphics[width=4cm]{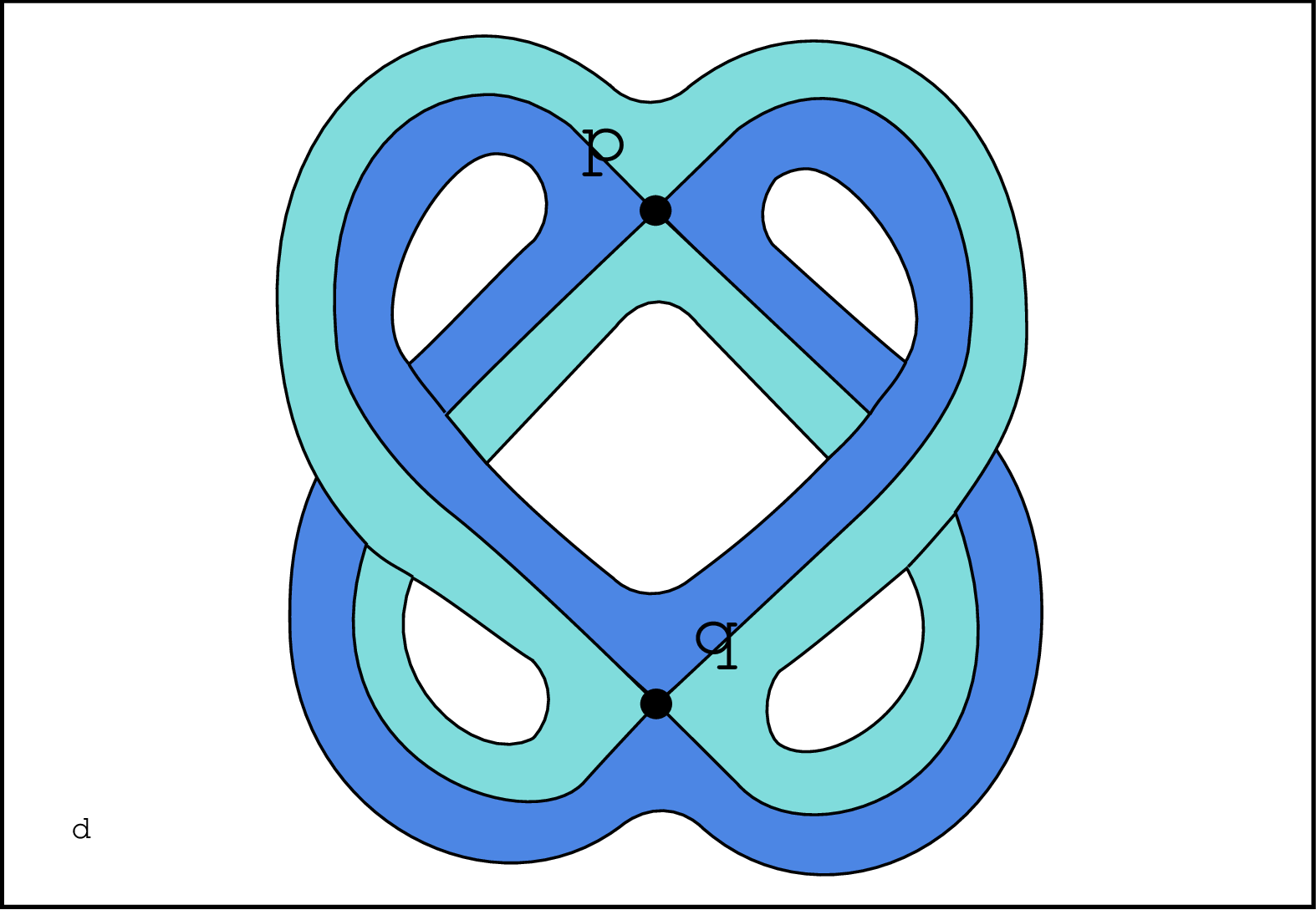}\end{tabular}
\end{center}
\caption{\footnotesize{Two critical points in the same connected
component of the same critical level. The dark (resp. light)
regions correspond to points below (resp. above) this critical
level. Possibly inverting the colors of this component, we have
all the possible cases.}\label{Saddles}}
\end{figure}
\end{enumerate}

In the situation (1),  for every  $\eta>0$ sufficiently small, the
labeled Reeb graphs of $G(-\eta,\cdot)$ and $G(\eta,\cdot)$ can be
obtained one from the other through an elementary deformation $T$
of type (R) (see, e.g., Figure~\ref{RFbeta}).
\begin{figure}[htbp]
\psfrag{f}{$\overline{h}=G(0,\cdot)$}\psfrag{f2}{$G(\eta,\cdot),\eta>0$}\psfrag{f1}{$G(\eta,\cdot),\eta<0$}\psfrag{R}{(R)}\psfrag{K3}{(K$_3$)}\psfrag{p}{$\overline{p}$}\psfrag{q}{$\overline{q}$}
\begin{center}
\includegraphics[width=12cm]{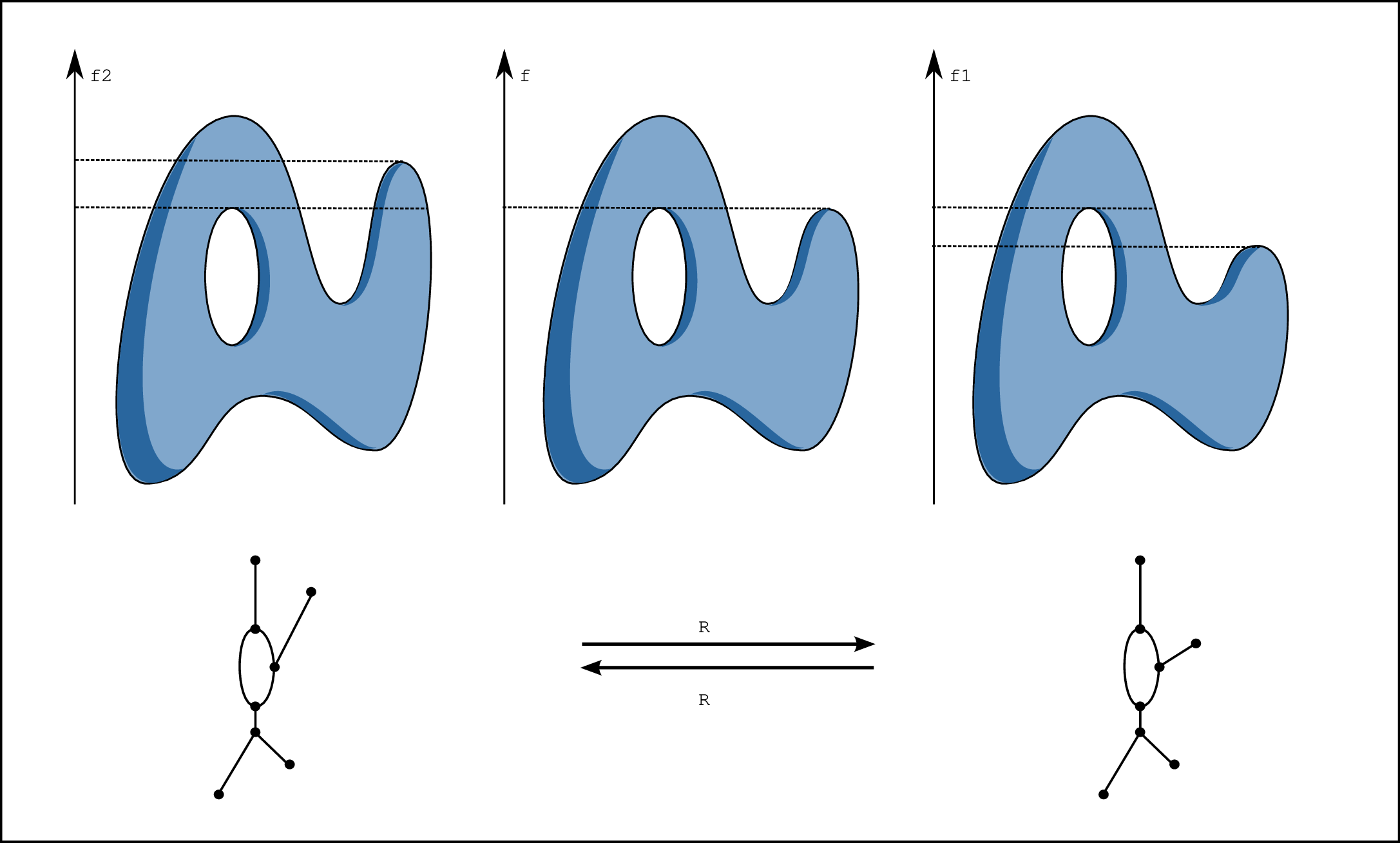}
\end{center}
\caption{\footnotesize{Center: A function
$\overline{h}\in\F^1_{\beta}(\M)$ as in case (1); left-right: The
universal deformation $G(\eta,\cdot)$ with the associated  labeled
Reeb graphs for $\eta<0$ and $\eta>0$.}\label{RFbeta}}
\end{figure}

In the situation (2), the following elementary deformations need to be
considered:
\begin{itemize}
\item If $\overline{p}$ and $\overline{q}$ are as in Figure
\ref{Saddles} $(a)$, then, for every $\eta>0$ sufficiently small, the
labeled Reeb graphs of $G(-\eta,\cdot)$ and $G(\eta,\cdot)$ can be
obtained one from the other through an elementary deformation $T$
of type (K$_1$) (see, e.g., Figure \ref{K1}).
\item If $\overline{p}$ and $\overline{q}$ are as in Figure
\ref{Saddles} $(b)$, then, for every $\eta>0$ sufficiently small, the
labeled Reeb graphs of $G(-\eta,\cdot)$ and $G(\eta,\cdot)$ can be
obtained one from the other through an elementary deformation $T$
of type (K$_3$) or (K$_2$) (see, e.g., Figure \ref{K2K3}).

\item If $\overline{p}$ and $\overline{q}$ are as in Figure
\ref{Saddles} $(c)$ or $(d)$, then, for every $\eta>0$ sufficiently
small, the labeled Reeb graphs of $G(-\eta,\cdot)$ and
$G(\eta,\cdot)$ can be obtained one from the other through an
elementary deformation $T$ of type (R) (see, e.g., Figures
\ref{R1}-\ref{R2}).
\end{itemize}
In all the cases, for every $\eta>0$ sufficiently small, the cost
of the considered deformation $T$ is:
{\setlength\arraycolsep{2pt}\begin{eqnarray}\label{costbeta}
c(T)&=&|\overline{h}(\overline{p})-\eta-(\overline{h}(\overline{p})+\eta)|=2\eta.
\end{eqnarray}}
In conclusion, from equalities (\ref{costalpha}) and
(\ref{costbeta}), for every $\eta>0$ sufficiently small, we get
$$d_E((\Gamma_{G(-\eta,\cdot)},\ell_{G(-\eta,\cdot)}),(\Gamma_{G(\eta,\cdot)},\ell_{G(\eta,\cdot)}))\le\max\left\{2\cdot\left(\frac{\eta}{3}\right)^{3/2},2\eta\right\}.$$
Thus, for every $\delta>0$, we can always take a value $|s|$
sufficiently small that $|\eta(s)|$ results small enough to imply
the following inequality:
{\setlength\arraycolsep{2pt}\begin{eqnarray}\label{costalphabeta}
d_E((\Gamma_{G(-\eta(s),\cdot)},\ell_{G(-\eta(s),\cdot)}),(\Gamma_{G(\eta(s),\cdot)},\ell_{G(\eta(s),\cdot)}))\le\delta/3.
\end{eqnarray}}
\begin{figure}[htbp]
\psfrag{f}{$\overline{h}=G(0,\cdot)$}\psfrag{f1}{$G(\eta,\cdot),\eta>0$}\psfrag{f2}{$G(\eta,\cdot),\eta<0$}\psfrag{R}{(R)}
\psfrag{K1}{(K$_1$)}\psfrag{K2}{(K$_2$)}\psfrag{K3}{(K$_3$)}\psfrag{p}{$\overline{p}$}\psfrag{q}{$\overline{q}$}
\centering
\includegraphics[width=12cm]{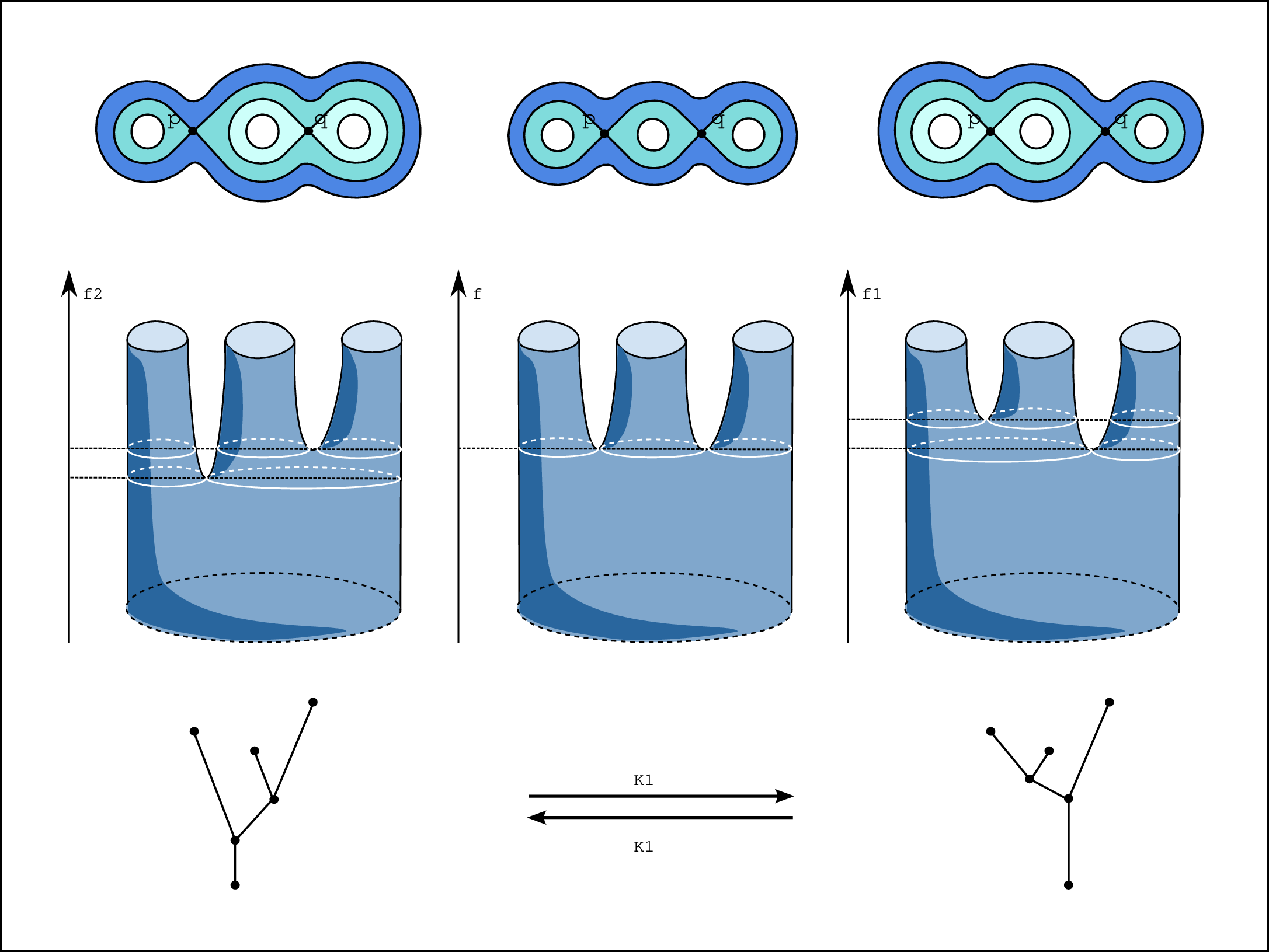}
\caption{\footnotesize{Center: A function
$\overline{h}\in\F^1_{\beta}(\M)$ as in case (2) with
$\overline{p},\overline{q}$ as in Figure \ref{Saddles}~$(a)$;
left-right: The universal deformation $G(\eta,\cdot)$ with the
associated  labeled Reeb graphs for $\eta<0$ and
$\eta>0$.}\label{K1}}
\psfrag{f}{$\overline{h}=G(0,\cdot)$}\psfrag{f1}{$G(\eta,\cdot),\eta>0$}\psfrag{f2}{$G(\eta,\cdot),\eta<0$}\psfrag{R}{(R)}
\psfrag{K1}{(K$_1$)}\psfrag{K2}{(K$_2$)}\psfrag{K3}{(K$_3$)}\psfrag{p}{$\overline{p}$}\psfrag{q}{$\overline{q}$}
\centering
\includegraphics[width=12cm]{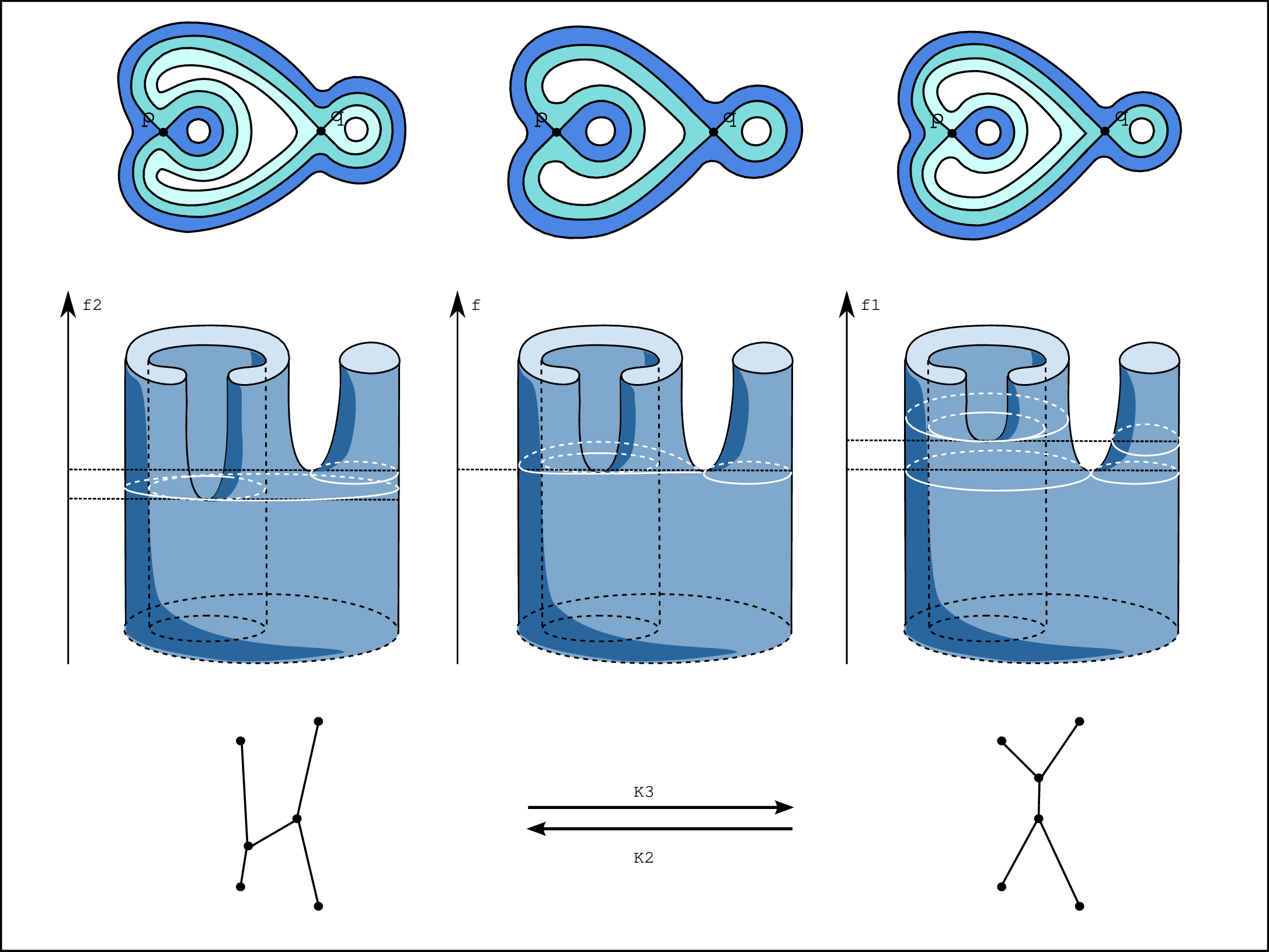}
\caption{\footnotesize{Center: A function
$\overline{h}\in\F^1_{\beta}(\M)$ as in case (2) with
$\overline{p},\overline{q}$ as in Figure \ref{Saddles}~$(b)$;
left-right: The universal deformation $G(\eta,\cdot)$ with the
associated  labeled Reeb graphs for $\eta<0$ and
$\eta>0$.}\label{K2K3}}
\end{figure}

\begin{figure}[htbp]
\psfrag{f}{$\overline{h}=G(0,\cdot)$}\psfrag{f1}{$G(\eta,\cdot),\eta>0$}\psfrag{f2}{$G(\eta,\cdot),\eta<0$}\psfrag{R}{(R)}
\psfrag{K1}{(K$_1$)}\psfrag{K2}{(K$_2$)}\psfrag{K3}{(K$_3$)}\psfrag{p}{$\overline{p}$}\psfrag{q}{$\overline{q}$}
\centering
\includegraphics[width=12cm]{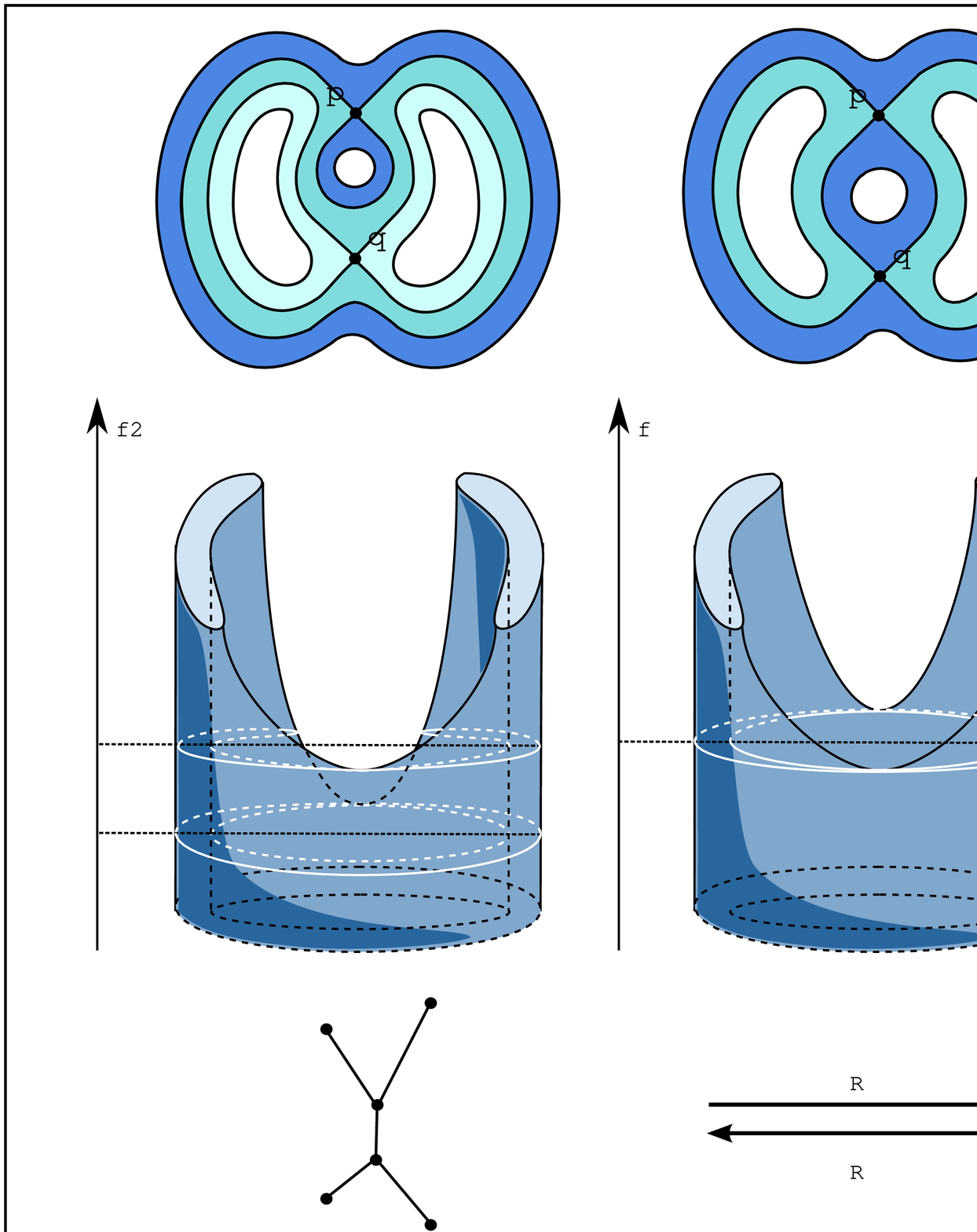}
\caption{\footnotesize{Center: A function
$\overline{h}\in\F^1_{\beta}(\M)$ as in case (2) with
$\overline{p},\overline{q}$ as in Figure \ref{Saddles}~$(c)$;
left-right: The universal deformation $G(\eta,\cdot)$ with the
associated  labeled Reeb graphs for $\eta<0$ and
$\eta>0$.}\label{R1}}
\psfrag{f}{$\overline{h}=G(0,\cdot)$}\psfrag{f1}{$G(\eta,\cdot),\eta>0$}\psfrag{f2}{$G(\eta,\cdot),\eta<0$}\psfrag{R}{(R)}
\psfrag{K1}{(K$_1$)}\psfrag{K2}{(K$_2$)}\psfrag{K3}{(K$_3$)}\psfrag{p}{$\overline{p}$}\psfrag{q}{$\overline{q}$}
\centering
\includegraphics[width=12cm]{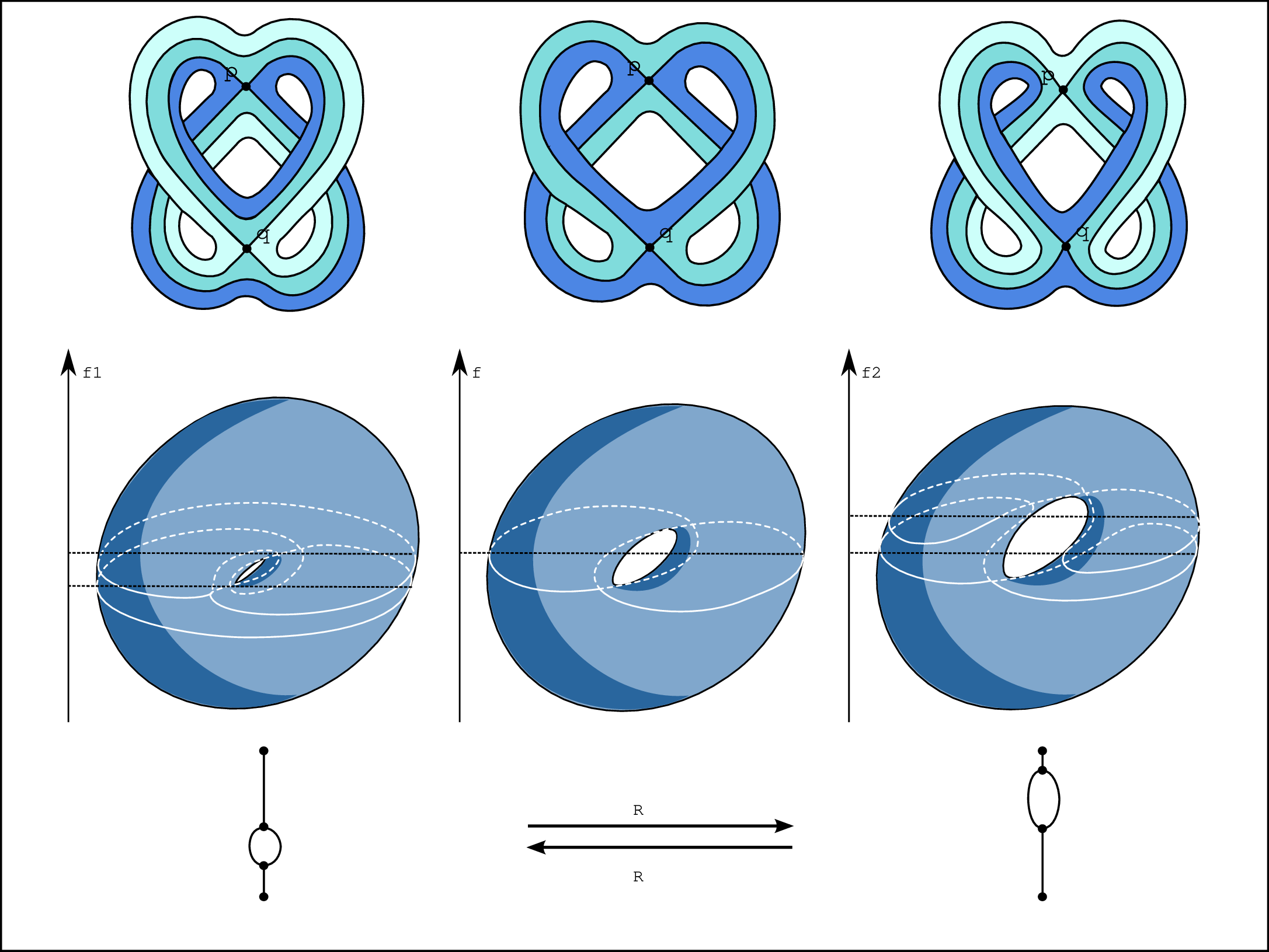}\\
\caption{\footnotesize{Center: A function
$\overline{h}\in\F^1_{\beta}(\M)$ as in case (2) with
$\overline{p},\overline{q}$ as in Figure \ref{Saddles}~$(d)$;
left-right: The universal deformation $G(\eta,\cdot)$ with the
associated  labeled Reeb graphs for $\eta<0$ and
$\eta>0$.}\label{R2}}
\end{figure}

\end{proof}

\begin{lem}\label{pathstability}
If $\h(\lambda)$ belongs to $\F^0(\M)$ for every $\lambda \in
[0,1]$ apart from one value $0<\overline{\lambda}<1$ at which $\h$
transversely intersects $\F^1(\M)$, then
$d_E((\Gamma_{f},\lf),(\Gamma_{g},\Lg))\le {\|f-g\|}_{C^0}.$
\end{lem}

\begin{proof}
Let  $\overline{h}=h(\overline{\lambda})$. By Lemma
\ref{pointonF1},  for every real number $\delta>0$,we can find two
values $0<\lambda'<\overline{\lambda}<\lambda ''<1$ such that
$d_E((\Gamma_{h(\lambda')},\ell_{h(\lambda')}),(\Gamma_{h(\lambda'')},\ell_{h(\lambda'')}))\le
\delta.$

 Applying the triangle inequality, we have:
{\setlength\arraycolsep{2pt}\begin{eqnarray*}
d_E((\Gamma_{f},\lf),(\Gamma_{g},\Lg))&\le&d_E((\Gamma_{f},\lf),(\Gamma_{h(\lambda')},\ell_{h(\lambda')}))+d_E((\Gamma_{h(\lambda')},\ell_{h(\lambda')}),
(\Gamma_{h(\lambda'')},\ell_{h(\lambda'')}))\\&&+d_E((\Gamma_{h(\lambda'')},\ell_{h(\lambda'')}),(\Gamma_{g},\Lg)).
\end{eqnarray*}}
Moreover, we get
$$d_E((\Gamma_{f},\lf),(\Gamma_{h(\lambda')},\ell_{h(\lambda')}))\le
\|f-h(\lambda')\|_{C^0}=\lambda'\cdot \|f-g\|_{C^0},$$ and
$$d_E((\Gamma_{h(\lambda'')},\ell_{h(\lambda'')}),(\Gamma_{g},\Lg))\le
\|h(\lambda'')-g\|_{C^0}=(1-\lambda'')\cdot \|f-g\|_{C^0},$$ where
the inequalities follow from Lemma \ref{pathstabilityF0}, and
equalities from Lemma \ref{pathH} with $f=h(0),g=h(1)$. Hence,
$$d_E((\Gamma_{f},\lf),(\Gamma_{g},\Lg))\le(1+\lambda'-\lambda'')\cdot\|f-g\|_{C^0}+\delta.$$
In conclusion, given that $0<\lambda'<\lambda''$, the inequality
$d_E((\Gamma_{f},\lf),(\Gamma_{g},\Lg))\le \|f-g\|_{C^0}+\delta$
holds. This yields the claim by the arbitrariness of $\delta$.
\end{proof}

We are now ready to prove the stability Theorem \ref{global}.

\begin{proof}[Proof of Theorem \ref{global}]
Recall from \cite{Hi76} that $\F^0(\M)$ is open in $\F(\M)$
endowed with the $C^2$ topology. Thus, for every sufficiently
small real number $\delta>0$, the neighborhoods $N(f,\delta)$ and
$N(g,\delta)$ are contained in $\F^0(\M)$. Take $\widehat{f}\in
N(f,\delta)$ and $\widehat{g} \in N(g,\delta)$ such that the path
$h:[0,1]\to \F(\M)$, with
$h(\lambda)=(1-\lambda)\widehat{f}+\lambda\widehat{g}$, belongs to
$\F^0(\M)$ for every $\lambda \in [0,1]$, except for at most a
finite number $n$ of values, $\mu_1,\mu_2,\dots,\mu_n$, at which
$h$ transversely intersects $\F^1(\M)$. We begin by proving our
statement for $\widehat{f}$ and $\widehat{g}$, and then show its
validity for $f$ and $g$. We proceed by induction on $n$. If $n=0$
or $n=1$, the inequality
$d_E((\Gamma_{\widehat{f}},\ell_{\widehat{f}}),(\Gamma_{\widehat{g}},\ell_{\widehat{g}}))\le{\|\widehat{f}-\widehat{g}\|}_{C^0}$
holds because of Lemma~\ref{pathstabilityF0}
or~\ref{pathstability}, respectively. Let us assume the claim is
true for $n\ge 1$, and prove it for $n+1$. Let
$0<\mu_1<\lambda_1<\mu_2<\lambda_2<\ldots<\mu_n<\lambda_n<\mu_{n+1}<1$,
with $h(0)=\widehat{f}$, $h(1)=\widehat{g}$,
$h(\mu_i)\in\F^1(\M)$, for every $i=1,\ldots,n+1$, and
$h(\lambda_j)\in\F^0(\M)$, for every $j=1,\ldots,n$. We consider
$h$ as the concatenation of the paths $h^1,h^2:[0,1]\to \F(\M)$,
defined, respectively, as
$h^1(\lambda)=(1-\lambda)\widehat{f}+\lambda h(\lambda_n)$, and
$h^2(\lambda)=(1-\lambda)h(\lambda_n)+\lambda \widehat{g}$. The
path $h^1$ transversally intersect $\F^1(\M)$ at $n$ values
$\mu_1,\ldots,\mu_n$. Hence, by the inductive hypothesis, we have
$d_E((\Gamma_{\widehat{f}},\ell_{\widehat{f}}),(\Gamma_{h(\lambda_n)},\ell_{h(\lambda_n)}))\le
{\|\widehat{f}-h(\lambda_n)\|}_{C^0}$. Moreover, the path $h^2$
transversally intersect $\F^1(\M)$ only at the value $\mu_{n+1}$.
Consequently, by Lemma~\ref{pathstability}, we have
$d_E((\Gamma_{h(\lambda_n)},\ell_{h(\lambda_n)}),(\Gamma_{\widehat{g}},\ell_{\widehat{g}}))\le
{\|h(\lambda_n)-\widehat{g}\|}_{C^0}$. Using the triangle
inequality and Lemma~\ref{pathH}, we can conclude that:
{\setlength\arraycolsep{2pt}\begin{eqnarray}\label{triangular}
d_E((\Gamma_{\widehat{f}},\ell_{\widehat{f}}),(\Gamma_{\widehat{g}},\ell_{\widehat{g}}))&\le&d_E((\Gamma_{\widehat{f}},\ell_{\widehat{f}}),(\Gamma_{h(\lambda_n)},\ell_{h(\lambda_n)}))+
d_E((\Gamma_{h(\lambda_n)},\ell_{h(\lambda_n)}),(\Gamma_{\widehat{g}},\ell_{\widehat{g}}))\nonumber\\
&\le
&\lambda_n{\|\widehat{f}-\widehat{g}\|}_{C^0}+(1-\lambda_n){\|\widehat{f}-\widehat{g}\|}_{C^0}={\|\widehat{f}-\widehat{g}\|}_{C^0}.
\end{eqnarray}}

Let us now estimate $d_E((\Gamma_{f},\lf),(\Gamma_{g},\Lg))$. By
the triangle inequality, we have:
$$d_E((\Gamma_{f},\lf),(\Gamma_{g},\Lg))\le
d_E((\Gamma_{f},\lf),(\Gamma_{\widehat{f}},\ell_{\widehat{f}}))+d_E((\Gamma_{\widehat{f}},\ell_{\widehat{f}}),(\Gamma_{\widehat{g}},\ell_{\widehat{g}}))+d_E((\Gamma_{\widehat{g}},\ell_{\widehat{g}}),(\Gamma_{g},\Lg)).$$
Since $\widehat{f}\in N(f,\delta)\subset \F^0(\M)$ and
$\widehat{g} \in N(g,\delta)\subset \F^0(\M)$, the following facts
hold: $(a)$ for every $\lambda\in[0,1]$, $(1-\lambda)f+\lambda
\widehat{f},(1-\lambda)g+\lambda \widehat{g}\in\F^0(\M)$; $(b)$
${\|f-\widehat{f}\|}_{C^0}\le\delta$ and
${\|\widehat{g}-g\|}_{C^0}\le\delta$. Hence, from $(a)$ and
Lemma~\ref{pathstabilityF0}, we get
$d_E((\Gamma_{f},\lf),(\Gamma_{\widehat{f}},\ell_{\widehat{f}}))\le{\|f-\widehat{f}\|}_{C^0}$,
and
$d_E((\Gamma_{g},\Lg),(\Gamma_{\widehat{g}},\ell_{\widehat{g}}))\le{\|\widehat{g}-g\|}_{C^0}$.
Using inequality (\ref{triangular}) and the triangle inequality of
$\|\cdot\|_{C^0}$, we deduce that
{\setlength\arraycolsep{2pt}\begin{eqnarray*}
d_E((\Gamma_{f},\lf),(\Gamma_{g},\Lg))&\le &
{\|f-\widehat{f}\|}_{C^0}+\|\widehat{f}-\widehat{g}\|_{C^0}+{\|\widehat{g}-g\|}_{C^0}\\
&\le &
\|f-g\|_{C^0}+2({\|f-\widehat{f}\|}_{C^0}+{\|\widehat{g}-g\|}_{C^0}).
\end{eqnarray*}}
Hence, from $(b)$, we have
$d_E((\Gamma_{f},\lf),(\Gamma_{g},\Lg))\le\|f-g\|_{C^0}+4\delta$.
This yields the conclusion by the arbitrariness of $\delta$.
\end{proof}

\section{Relationships with other stable metrics}\label{lowbound}

In this section, we consider relationships between the edit
distance and other  metrics for shape comparison: the natural pseudo-distance
 between functions \cite{DoFr04}, the functional distortion distance between Reeb graphs \cite{BaGeWa},
and the bottleneck distance between persistence diagrams
\cite{CoEdHa07}. More precisely, the main result we are going to show states that
the natural pseudo-distance between two simple Morse functions $f$
and $g$ and the edit distance between the corresponding Reeb
graphs actually coincide (Theorem \ref{equal}). As a consequence, we deduce that the edit distance is a metric
(Corollary \ref{defpos}), and that it is more discriminative than
the bottleneck distance between persistence
diagrams  (Corollary \ref{bottle}) and the functional distortion distance between Reeb graphs (Corollary \ref{distortion}). \\

The natural pseudo-distance is a dissimilarity measure between any
two  functions defined on the same compact manifold obtained by
minimizing the difference in the functions via a
re-parameterization of the manifold \cite{DoFr04}. In general, the natural pseudo-distance is only a pseudo-metric.
However it turns out to be a metric in some particular cases such
as the case of simple Morse functions on a smooth closed connected
surface, considered up to $R$-equivalence, as proved in
\cite{CaDiLa13}. We give the definition in this context.

\begin{defi}\label{pseudod}
The \emph{natural pseudo-distance} between $R$-equivalence classes
of simple Morse functions $f,g$ on the same surface $\M$ is
defined as
$$
d_N([f],[g])= \underset{\xi \in {\mathcal D}(\M)}\inf \|f-g\circ
\xi\|_{C^0},
$$
where ${\mathcal D}(\M)$ is the set of self-diffeomorphisms on
$\M$.
\end{defi}

In order to study $d_N$, it is often useful to consider the
following  fact.

\begin{prop}\label{homeo}
Letting  ${\mathcal H}(\M)$ be the set of self-homeomorphisms on
$\M$, it holds that $d_N([f],[g])= \underset{\xi \in {\mathcal H}(\M)}\inf
\|f-g\circ \xi\|_{C^0}$.
\end{prop}

\begin{proof}
Let $d=d_N([f],[g])$. Clearly $\underset{\xi \in {\mathcal
H}(\M)}\inf \|f-g\circ \xi\|_{C^0}\le d$. By contradiction, assuming that
$\underset{\xi \in {\mathcal H}(\M)}\inf \|f-g\circ \xi\|_{C^0}<
d$, there  exists a homeomorphism $\overline{\xi}$ such that
$\|f-g\circ \overline{\xi}\|_{C^0}<d$. On the other hand, by
\cite[Cor. 1.18]{Wh61}, for every metric $\delta$ on $\M$ and for
every $n\in\N$, there exists a diffeomorphism $\xi_n:\M\to\M$ such
that $\delta(\xi_n(p),\overline{\xi}(p))<1/n$, for every $p\in
\M$. Hence, by the continuity of $g$, and applying the reverse
triangle inequality, we deduce that
$$\lim_{n\to\infty}\left|\|f-g\circ \overline{\xi}\|_{C^0}-\|f-g\circ \xi_n\|_{C^0}\right|\le \lim_{n\to\infty} \|g\circ\xi_n-g\circ \overline{\xi}\|_{C^0}=0.$$
Therefore, for $n$ sufficiently large, there  exists a
diffeomorphism $\xi_n$ such that $\|f-g\circ \xi_n\|_{C^0}< d$,
yielding a contradiction.
\end{proof}

The following Lemmas \ref{c(T)=R}-\ref{c(T)=K} state that the cost
of each elementary deformation upper-bounds the natural
pseudo-distance. Their proofs deploy the concepts of elementary
cobordism and rearrangement, whose detailed treatment  can be
found in \cite{Mi65}.

\begin{lem}\label{c(T)=R}
For every elementary deformation $T\in {\mathcal
T}((\Gamma_f,\lf),(\Gamma_g,\Lg))$ of type (R), $c(T)\ge d_N(
[f],[g])$.
\end{lem}

\begin{proof}
Since $T$ is of type (R),  there exists an edge preserving
bijection $\Phi:V(\Gamma_f)\to V(\Gamma_g)$. Hence, $f$ and $g$
have the same number of critical points of the same type:
$K_f=\{p_1,\ldots,p_n\}$, $K_g=\{p'_1,\ldots,p'_n\}$,  with
$\Phi(p_i)=p'_i$, and $p_i,p'_i$ both being of index 0,1, or 2.

Let $c_i=f(p_i)$ and $c'_i=g(p_i')$ for $i=1,\ldots,n$. We shall
construct a homeomorphism $\xi:\M\to\M$ such that
$\xi_{|K_f}=\Phi$ and ${\|f-g\circ
\xi\|}_{C^0}=\underset{i=1,\ldots, n}\max|c_i-c_i'|=c(T)$. By
Proposition~\ref{homeo}, this will yield the claim.

Let us endow $\M$ with a Riemannian metric, and consider the
smooth vector field $X=-\frac{\nabla f}{\|\nabla f\|^2}$ on
$\M\setminus K_f$, and the smooth vector field $Y=\frac{\nabla
g}{\|\nabla g\|^2}$ on $\M\setminus K_g$. Let us denote by
$\varphi_t(p)$ and $\psi_t(p)$ the flow lines defined by $X$ and
$Y$, on $\M\setminus K_f$ and $\M\setminus K_g$, respectively.
We observe that $f$ strictly decreases along $X$-trajectories,
while $g$ strictly increases along $Y$-trajectories. Moreover, no
two $X$-trajectories (resp. $Y$-trajectories) pass through the
same $p$. Hence, $\varphi_t(p)$ and $\psi_t(p)$ are injective
functions of $t$ and $p$, separately. By \cite[Prop. 1.3]{Pade82},
$\varphi$ and $\psi$ are continuous in $t$ and $p$ when restricted
to  compact submanifolds of $\M\setminus K_f$ and $\M\setminus
K_g$, respectively.

Let us fix a real number $\eps>0$ sufficiently small so that, for
 $i=1,\ldots,n$, $f^{-1}([c_i-\eps,c_i+\eps])\cap
K_f=\{p_i\}$ and $g^{-1}([c'_i-\eps,c'_i+\eps])\cap K_g=\{p'_i\}$.

In order to construct the desired homeomorphism $\xi$ on $\M$, the
main idea is to cut $\M$ into cobordisms and define suitable
homeomorphisms on each of these cobordisms that can be glued
together to obtain $\xi$. The fact that $\xi$ is not required to
be differentiable but only continuous facilitates the gluing
process.

Let us consider the cobordisms obtained cutting $\M$ along the
level curves $f^{-1}(c_i\pm\eps)$ and $g^{-1}(c'_i\pm\eps)$ for
$i=1,\ldots,n$.  According to whether these cobordisms contain
points of maximum, minimum, saddle points, or no critical points
at all, we treat the cases differently.\\

\noindent{\bf Case 1:} Let $p_i,p_i'$  be points of  maximum  or
minimum of $f$ and $g$, respectively. Let $D=D_i$  (resp. $D'=D_i'$,
) be the connected component of $f^{-1}([c_i-\eps,c_i+\eps])$
(resp. $g^{-1}([c'_i-\eps,c'_i+\eps])$) that contains $p_i$
(resp. $p'_i$). $D$ and $D'$ are topolological disks. Let
$\sigma^D:\partial D\to\partial D'$ be a given homeomorphism
between the boundaries of $D$ and $D'$.

\noindent{\bf Claim 1.} There exists a homeomorphism $\xi^D:D\to
D'$ such that:
\begin{itemize}
\item[$(a_1)$] $\xi^D|_{\partial D}=\sigma^D$;
 \item[$(b_1)$] $\underset{p\in D}\max|f(p)-g\circ \xi^D(p)|=|c_i-c'_i|.$
\end{itemize}
\noindent{\em Proof of Claim 1.} We first prove Claim 1 for
maxima. We set $\xi^D(p_i)=p'_i$, and, for every $p\in
D\setminus\{p_i\}$, $\xi^D(p)=p'$, where
$p'=\psi_{f(p)-c_i+\eps}\circ\sigma^D\circ\varphi_{f(p)-c_i+\eps}(p)$.
In plain words, for each $p\in D$ we follow the $X$-flow downwards
until the intersection with $f^{-1}(c_i-\eps)$; then we apply the
homeomorphism $\sigma^D$ to go from $f^{-1}(c_i-\eps)$ to
$g^{-1}(c_i'-\eps)$; finally, we follow the $Y$-flow upwards.

The function $\xi^D$ is injective as can be seen using the
aforementioned injectivity property of $\varphi$ and $\psi$.
Moreover, $\xi^D$ is surjective because, for every $p\in
D\setminus \{p_i\}$, there exists a flow line passing for $p$.
Furthermore, $\xi^D$ is continuous on $D\setminus\{p_i\}$ because
composition of continuous functions. The continuity can be
extended to the whole $D$ as can be seen taking a sequence $(q_j)$
in $D \setminus\{p_i\}$ converging to $p_i$. Since $\lim_jf(q_j)
= c_i$, by construction of $\xi^D$ it holds that
$\lim_jg(\xi^D(q_j))=c_i'$. We see  that
$\lim_j\xi^D(q_j)=p_i'$ because $p_i'$ is the only point of $D'$
where $g$ takes value equal to $c_i'$,. Therefore $\xi^D$ is continuous on $D$.
Moreover, since $\xi^D$ is a continuous bijection from a compact
space to a Hausdorff space, it is a homeomorphism. Finally,
property $(a_1)$  holds by construction and property  $(b_1)$
holds because, for every $p\in D$, $g(\xi^D(p))=f(p)+c_i'-c_i$.

To prove Claim 1 when  $p_i,p'_i$ are minimum points of $f$ and
$g$, it is sufficient to  replace $\varphi_{f(p)-c_i+\eps}(p)$ and
$\psi_{f(p)-c_i+\eps}(p)$ by $\varphi_{f(p)-c_i-\eps}(p)$ and
$\psi_{f(p)-c_i-\eps}(p)$, respectively.\\

\noindent{\bf Case 2:} Let $p_i,p'_i$ be two splitting saddle
points or two joining saddle points of $f$ and $g$, respectively,
and let $P$ and $P'$ be the connected component of
$f^{-1}([c_i-\eps,c_i+\eps])$ and $g^{-1}([c'_i-\eps,c'_i+\eps])$,
respectively,  that contain $p_i$ and $p'_i$.  Let
$\sigma^P:\partial^- P\to\partial^-P'$ be a given homeomorphism
between the lower boundaries of $P$ and $P'$.

\noindent{\bf Claim 2.} There exists a homeomorphism $\xi^P:P\to
P'$ such that:
\begin{itemize}
\item[$(a_2)$] $\xi^P|_{\partial^{-} P}=\sigma^P$; \item[$(b_2)$]
$\underset{p\in P}\max|f(p)-g\circ \xi^P(p)|=|c_i-c'_i|.$
\end{itemize}

\noindent{\em Proof of Claim 2.} Let us consider the case
$p_i,p'_i$ are two splitting saddle points of $f$ and $g$
respectively, so that $P$ and $P'$ are two upside-down pairs of
pants. We let $p_a$, $p_b$ be the only two points of intersection
of $f^{-1}(c_i-\eps/2)$ with the trajectories  of the gradient
vector field $X$ coming from $p_i$. Analogously, we let $p_a'$,
$p_b'$ be the only two points of intersection  of
$g^{-1}(c'_i-\eps/2)$ with the trajectories of the gradient vector
field $Y$ leading to $p_i'$.

The pair of pants $P$ can be decomposed into $P=M\cup N\cup O$
with $M=\{p\in P: f(p)\in[c_i-\eps,c_i-\eps/2]\}$, $N=\{p\in P:
f(p)\in[c_i-\eps/2,c_i]\}$ and $O=\{p\in P:
f(p)\in[c_i,c_i+\eps]\}$. Analogously,  the pair of pants $P'$ can
be decomposed into $P'=M'\cup N'\cup O'$ with $M'=\{p'\in P':
g(p')\in[c_i'-\eps,c_i'-\eps/2]\}$, $N'=\{p'\in P':
g(p')\in[c'_i-\eps/2,c'_i]\}$, and $O'=\{p'\in P':
g(p')\in[c'_i,c'_i+\eps]\}$.

The construction of $\xi^P$ is based on gluing three
homeomorphisms $\xi^M:M\to M'$, $\xi^N:N\to N'$, $\xi^O:O\to O'$
together.

First, we observe that, $M$ and $M'$ being cylinders, it is
possible to construct a homeomorphism $\xi^M$ that extends
$\sigma^P$ to $M$ in such a way that $\xi^M(p_a)=p_a'$ and
$\xi^M(p_b)=p_b'$, also sending the level-sets of $f$ into those
of $g$. In this way $\underset{p\in M}\max|f(p)-g\circ
\xi^M(p)|=|c_i-c'_i|.$

Next, we define  $\xi^N$   by setting $\xi^N(p_i)=p'_i$, and, for
every $p\ne p_i$,  $\xi^N(p)=p'$, where
$p'=\psi_{f(p)-c_i+\eps/2}\circ\xi^M\circ\varphi_{f(p)-c_i+\eps/2}(p)$.
It agrees with $\xi^M$ on $\partial M\cap \partial N$ and
$\underset{p\in N}\max|f(p)-g\circ \xi^N(p)|=|c_i-c'_i|$.
Moreover,  $\xi^N$ is bijective and continuous  on $N\setminus \{
p_i\}$ by arguments similar to those used in the proof of Claim 1.
To see that continuity  extends to $p_i$, let $(q_j)$ be a
sequence converging to $ p_i$. The sequence
$(\varphi_{f(p)-c_i+\eps/2}(q_j))$  has at most two accumulating
points, precisely   the points $p_a$, and $p_b$. By the
construction of $\xi^M$, $\xi^M(p_a)=p_a'$ and $\xi^M(p_b)=p_b'$,
hence the sequence $(\xi^N(q_j))$ converges to $p_i'$. In
conclusion, $\xi^N$ is bijective and continuous, therefore it is a
homeomorphism.

Finally, we construct $\xi^O$  by using again the trajectories of
$X$ and $Y$: for each $p\in O$ we follow the flow of $X$ downwards
until the intersection with $f^{-1}(c_i)$. If the intersection
point $q$ is different from  $p_i$, we set $\xi^O(p)$ equal to the
point $p'$ on the trajectory of $ \xi^N(q)$ such that
$p'=\psi_{f(p)-c_i}(\xi^N(q))$. Otherwise, if $q=p_i$, we consider
a sequence $(r_j)$ of points in the same connected component of
$O\setminus \{p_i\}$ as $p$ and converging to $p$. The
intersection of $f^{-1}(c_i)$ with the downward flow through  $r_j$, $j\in\N$, gives a sequence $(q_j)$ converging to $p_i$ and
belonging to one and the same component of  $f^{-1}(c_i)\setminus
\{p_i\}$ as $p$. By the continuity of $\xi^N$ the sequence
$(\xi^N(q_j))$ converges to $p_i'$ and its points belong to one
and the same component of $g^{-1}(c_i')\setminus\{p_i'\}$. Hence
the sequence $(\psi_{f(r_j)-c_i}(\xi^N(q_j)))$ converges to a
point $p'$. We set $\xi^O(p)=p'$. By the continuity of $\p$ and
$\psi$, this definition does not depend on the choice of the
sequence $(r_j)$. By construction, $\xi^O$ is continuous and the proof that it is a
homeomorphism can be handled by arguments similar to those used
for $\xi^N$. Moreover, it agrees with $\xi^N$ on $\partial N\cap
\partial O$ and $\underset{p\in N}\max|f(p)-g\circ
\xi^N(p)|=|c_i-c'_i|.$

In conclusion, $\xi^P$ can be constructed by gluing the
homeomorphisms $\xi^M$, $\xi^N$, $\xi^O$ together and the
properties $(a_2)$ and $(b_2)$ hold by construction.

The case when $p_i,p'_i$ are two joining saddle points of $f$ and $g$,
respectively, can be treated analogously. We have only to take
into account that $P$ is now a pairs of pants, and hence
 $M=\{p\in P: f(p)\in[c_i-\eps,c_i-\eps/2]\}$ is a pair of cylinders each containing one point of intersection between  $f^{-1}(c_i-\eps/2)$ and the trajectories  of the gradient
vector field $X$ coming from $p_i$.  Similarly for $P'$.\\

%
%
%

\noindent{\bf Case 3:} Let $p_i,p_j$ (resp. $p_i',p_j'$) be
critical points connected by an edge in the Reeb graph of $f$
(resp. $g$), and assume $c_i<c_j$ (resp. $c'_i<c'_j$). Let
$C=\{p\in \M: [p]\in e(p_i,p_j),\ c_i+\eps\le f(p)\le c_j-\eps\}$
and $C'=\{p\in\M: [p]\in e(p'_i,p'_j),\ c'_i+\eps\le g(p)\le
c'_j-\eps\}$. $C$ and $C'$ are two topological cylinders. Let
$\sigma^C:\partial^- C\to\partial^-C'$
  be a given homeomorphism between the lower boundaries of $C$ and
  $C'$.

\noindent{\bf Claim 3.} There exists a homeomorphism $\xi^C:C\to
C'$ such that:
\begin{itemize}
\item[$(a_3)$] $\xi^C|_{\partial^{-} C}=\sigma^C$; \item[$(b_3)$]
$\underset{p\in C}\max|f(p)-g\circ
\xi^C(p)|=\max\{|c_i-c'_i|,|c_j-c'_j|\}.$
\end{itemize}
\noindent{\em Proof of Claim 3.} To prove Claim 3, for every $p\in
C$, we set $\lambda_p$ equal to the only value in $[0,1]$ for
which $f(p)=(1-\lambda_p)(c_i+\eps)+\lambda_p(c_j-\eps)$, and
define $\xi^C(p)=p'$, with
$p'=\psi_{\lambda_p(c'_j-c'_i-2\eps)}\circ\sigma^C\circ\varphi_{\lambda_p(c_j-c_i-2\eps)}(p)$.

By the same arguments as used to prove the previous Claims 1 and
2, $\xi^C$ is a homeomorphism. It satisfies property $(a_3)$ by
construction. To prove  $(b_3)$, it is sufficient to observe that,
for every $p\in C$, {\setlength\arraycolsep{2pt}\begin{eqnarray*}
|f(p)-g(\xi^C(p))|&=&|(1-\lambda_p)(c_i+\eps)+\lambda_p(c_j-\eps)-(c'_i+\eps+\lambda_p(c'_j-c'_i-2\eps))|\\
&=&|(1-\lambda_p)(c_i-c'_i)+\lambda_p(c_j-c'_j)|.
\end{eqnarray*}}

Let us now construct the desired homeomorphism $\xi:\M\to\M$. Let
$\{p_1,\ldots,p_s\}\subseteq K_f$, $s\le n$, be the set of
critical points of $f$ of index 0 or 2, and, for $i = 1,\ldots,
s$, let $D_i, D_i'$  be as in Claim 1.

The spaces $W = \overline{\M \setminus \bigcup_{i=1}^s D_i}$ and $W'=\overline{\M
\setminus \bigcup_{i=1}^s D_i'}$ can be decomposed into the union
of cobordisms containing either no critical points or exactly one
critical point of index 1.  By Claims 2 and 3, it is possible to
extend a given homeomorphism  $\sigma^W:
\partial^-W\to\partial^-W'$ defined between  the lower boundaries of $W$ and $W'$ to a homeomorphism $\xi^W:W\to W'$
by gluing all the homeomorphisms on cobordisms along their
boundary components in the direction of the increasing of the
functions $f$ and $g$. Next, by Claim 1, we can glue this
homeomorphism $\xi^W$ along each boundary component of $W$ to a
homeomorphism $\xi^{D_i}: D_i\to D_i'$, for $i=1,\ldots, s$. As a
result, we get the desired self-homeomorphism $\xi$ of $\M$ such
that $\underset{p\in \M}\max|f(p)-g\circ
\xi(p)|=\underset{i=1,\ldots ,n}\max|c_i-c'_i|$.
\end{proof}

\begin{lem}\label{c(T)=BD}
For every elementary deformation $T\in {\mathcal
T}((\Gamma_f,\lf),(\Gamma_g,\Lg))$ of type {\em (B)} or {\em (D)},
$ c(T)\ge d_N([f],[g]) $.
\end{lem}

\begin{proof}
We prove the assertion only for the case when $T$ is of type (D),
because the other case will then follow from $c(T^{-1})=c(T)$ and
the symmetry property of $d_N$.

By definition of elementary deformation of type (D), $T$
transforms $(\Gamma_f,\lf)$ into a labeled Reeb graph that differs
from $(\Gamma_f,\lf)$  in that two vertices, say $p_1,p_2\in K_f$,
have been deleted together with their connecting edges. Otherwise
vertices, adjacencies and labels are the same. Assuming
$f(p_1)=c_1, f(p_2)=c_2$, with $c_1<c_2$, we have
$c(T)=\frac{c_2-c_1}{2}$. We recall that $f^{-1}([c_1,c_2])\cap
K_f=\{p_1,p_2\}$. By Proposition \ref{DversusDR}, there exists a
deformation $S=(S_0,S_1)\in {\mathcal
T}((\Gamma_f,\lf),(\Gamma_g,\Lg))$, $S_0$ being of type (R), $S_1$
of type (D), such that $c(S)= c(T)$. In particular, as shown in
the proof of the same proposition (formulas (\ref{c(S0)}) and
(\ref{c(S1)})), for every $\eps>0$ sufficiently small, $S_0$ and
$S_1$ can be built so that $c(S_0)=\frac{c_2-c_1}{2}-\eps$ and
$c(S_1)=\eps$.

For any $h_\eps$ for which $S_0(\Gamma_f,\lf)\cong
(\Gamma_{h_\eps}, \ell_{h_\eps})$, by Lemma \ref{c(T)=R} we have
$d_N([f],[h_\eps])\le c(S_0)=\frac{c_2-c_1}{2}-\eps$. Thus,
$$d_N([f],[g])\le d_N([f],[h_\eps])+d_N([h_\eps],[g])\le
\frac{c_2-c_1}{2}-\eps+ d_N([h_\eps],[g]).$$ Therefore, proving
that $d_N([h_\eps],[g])\le 4\eps$ will yield the claim,  by the
arbitrariness of $\eps>0$.

Let $W_\eps$ be the connected component of
$h_\eps^{-1}([\frac{c_1+c_2}{2}-2\eps,\frac{c_1+c_2}{2}+2\eps])$
containing $p_1,p_2$, and let us assume that $\eps$ is so small
that
$h_\eps^{-1}([\frac{c_1+c_2}{2}-2\eps,\frac{c_1+c_2}{2}+2\eps])$
does not contain  other critical points of $h_\eps$. By the
Cancellation Theorem in \cite[Sect. 5]{Mi65},  it is possible to
define a new simple Morse function $h'_\eps:\M\to\R$ which
coincides with $h_\eps$ on $\M\setminus W_\eps$, and has no
critical points in $W_\eps$. In particular,
$(\Gamma_{h'_\eps},\ell_{h'_\eps})\cong(\Gamma_g,\Lg)$, implying
that $h'_\eps$ and $g$ are $R$-equivalent. It necessarily holds that
$$d_N([h_\eps],[h'_\eps])\le \max_{p\in \M}|h_\eps(p)-h'_\eps(p)|=\max_{p\in W_\eps}|h_\eps(p)-h'_\eps(p)|\le 4\eps.$$
Moreover, by the $R$-equivalence of $h'_\eps$ and $g$, we have
$d_N([h'_\eps],[g])=0$, so that  $d_N([h_\eps],[g])\le 4\eps$ by
the triangle inequality property of $d_N$.
\end{proof}

\begin{lem}\label{c(T)=K}
For every elementary deformation  $T\in{\mathcal
T}((\Gamma_f,\lf),(\Gamma_g,\Lg))$ of type (K$_i$), $i=1,2,3$,
$c(T)\ge d_N([f],[g])$.
\end{lem}

\begin{proof}
For an elementary deformation $T$  of type (K$_i$), $i=1,2,3$,
 the sets $K_f$ and $K_g$ have the same cardinality, and all
but at most two of the critical values of $f$ and $g$ coincide.
Let $K_f=\{p_1,\ldots,p_n\}$ and $K_g=\{p'_1,\ldots,p'_n\}$, with
$f(p_k)=c_k$, $g(p'_k)=c'_k$ for every $k=1,\dots,n$. Assuming that
 the points $p_1,p_2$ correspond to the vertices
$u_1,u_2$ of $\Gamma_f$ shown in Table \ref{deformations}, rows
3-4, it holds that $c_1<c_2$, $c'_1>c'_2$, and $c_k=c'_k$ for
$k=3,\ldots ,n$. Moreover,  $K_f\cap
f^{-1}([c_1,c_2])=\{p_1,p_2\}$ and $K_g\cap
g^{-1}([c'_2,c'_1])=\{p'_1,p'_2\}$. Since $f,g\in\F^0(\M)$, there
exist $a,b\in\R$, with $a<b$, such that $c_1,c_2$ and $c'_1,c'_2$
are the sole critical values of $f$ and $g$, respectively, that
belong to the interval $[a,b]$. Let us denote by $W$ the connected
component of $f^{-1}([a,b])$ containing $p_1,p_2$. Under our
assumptions, we can apply the Preliminary Rearrangement Theorem
\cite[Thm 4.1]{Mi65}, and deduce that, for some choice of a
gradient-like vector field $X$ for $f$, there exists a Morse
function $h:W\to\R$ that has the same gradient-like vector field
as $f$, coincides with $f_{|W}$ near $\partial W$ and is equal to
$f$ plus a constant in some neighborhood of $p_1$ and in some
neighborhood of $p_2$. Moreover, $K_h=K_{f_{|W}}$, $h(p_1)=c'_1$,
$h(p_2)=c'_2$. We can extend $h$ to the whole surface by defining
$$
\widehat{h}(p)=\left\{
\begin{array}{ll}
f(p),&\mbox{if}\,p\in\M\setminus W,\\ h&\mbox{if}\,p\in W.
\end{array}
\right.
$$
Hence, $\widehat{h}\in\F^0(\M)$ and
$(\Gamma_{\widehat{h}},\ell_{\widehat{h}})\cong T(\Gamma_f,\lf)$, implying that $\widehat{h}$ is $R$-equivalent to $g$.
Therefore, by Definition \ref{equivalence},
$d_N([f],[g])=d_N([f],[\widehat{h}])$.

Let us prove that $d_N([f],[\widehat{h}])\le c(T)$. We observe
that, by the definitions of $d_N$ and $\widehat{h}$, we get:
{\setlength\arraycolsep{2pt}\begin{eqnarray}\label{inequalityDh}
d_N([f],[\widehat{h}])&\le &
\|f-\widehat{h}\|_{C^0}=\underset{p\in\M}\max|f(p)-\widehat{h}(p)|=\underset{p\in
W}\max|f(p)-h(p)|.
\end{eqnarray}}
To estimate  the value of $\underset{p\in W}\max|f(p)-h(p)|$, we
review the construction of the function $h$, as given in
\cite{Mi65}. Let $\mu:W\to [a,b]$ be a smooth function that is
constant on each trajectory of $X$, zero near the set of points on
trajectories going to or from $p_1$, and one near the set of
points on trajectories going to or from $p_2$. Then the function
$h$ can be defined as $h(p)=G(f(p),\mu(p))$, where $G:[a,b]\times
[0,1]\to [a,b]$ is a smooth function defined as $G(x,t)=(1-t)\cdot
G(x,0)+t\cdot G(x,1)$, with the following properties (see also
Figure \ref{c(K)>D1}):
\begin{itemize}
\item $\frac{\partial G}{\partial x}(x,0)=1$ for $x$ in a
neighborhood of $c_1$ (in particular $G(x,0)=x+c'_1-c_1$  for $x$
in a neighborhood of $c_1$),\\ $\frac{\partial G}{\partial
x}(x,1)=1$ for $x$ in a neighborhood of $c_2$ (in particular
$G(x,1)=x+c'_2-c_2$ for $x$ in a neighborhood of $c_2$); \item For
all $x$ and $t$, $G(x,t)$ monotonically increases from $a$ to $b$
as $x$ increases from $a$ to $b$;  \item $G(x,t)=x$ for $x$ near
to $a$ or $b$ and for every $t\in [0,1]$.
\end{itemize}
\begin{figure}[htbp]
\begin{center}
\psfrag{x}{$x$}\psfrag{y}{$z$}
\psfrag{z=G(x,0)}{$z=G(x,0)$}\psfrag{z=G(x,1)}{$z=G(x,1)$}
\psfrag{aa}{$(a,a)$} \psfrag{bb}{$(b,b)$} \psfrag{ci}{$c_1$}
\psfrag{cj}{$c_2$}\psfrag{ci'}{$c'_1$}
\psfrag{cj'}{$c'_2$}\psfrag{c1-c1'}{$c'_1-c_1$}\psfrag{c2-c2'}{$c_2-c'_2$}
\psfrag{a}{$a$} \psfrag{b}{$b$}\psfrag{c}{$\min
f$}\psfrag{d}{$\max f$}\psfrag{S}{$(\Gamma_g,\Lg)$}
\includegraphics[width=8cm]{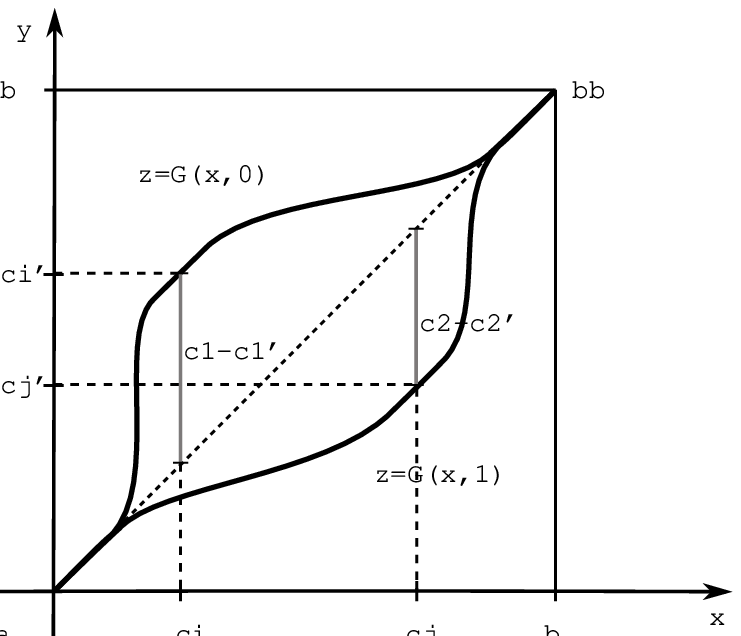}
\caption{\footnotesize{The function $G$ introduced in \cite{Mi65}
and used in the proof of Lemma \ref{c(T)=K}.}}\label{c(K)>D1}
\end{center}
\end{figure}

By the construction of $h$ and the
inequality (\ref{inequalityDh}), we have:
{\setlength\arraycolsep{2pt}\begin{eqnarray*}
d_N([f],[\widehat{h}])&\le &\underset{p\in
W}\max|f(p)-G(f(p),\mu(p))|=\max\{|f(p)-G(f(p),0)|,|f(p)-G(f(p),1)|\}\\&=&\max\{|c_1-c'_1|,|c_2-c'_2|\}=c(T).
\end{eqnarray*}}
\end{proof}

\begin{theorem}\label{equal}
Let $f,g\in\F^0(\M)$, and $(\Gamma_{f},\lf)$, $(\Gamma_{g},\Lg)$
be the associated labeled Reeb graphs. Then
$d_E((\Gamma_{f},\lf),(\Gamma_{g},\Lg))= d_N([f],[g]).$
\end{theorem}

\begin{proof}
The inequality $d_E((\Gamma_{f},\lf),(\Gamma_{g},\Lg))\ge
d_N([f],[g])$ holds because, for every
deformation $T\in\mathcal{T}((\Gamma_{f},\lf),(\Gamma_{g},\Lg))$,
$c(T)\ge d_N([f],[g])$. To see this, let $T=(T_1, \ldots, T_n)$, and set $T_i
\cdots T_1(\Gamma_{f},\lf)\cong
(\Gamma_{f^{(i)}},\ell_{f^{(i)}})$, $f=f^{(0)}$, $g=f^{(n)}$. From
Lemmas \ref{c(T)=R}-\ref{c(T)=K} and the triangle inequality property of
$d_N$, we get {\setlength\arraycolsep{2pt}\begin{eqnarray*}
c(T)&=& \underset{i=1}{\overset{n}\sum}c(T_i)\ge
\underset{i=1}{\overset{n}\sum}d_N([f^{(i-1)}],[f^{(i)}])\ge
d_N([f],[g]).
\end{eqnarray*}}
Conversely, by Theorem \ref{global},
$d_E((\Gamma_{f},\lf),(\Gamma_{g\circ\xi},\ell_{g\circ\xi}))\le
{\|f-g\circ\xi\|}_{C^0}$, for every $\xi\in
{\mathcal D}(\M)$.
Therefore $d_E((\Gamma_{f},\lf),(\Gamma_{g},\Lg))\le
\underset{\xi\in{\mathcal
D}(\M)}\inf\|f-g\circ\xi\|_{C^0}=d_N([f],[g])$ because $d_E((\Gamma_{f},\lf),(\Gamma_{g},\Lg))=d_E((\Gamma_{f},\lf),(\Gamma_{g\circ\xi},\ell_{g\circ\xi}))$.
\end{proof}

\begin{cor}\label{defpos}
For every $f,g\in\F^0(\M)$, the edit distance between the
associated labeled Reeb graphs is a metric on isomorphism classes
of labeled Reeb graphs.
\end{cor}
\begin{proof}
The claim is an immediate consequence of Theorem \ref{equal}
together with \cite[Thm. 4.2]{CaDiLa13}, which states that the
natural pseudo-distance is actually a metric on the space
$\F^0(\M)$.
\end{proof}

\begin{cor}\label{bottle}
For every $f ,g \in\F^0(\M)$,
$d_E((\Gamma_{f},\lf),(\Gamma_{g},\Lg))\ge d_B(D_f,D_g)$, where
$d_B$ denotes the bottleneck distance between the persistence
diagrams $D_f$ and $D_g$ of $f$ and $g$. In some cases this
inequality is strict.
\end{cor}

\begin{proof}
The inequality $d_E((\Gamma_{f},\lf),(\Gamma_{g},\Lg))\ge
d_B(D_f,D_g)$ holds because of Theorem \ref{equal} and  the fact that the bottleneck distance is a lower bound for
the natural pseudo-distance (cf. \cite{CeDi*13}).

As for the second statement, an example showing that the edit
distance between the labeled Reeb graphs of two functions
$f,g\in\F^0(\M)$ can be strictly greater than the bottleneck
distance between the persistence diagrams of $f$ and $g$ is
displayed in Figure \ref{esempio_tri}.
\begin{figure}[htbp]
\begin{center}
\psfrag{0}{$0$}\psfrag{1}{$1$} \psfrag{2}{$2$}\psfrag{3}{$3$}
\psfrag{4}{$4$} \psfrag{5}{$5$} \psfrag{6}{$6$}
\psfrag{7}{$7$}\psfrag{f}{$f$} \psfrag{g}{$g$} \psfrag{M}{$\max
z$} \psfrag{L}{$(\Gamma_f,\lf)$}\psfrag{S}{$(\Gamma_g,\Lg)$}
\includegraphics[width=12cm]{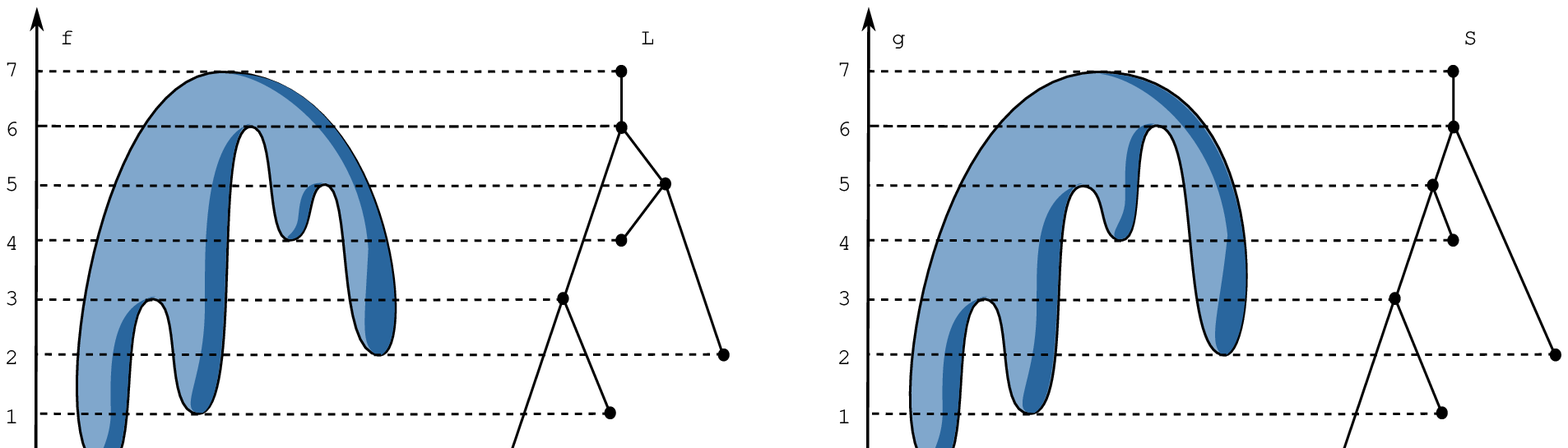}
\caption{\footnotesize{The example used in the proof of Corollary
\ref{bottle} to show that the edit distance between labeled
Reeb graphs can be more discriminative than the bottleneck
distance between persistence diagrams whenever the same functions
are considered.}}\label{esempio_tri}
\end{center}
\end{figure}
Indeed, $f$ and $g$ have the same persistence diagrams for any
homology degree implying that $d_B(D_f,D_g)=0$, whereas the labeled Reeb graphs are not
isomorphic, implying that $d_E((\Gamma_{f},\lf),(\Gamma_{g},\Lg))>0$.
\end{proof}


\begin{cor}\label{distortion}
For every $f ,g \in\F^0(\M)$,
$d_E((\Gamma_{f},\lf),(\Gamma_{g},\Lg))\ge d_{FD}(R_f,R_g)$, where
$d_{FD}$ denotes the functional distortion distance between the
spaces  $R_f=\M/\sim_f$ and $R_g=\M/\sim_g$. In some cases this
inequality is strict.
\end{cor}

\begin{proof}
The inequality $d_E((\Gamma_{f},\lf),(\Gamma_{g},\Lg))\ge
d_{FD}(R_f,R_g)$ is a consequence of the stability of Reeb graphs
with respect to $d_{FD}$ \cite[Thm. 4.1]{BaGeWa}, and can be seen
in the same way as the second inequality shown in the proof of
Theorem \ref{equal}.

As for the second statement, an example showing that, for two
functions $f,g\in\F^0(\M)$,
$d_E((\Gamma_{f},\lf),(\Gamma_{g},\Lg))$ can be strictly greater
than $d_{FD}(R_f,R_g)$ is displayed in Figure \ref{esempio_DF}.
\begin{figure}[htbp]
\begin{center}
\psfrag{b}{$b$}\psfrag{d}{$d$}
\psfrag{c1}{$c_1$}\psfrag{c2}{$c_2$}
\psfrag{c1+a}{$c_1+a$}\psfrag{c2+a}{$c_2+a$}\psfrag{f}{$g$}
\psfrag{g}{$f$} \psfrag{M}{$\max z$}
\psfrag{S}{$(\Gamma_f,\lf)$}\psfrag{L}{$(\Gamma_g,\Lg)$}
\includegraphics[width=12cm]{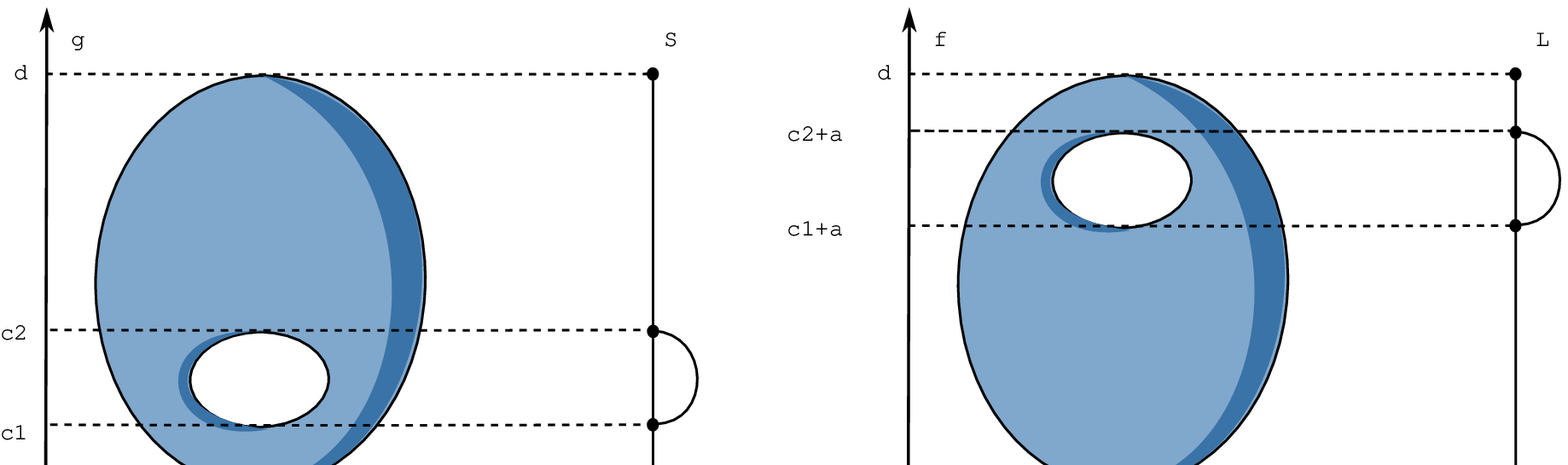}
\caption{\footnotesize{The example used in the proof of Corollary
\ref{distortion} to show that the edit distance between labeled
Reeb graphs can be more discriminative than the functional
distorsion distance between Reeb graphs whenever the same
functions are considered.}}\label{esempio_DF}
\end{center}
\end{figure}
In this case, $d_E((\Gamma_{f},\lf),(\Gamma_{g},\Lg))=a$, because
$a$ is both the cost of the deformation $T$ of type (R) that
changes the vertex label $c_i$ into $c_{i}+a$, $i=1,2$, and the value of the bottleneck
distance between the 1st homology degree (ordinary) persistence
diagrams of $f$ and $g$. On the other hand, $d_{FD}(R_f,R_g)\le
(c_2-c_1)/4$ as can be seen by considering any continuous map
$\Phi:R_f\to R_g$ that takes each point of $R_f$ to a point of
$R_g$ with the same function value, together with any continuous
map $\Psi:R_g\to R_f$ that takes each point of $R_g$ to a point of
$R_f$ with the same function value.
\end{proof}

\subsection*{Acknowledgments}
The authors wish to thank Professor V. V. Sharko for his
clarifications on the uniqueness property of Reeb graphs of
surfaces and for indicating the reference \cite{Ku98}. The
research described in this article has been partially supported by
GNSAGA-INdAM (Italy).

\bibliographystyle{amsplain}
\bibliography{biblio}

\end{document}